\documentclass[a4paper,UKenglish,cleveref, autoref]{lipics-v2019}

\usepackage{amsmath,tabu}
\usepackage{amssymb,amsfonts,textcomp}
\usepackage{amsthm}
\usepackage{stmaryrd}
\usepackage{booktabs}
\usepackage[ruled, linesnumbered, vlined]{algorithm2e} 

\title{Complex Event Forecasting with Prediction Suffix Trees: Extended Technical Report\footnote{ \textcolor{red}{This is the extended technical report for the paper \emph{Complex Event Forecasting with Prediction Suffix Trees} to be published at the VLBD Journal (VLDBJ). Please, use the VLDBJ version, when it becomes available, if you need to cite the paper.}}} 

\author{Elias Alevizos}{Department of Informatics, National and Kapodistrian University of Athens, Greece \and Institute of Informatics \& Telecommunications, National Center for Scientific Research ``Demokritos'', Greece}{ilalev@di.uoa.gr,alevizos.elias@iit.demokritos.gr}{https://orcid.org/0000-0002-9260-0024}{}

\author{Alexander Artikis}{Department of Maritime Studies, University of Piraeus, Greece \and Institute of Informatics \& Telecommunications, National Center for Scientific Research ``Demokritos'', Greece}{a.artikis@unipi.gr}{https://orcid.org/0000-0001-6899-4599}{}

\author{Georgios Paliouras}{Institute of Informatics \& Telecommunications, National Center for Scientific Research ``Demokritos'', Greece}{paliourg@iit.demokritos.gr}{https://orcid.org/0000-0001-9629-2367}{}

\authorrunning{Alevizos et al.}

\Copyright{Elias Alevizos, Alexander Artikis and Georgios Paliouras}



\category{}

\relatedversion{}

\supplement{}

\acknowledgements{This work was supported by the INFORE project, which has received funding from the European Union's Horizon 2020 research and innovation program, under grant agreement No 825070.}

\nolinenumbers 

\hideLIPIcs  

\def\true{\textsf{\footnotesize TRUE}}
\def\false{\textsf{\footnotesize FALSE}}
\def\sfa{$\mathit{SFA}$}
\def\sre{$\mathit{SRE}$}
\def\ssfa{$\mathit{sSFA}$}
\def\ssre{$\mathit{sSRE}$}
\def\dsfa{$\mathit{DSFA}$}
\def\dfa{$\mathit{DFA}$}
\def\psa{$\mathit{PSA}$}
\def\pst{$\mathit{PST}$}
\def\cst{$\mathit{CST}$}
\def\skipany{\textsf{\footnotesize skip-till-any-match}}
\def\where{\textsf{\footnotesize WHERE}}

\newtheorem*{proposition*}{Proposition}
\newtheorem*{theorem*}{Theorem}

\begin{document}

\maketitle

\begin{abstract}
Complex Event Recognition (CER) systems have become popular in the past two decades due to their ability to ``instantly'' detect patterns on real-time streams of events.
However, there is a lack of methods for forecasting when a pattern might occur before such an occurrence is actually detected by a CER engine.
We present a formal framework that attempts to address the issue of Complex Event Forecasting (CEF).
Our framework combines two formalisms:
a) symbolic automata which are used to encode complex event patterns; and
b) prediction suffix trees which can provide a succinct probabilistic description of an automaton's behavior.
We compare our proposed approach against state-of-the-art methods and show its advantage in terms of accuracy and efficiency. 
In particular, prediction suffix trees,
being variable-order Markov models,
have the ability to capture long-term dependencies in a stream by remembering only those past sequences that are informative enough. 
Our experimental results demonstrate the benefits,
in terms of accuracy,
of being able to capture such long-term dependencies.
This is achieved by increasing the order of our model beyond what is possible with full-order Markov models that need to perform an exhaustive enumeration of all possible past sequences of a given order.
We also discuss extensively how CEF solutions should be best evaluated on the quality of their forecasts.     
\end{abstract}

\section{Introduction}

The avalanche of streaming data in the last decade has sparked an interest in technologies processing high-velocity data streams.
Complex Event Recognition (CER) is one of these technologies which have enjoyed increased popularity
\cite{DBLP:journals/csur/CugolaM12,DBLP:journals/vldb/GiatrakosAADG20}.
The main goal of a CER system is to detect interesting activity patterns occurring within a stream of events,
coming from sensors or other devices. 
Complex Events must be detected with minimal latency.
As a result,
a significant body of work has been devoted to computational optimization issues. 
Less attention has been paid to forecasting event patterns \cite{DBLP:journals/vldb/GiatrakosAADG20},
despite the fact that forecasting has attracted considerable attention in various related research areas,
such as time-series forecasting \cite{montgomery2015introduction},
sequence prediction \cite{DBLP:journals/jair/BegleiterEY04,buhlmann1999variable,DBLP:journals/ml/RonST96,DBLP:journals/tcom/ClearyW84,DBLP:journals/tit/WillemsST95}, 
temporal mining \cite{DBLP:conf/icdm/VilaltaM02,DBLP:conf/kdd/LaxmanTW08,DBLP:journals/eswa/ZhouCG15,DBLP:journals/vldb/ChoWYZC11} 
and process mining \cite{DBLP:journals/tsc/Marquez-Chamorro18}. 
The need for Complex Event Forecasting (CEF) has been acknowledged though,
as evidenced by several conceptual proposals \cite{DBLP:conf/bci/FulopBTDVF12,DBLP:conf/edoc/ChristKK16,DBLP:journals/tasm/ArtikisBBWEFGHLPSS14,DBLP:conf/debs/EngelE11}.

Consider, for example, the domain of credit card fraud management \cite{DBLP:conf/debs/ArtikisKCBMSFP17},
where the detection of suspicious activity patterns of credit cards must occur with minimal latency that is in the order of a few milliseconds.
The decision margin is extremely narrow. 
Being able to forecast that a certain sequence of transactions is very likely to be a fraudulent pattern provides wider margins both for decision and for action.
For example, a processing system might decide to devote more resources and higher priority to those suspicious patterns to ensure that the latency requirement will be satisfied.
The field of moving object monitoring (for ships at sea, aircrafts in the air or vehicles on the ground) provides yet another example where CEF could be a crucial functionality \cite{DBLP:conf/edbt/VourosVSDPGTPAA18}.
Collision avoidance is obviously of paramount importance for this domain.
A monitoring system with the ability to infer that two (or more) moving objects are on a collision course and forecast that they will indeed collide if no action is taken would provide significant help to the relevant authorities.
CEF could play an important role even in in-silico biology,
where computationally demanding simulations of biological systems are often executed to determine the properties of these systems and their response to treatments \cite{ozik2019learning}.
These simulations are typically run on supercomputers and are evaluated afterwards to determine which of them seem promising enough from a therapeutic point of view. 
A system that could monitor these simulations as they run, forecast which of them will turn out to be non-pertinent and decide to terminate them at an early stage, could thus save valuable computational resources and significantly speed-up the execution of such in-silico experiments.
Note that these are domains with different characteristics.
For example, some of them have a strong geospatial component (monitoring of moving entities),
whereas in others this component is minimal (in-silico biology).
Domain-specific solutions (e.g., trajectory prediction for moving objects) cannot thus be universally applied. 
We need a more general framework.

Towards this direction,
we present a formal framework for CEF,
along with an implementation and extensive experimental results on real and synthetic data from diverse application domains.
Our framework allows a user to define a pattern for a complex event, 
e.g., a pattern for fraudulent credit card transactions or for two moving objects moving in close proximity and towards each other.
It then constructs a probabilistic model for such a pattern in order to forecast,
on the basis of an event stream,
if and when a complex event is expected to occur.
We use the formalism of symbolic automata \cite{DBLP:conf/cav/DAntoniV17} to encode a pattern and that of prediction suffix trees \cite{DBLP:journals/ml/RonST96,DBLP:conf/nips/RonST93} to learn a probabilistic model for the pattern.    
We formally show how symbolic automata can be combined with prediction suffix trees to perform CEF.
Prediction suffix trees fall under the class of the so-called variable-order Markov models,
i.e., Markov models whose order (how deep into the past they can look for dependencies) can be increased beyond what is computationally possible with full-order models. 
They can do this by avoiding a full enumeration of every possible dependency and focusing only on ``meaningful'' dependencies. 

Our empirical analysis shows the advantage of being able to use high-order models over related non-Markov methods for CEF and methods based on low-order Markov models (or Hidden Markov Models). 
The price we have to pay for this increased accuracy is a decrease in throughput,
which still however remains high (typically tens of thousands of events per second).
The training time is also increased, 
but still remains within the same order of magnitude.
This fact allows us to be confident that training could also be performed online.

Our contributions may be summarized as follows:
\begin{itemize}
	\item We present a CEF framework that is both formal and easy to use. 
	It is often the case that CER frameworks lack clear semantics,
	which in turn leads to confusion about how patterns should be written and which operators are allowed \cite{DBLP:journals/vldb/GiatrakosAADG20}. 
	This problem is exacerbated in CEF, 
	where a formalism for defining the patterns to be forecast may be lacking completely. 
	Our framework is formal, compositional and as easy to use as writing regular expressions.
	The only basic requirement is that the user declaratively define a pattern and provide a training dataset.
	\item Our framework can uncover deep probabilistic dependencies in a stream by using a variable-order Markov model. 
	By being able to look deeper into the past, we achieve higher accuracy scores compared to other state-of-the-art solutions for CEF, as shown in our extensive empirical analysis.
	\item Our framework can perform various types of forecasting and thus subsumes previous methods 		that restrict themselves to one type of forecasting.
	It can perform both simple event forecasting (i.e., predicting what the next input event might 			be) and Complex Event forecasting (events defined through a pattern). 
	As we explain later, moving from simple event to Complex Event forecasting is not trivial.
	Using simple event forecasting to project in the future the most probable sequence of input 			events and then attempt to detect Complex Events on this future sequence yields sub-optimal 			results. 
	A system that can perform simple event forecasting cannot thus be assumed to perform CEF as 			well.
	\item We also discuss the issue of how the forecasts of a CEF system may be evaluated with 				respect to their quality.
	Previous methods have used metrics borrowed from time-series forecasting (e.g., the root mean 			square error) or typical machine learning tasks (e.g., precision). 
	We propose a more comprehensive set of metrics that takes into account the idiosyncrasies of CEF.
	Besides accuracy itself, the usefulness of forecasts is also judged by their ``earliness''.
	We discuss how the notion of earliness may be quantified.
\end{itemize}

\subsection{Running Example}
\label{sec:example}

We now present the general approach of CER/CEF systems,
along with an example that we will use throughout the rest of the paper to make our presentation more accessible.

The input to a CER system consists of two main components: 
a stream of events, also called simple derived events (SDEs);
and a set of patterns that define relations among the SDEs.
Instances of pattern satisfaction are called Complex Events (CEs).
The output of the system is another stream, composed of the detected CEs.
Typically, CEs must be detected with very low latency,
which, in certain cases, may even be in the order of a few milliseconds \cite{luckham_power_2001,DBLP:books/daglib/0024062,hedtstuck_complex_2017}.

\begin{table*}[!ht]
\centering
\caption{Example event stream from the maritime domain.}
\begin{tabular}{cccccccc} 
\toprule
Navigational status & fishing & fishing & fishing & under way & under way & under way & ... \\ 
\midrule
vessel id & 78986 & 78986 & 78986 & 78986 & 78986 & 78986 & ... \\
\midrule
speed & 2 & 1 & 3 & 22 & 19 & 27 & ... \\
\midrule
timestamp & 1 & 2 & 3 & 4 & 5 & 6 & ... \\
\bottomrule
\end{tabular}
\label{table:stream}
\end{table*}

As an example, consider the scenario of a system receiving an input stream consisting of events emitted from vessels sailing at sea.
These events may contain information regarding the status of a vessel,
e.g., its location, speed and heading.
This is indeed a real-world scenario and the emitted messages are called AIS (Automatic Identification System) messages.
Besides information about a vessel's kinematic behavior,
each such message may contain additional information about the vessel's status (e.g., whether it is fishing),
along with a timestamp and a unique vessel identifier.
Table \ref{table:stream} shows a possible stream of AIS messages,
including \emph{speed} and \emph{timestamp} information.
A maritime expert may be interested to detect several activity patterns for the monitored vessels,
such as sudden changes in the kinematic behavior of a vessel (e.g., sudden accelerations),
sailing in protected (e.g., NATURA) areas, etc.
The typical workflow consists of the analyst first writing these patterns in some (usually) declarative language,
which are then used by a computational model applied on the stream of SDEs to detect CEs. 

\subsection{Structure of the Paper}

The rest of the paper is structured as follows. 
We start by presenting in Section \ref{sec:related} the relevant literature on CEF. 
Since work on CEF has been limited thus far,
we also briefly mention forecasting ideas from some other related fields that can provide inspiration to CEF.
Subsequently, in Section \ref{sec:symbolic} we discuss the formalism of symbolic automata and how it can be adapted to perform recognition on real-time event streams.
Section \ref{sec:prob} shows how we can create a probabilistic model for a symbolic automaton by using prediction suffix trees,
while Section \ref{sec:complexity} presents a detailed complexity analysis.
We then discuss how we can quantify the quality of forecasts in Section \ref{sec:metrics}.
We finally demonstrate the efficacy of our framework in Section \ref{sec:experiments},
by showing experimental results on two application domains.
We conclude with Section \ref{sec:outro},
discussing some possible directions for future work. 
The paper assumes a basic familiarity with automata theory, logic and Markov chains.
In Table \ref{table:notation} we have gathered the notation that we use throughout the paper, 
along with a brief description of every symbol.

\begin{table*}
\centering
\scriptsize
\caption{Notation used throughout the paper.}
\begin{tabular}{ccc}
\toprule
Symbol & Meaning \\
\midrule
\multicolumn{2}{c}{\textbf{Boolean algebra}} \\
\midrule
$\mathcal{A}$ & effective Boolean algebra \\
\midrule
$\mathcal{D}$ & domain elements of a Boolean algebra \\
\midrule
$\Psi$ & predicates of a Boolean algebra \\
\midrule
$\bot,\top$ & \false\ and \true\ predicates of a Boolean algebra  \\
\midrule
$\wedge$, $\vee$, $\neg$ & logical conjunction, disjunction, negation \\
\midrule
$\mathcal{L}$ ($\mathcal{L} \subseteq \mathcal{D}^{*}$) & a language over $\mathcal{D}$ \\
\midrule
\multicolumn{2}{c}{\textbf{Symbolic expressions and automata}} \\
\midrule
$\epsilon$ & the ``empty'' symbol \\
\midrule
$R_{1} + R_{2}$, $R_{1} \cdot R_{2}$, $R^{*}$ & regular disjunction / concatenation / iteration\\
\midrule
$M$ & automaton \\
\midrule
$Q$, $q^{s}$, $Q^{f}$ & automaton states / start state / final states \\
\midrule
$\Delta$, $\delta$ & automaton transition function / transition \\
\midrule
$\mathcal{L}(R)$, $\mathcal{L}(M)$ & language of expression $R$ / automaton $M$ \\ 
\midrule
\multicolumn{2}{c}{\textbf{Streaming expressions and automata}} \\
\midrule
$t_{i} \in \mathcal{D}$ & tuple / simple event \\
\midrule
$S=t_{1},t_{2},\cdots$, $S_{i..j}=t_{i},\cdots,t_{j}$ & stream / stream ``slice'' from index $i$ to $j$ \\
\midrule
$c=[i,q]$ & \parbox{8.0cm}{\centering automaton configuration ($i$ current position, $q$ current state)} \\
\midrule
$[i,q] \overset{\delta}{\rightarrow} [i',q']$ & configuration succession \\
\midrule
$\varrho = [1,q_{1}] \overset{\delta_{1}}{\rightarrow} \cdots \overset{\delta_{k}}{\rightarrow} [k+1,q_{k+1}]$ & run of automaton $M$ over  stream $S_{1..k}$ \\
\midrule
$N=\mathit{Minterms}(\mathit{Predicates}(M))$ & minterms of automaton $M$ \\
\midrule
\multicolumn{2}{c}{\textbf{Variable-order Markov models}} \\
\midrule
$\Sigma$, $\sigma$, $s$ & alphabet of classical automaton / symbol / string \\
\midrule
$\mathit{suffix}(s)$ & the longest suffix of $s$ different from $s$ \\
\midrule
$\hat{P}$ & predictor \\
\midrule
$l(\hat{P},S_{1..k})$ & average log-loss of $\hat{P}$ over $S_{1..k}$ \\
\midrule
$T$ & prediction suffix tree \\
\midrule
$\gamma_{s}$ & next symbol probability function for a node of $T$\\
\midrule
$\theta_{1}$,$\theta_{2}$ & thresholds for tree learning \\
\midrule
$\alpha$ & approximation parameter for tree learning \\
\midrule
$n$ & \parbox{8.0cm}{\centering maximum number of states for learned suffix automaton} \\
\midrule
$m$ & order of suffix tree / automaton / Markov model \\
\midrule
$\tau$ & transition function of suffix automaton \\
\midrule
$\gamma$, $\Gamma$ & \parbox{8.0cm}{\centering next symbol probability function of suffix automaton / embedding}\\
\midrule
$\pi$ & \parbox{8.0cm}{\centering initial probability distribution over start states of automaton / embedding}\\
\bottomrule
\end{tabular}
\label{table:notation}
\end{table*}

\section{Related Work}
\label{sec:related}

There are multiple ways to define the task of forecasting over time-evolving data streams. 
Before proceeding with the presentation of previous work on forecasting,
we first begin with a terminological clarification.
It is often the case that the terms ``forecasting'' and ``prediction'' are used interchangeably as equivalent terms. 
For reasons of clarity,
we opt for the term of ``forecasting'' to describe our work,
since there does exist a conceptual difference between forecasting and prediction, as the latter term is understood in machine learning.
In machine learning,
the goal is to ``predict'' the output of a function on previously unseen input data.
The input data need not necessarily have a temporal dimension and the term ``prediction'' refers to the output of the learned function on a new data point.  
For this reason we avoid using the term ``prediction''.
Instead, we choose the term ``forecasting'' to define the task of predicting the temporally future output of some function or the occurrence of an event.
Time is thus a crucial component for forecasting.
Moreover, an important challenge stems from the fact that,
from the (current) timepoint where a forecast is produced until the (future) timepoint for which we try to make a forecast,
no data is available.
A forecasting system must (implicitly or explicitly) fill in this data gap in order to produce a forecast.

In what follows,
we present previous work on CEF,
as defined above,
in order of increasing relevance to CER.
Since work on CEF has been limited thus far,
we start by briefly mentioning some forecasting ideas from other fields and discuss how CEF differs from these research areas.

\emph{Time-series forecasting.} 
Time-series forecasting is an area with some similarities to CEF and a significant history of contributions \cite{montgomery2015introduction}. 
However, it is not possible to directly apply techniques from time-series forecasting to CEF.
Time-series forecasting typically focuses on streams of (mostly) real-valued variables and the goal is to forecast relatively simple patterns. 
On the contrary, in CEF we are also interested in categorical values, 
related through complex patterns and involving multiple variables.
Another limitation of time-series forecasting methods is that they do not provide a language with which we can define complex patterns,
but simply try to forecast the next value(s) from the input stream/series. 
In CER, the equivalent task would be to forecast the next input event(s) (SDEs).
This task in itself is not very useful for CER though,
since the majority of SDE instances should be ignored and do not contribute to the detection of CEs
(see the discussion on selection policies in Section \ref{sec:symbolic}).
For example,
if we want to determine whether a ship is following a set of pre-determined waypoints at sea,
we are only interested in the messages where the ship ``touches'' each waypoint. 
All other intermediate messages are to be discarded and should not constitute part of the match. 
CEs are more like ``anomalies'' and their number is typically orders of magnitude lower than the number of SDEs.
One could possibly try to leverage techniques from SDE forecasting to perform CE forecasting.
At every timepoint, we could try to estimate the most probable sequence of future SDEs, 
then perform recognition on this future stream of SDEs and check whether any future CEs are detected.
We have experimentally observed that such an approach yields sub-optimal results. 
It almost always fails to detect any future CEs.
This behavior is due to the fact that CEs are rare.
As a result, projecting the input stream into the future creates a ``path'' with high probability but fails to include the rare ``paths'' that lead to a CE detection.
Because of this serious under-performance of this method,
we do not present detailed experimental results.

\emph{Sequence prediction (compression).}
Another related field is that of prediction of discrete sequences over finite alphabets and is closely related to the field of compression,
as any compression algorithm can be used for prediction and vice versa.
The relevant literature is extensive.
Here we focus on a sub-field with high importance for our work,
as we have borrowed ideas from it.
It is the field of sequence prediction via variable-order Markov models \cite{DBLP:journals/jair/BegleiterEY04,buhlmann1999variable,DBLP:journals/ml/RonST96,DBLP:conf/nips/RonST93,DBLP:journals/tcom/ClearyW84,DBLP:journals/tit/WillemsST95}.
As the name suggests,
the goal is to perform prediction by using a high-order Markov model.
Doing so in a straightforward manner,
by constructing a high-order Markov chain with all its possible states,
is prohibitively expensive due to the combinatorial explosion of the number of states.
Variable-order Markov models address this issue by retaining only those states that are ``informative'' enough.
In Section \ref{sec:vmm},
we discuss the relevant literature in more details.
The main limitation of previous methods for sequence prediction is that they they also do not provide a language for patterns and focus exclusively on next symbol prediction,
i.e., they try to forecast the next symbol(s) in a stream/string of discrete symbols.
As already discussed,
this is a serious limitation for CER.
An additional limitation is that they work on single-variable discrete sequences of symbols,
whereas CER systems consume streams of events,
i.e., streams of tuples with multiple variables, both numerical and categorical.
Notwithstanding these limitations,
we show that variable-order models can be combined with symbolic automata in order to overcome their restrictions and perform CEF.

\emph{Temporal mining.}
Forecasting methods have also appeared in the field of temporal pattern mining \cite{DBLP:conf/icdm/VilaltaM02,DBLP:conf/kdd/LaxmanTW08,DBLP:journals/eswa/ZhouCG15,DBLP:journals/vldb/ChoWYZC11}. 
A common assumption in these methods is that patterns are usually defined either as association rules \cite{DBLP:conf/sigmod/AgrawalIS93} or as frequent episodes \cite{DBLP:journals/datamine/MannilaTV97}.
In \cite{DBLP:conf/icdm/VilaltaM02} the goal is to identify sets of event types that frequently precede a rare, target event within a temporal window, using a framework similar to that of association rule mining.
In \cite{DBLP:conf/kdd/LaxmanTW08}, a forecasting model is presented,
based on a combination of standard frequent episode discovery algorithms, Hidden Markov Models and mixture models.
The goal is to calculate the probability of the immediately next event in the stream.
In \cite{DBLP:journals/eswa/ZhouCG15} a method is presented for batch, online mining of sequential patterns.
The learned patterns are used to test whether a prefix matches the last events seen in the stream and therefore make a forecast.
The method proposed in \cite{DBLP:journals/vldb/ChoWYZC11} starts with a given episode rule (as a Directed Acyclic Graph) and detects the
minimal occurrences of the antecedent of a rule defining a complex event, 
i.e., those ``clusters'' of antecedent events that are closer together in time.
From the perspective of CER,
the disadvantage of these methods is that they usually target simple patterns,
defined either as strict sequences or as sets of input events.
Moreover, the input stream is composed of symbols from a finite alphabet,
as is the case with the compression methods mentioned above. 

\emph{Sequence prediction based on neural networks.}
Lately, a significant body of work has focused on event sequence prediction and point-of-interest recommendations through the use of neural networks
(see, for example, \cite{DBLP:conf/ijcai/LiDL18,DBLP:conf/ijcai/ChangPPKK18}).
These methods are powerful in predicting the next input event(s) in a sequence of events,
but they suffer from limitations already mentioned above.
They do not provide a language for defining complex patterns among events and their focus is thus on SDE forecasting.
An additional motivation for us to first try a statistical method rather than going directly to neural networks is that, 
in other related fields, 
such as time series forecasting,
statistical methods have often been proven to be more accurate and less demanding in terms of computational resources than ML ones \cite{makridakis2018statistical}.

\emph{Process mining.} 
Compared to the previous categories for forecasting,
the field of process mining is more closely related to CER \cite{van2011process}.
Processes are typically defined as transition systems (e.g., automata or Petri nets) and are used to monitor a system, 
e.g., for conformance testing.
Process mining attempts to automatically learn a process from a set of traces,
i.e., a set of activity logs.
Since 2010,
a significant body of work has appeared,
targeting process prediction,
where the goal is to forecast if and when a process is expected to be completed
(for surveys, see \cite{DBLP:journals/tsc/Marquez-Chamorro18,DBLP:conf/bpm/Francescomarino18}). 
According to \cite{DBLP:journals/tsc/Marquez-Chamorro18},
until 2018, 39 papers in total have been published dealing with process prediction.
At a first glance,
process prediction seems very similar to CEF.
At a closer look though,
some important differences emerge.
An important difference is that processes are usually given directly as transition systems,
whereas CER patterns are defined in a declarative manner.
The transition systems defining processes are usually composed of long sequences of events.
On the other hand,
CER patterns are shorter,
may involve Kleene-star, iteration operators (usually not present in processes)  
and may even be instantaneous.
Consider,
for example,
a pattern for our running example,
trying to detect speed violations by simply checking whether a vessel's speed exceeds some threshold. 
This pattern could be expanded to detect more violations by adding more disjuncts,
e.g., for checking whether a vessel is sailing within a restricted area,
all of which might be instantaneous.
A CEF system cannot always rely on the memory implicitly encoded in a transition system and has to be able to learn the sequences of events that lead to a (possibly instantaneous) CE.
Another important difference is that process prediction focuses on traces, which are complete, full matches,
whereas CER focuses on continuously evolving streams which may contain many irrelevant events.
A learning method has to take into account the presence of these irrelevant events.
In addition to that,
since CEs are rare events, 
the datasets are highly imbalanced,
with the vast majority of ``labels'' being negative
(i.e., most forecasts should report that no CE is expected to occur,
with very few being positive).
A CEF system has to strike a fine balance between the positive and negative forecasts it produces
in order to avoid drowning the positives in the flood of all the negatives and, at the same time, avoid over-producing positives that lead to false alarms.
This is also an important issue for process prediction,
but becomes critical for a CEF system,
due to the imbalanced nature of the datasets.
In Section \ref{sec:experiments},
we have included one method from the field of process prediction to our empirical evaluation.
This ``unfair'' comparison (in the sense that it is applied on datasets more suitable for CER) shows that this method consistently under-performs with respect to other methods from the field of CEF. 

\emph{Complex event forecasting.}
Contrary to process prediction, forecasting has not received much attention in the field of CER,
although some conceptual proposals have acknowledged the need for CEF \cite{DBLP:conf/bci/FulopBTDVF12,DBLP:conf/debs/EngelE11,DBLP:conf/edoc/ChristKK16}.
The first concrete attempt at CEF was presented in \cite{DBLP:conf/debs/MuthusamyLJ10}.
A variant of regular expressions was used to define CE patterns,
which were then compiled into automata.
These automata were translated to Markov chains through a direct mapping,
where each automaton state was mapped to a Markov chain state.
Frequency counters on the transitions were used to estimate the Markov chain's transition matrix.
This Markov chain was finally used to estimate if a CE was expected to occur within some future window.
As we explain in Section \ref{sec:vmm},
in the worst case,
such an approach assumes that all SDEs are independent 
(even when the states of the Markov chain are not independent)  
and is thus unable to encode higher-order dependencies. 
This issue is explained in more detail in Section \ref{sec:vmm}.
Another example of event forecasting was presented in \cite{DBLP:conf/wf-iot/AkbarCMZ15}.
Using Support Vector Regression,
the proposed method was able to predict the next input event(s) within some future window.
This technique is similar to time-series forecasting,
as it mainly targets the prediction of the (numerical) values of the attributes of the input (SDE) events
(specifically, traffic speed and intensity from a traffic monitoring system).
Strictly speaking,
it cannot therefore be considered a CE forecasting method,
but a SDE forecasting one.
Nevertheless, the authors of \cite{DBLP:conf/wf-iot/AkbarCMZ15} proposed the idea that these future SDEs may be used by a CER engine to detect future CEs.
As we have already mentioned though,
in our experiments,
this idea has yielded poor results.
In \cite{DBLP:conf/colcom/PandeyNC11},
Hidden Markov Models (HMM) are used to construct a probabilistic model for the behavior of a transition system describing a CE.
The observable variable of the HMM corresponds to the states of the transition system,
i.e., an observation sequence of length $l$ for the HMM consists of the sequence of states visited by the system after consuming $l$ SDEs.
These $l$ SDEs are mapped to the hidden variable,
i.e.,
the last $l$ values of the hidden variable are the last $l$ SDEs.
In principle,
HMMs are more powerful than Markov chains.
In practice, however, HMMs are hard to train (\cite{DBLP:journals/jair/BegleiterEY04,DBLP:journals/ml/AbeW92}) and require elaborate domain modeling,
since mapping a CE pattern to a HMM is not straightforward (see Section \ref{sec:vmm} for details).
In contrast, 
our approach constructs seamlessly a probabilistic model from a given CE pattern (declaratively defined).
Automata and Markov chains are again used in \cite{DBLP:conf/debs/AlevizosAP17,DBLP:conf/lpar/AlevizosAP18}.
The main difference of these methods compared to \cite{DBLP:conf/debs/MuthusamyLJ10} is that they can accommodate higher-order dependencies by creating extra states for the automaton of a pattern.
The method presented in \cite{DBLP:conf/debs/AlevizosAP17} has two important limitations: 
first, it works only on discrete sequences of finite alphabets;
second, the number of states required to encode long-term dependencies grows exponentially.
The first issue was addressed in \cite{DBLP:conf/lpar/AlevizosAP18},
where symbolic automata are used that can handle infinite alphabets.
However, the problem of the exponential growth of the number of states still remains.
We show how this problem can be addressed by using variable-order Markov models.
A different approach is followed in \cite{li20data},
where knowledge graphs are used to encode events and their timing relationships. 
Stochastic gradient descent is employed to learn the weights of the graph's edges that determine how important an event is with respect to another target event.
However, this approach falls in the category of SDE forecasting, 
as it does not target complex events.
More precisely, it tries to forecast which predicates the forthcoming SDEs will satisfy,
without taking into account relationships between the events themselves (e.g., through simple sequences).

\section{Complex Event Recognition with Symbolic Automata}
\label{sec:symbolic}

Our approach for CEF is based on a specific formal framework for CER,
which we are presenting here.
There are various surveys of CER methods, 
presenting various CER systems and languages \cite{DBLP:journals/csur/CugolaM12,DBLP:journals/csur/AlevizosSAP17,DBLP:journals/vldb/GiatrakosAADG20}.
Despite this fact though, there is still no consensus about which operators must be supported by a CER language and what their semantics should be. 
In this paper, 
we follow \cite{DBLP:journals/vldb/GiatrakosAADG20} and \cite{DBLP:conf/icdt/GrezRU19}, 
which have established some core operators that are most often used.
In a spirit similar to \cite{DBLP:conf/icdt/GrezRU19},
we use automata as our computational model and define a CER language whose expressions can readily be converted to automata.
Instead of choosing one of the automaton models already proposed in the CER literature,
we employ symbolic regular expressions and automata \cite{DBLP:conf/cav/DAntoniV17,DBLP:journals/fmsd/DAntoniV15,DBLP:conf/icst/VeanesHT10}.
The rationale behind our choice is that,
contrary to other automata-based CER models,
symbolic regular expressions and automata have nice closure properties and clear (both declarative and operational), compositional semantics
(see \cite{DBLP:conf/icdt/GrezRU19} for a similar line of work,
based on symbolic transducers).
In previous automata-based CER systems,
it is unclear which operators may be used and if they can be arbitrarily combined
(see \cite{DBLP:conf/icdt/GrezRU19,DBLP:journals/vldb/GiatrakosAADG20} for a discussion of this issue).
On the contrary,
the use of symbolic automata allows us to construct any pattern that one may desire through an arbitrary use of the provided operators.
One can then check whether a stream satisfies some pattern either declaratively
(by making use of the definition for symbolic expressions, presented below) or operationally (by using a symbolic automaton).
This even allows us to write unit tests (as we have indeed done) ensuring that the semantics of symbolic regular expressions do indeed coincide with those of symbolic automata,
something not possible with other frameworks.
In previous methods,
there is also a lack of understanding with respect to the properties of the employed computational models,
e.g., whether the proposed automata are determinizable,
an important feature for our work.
Symbolic automata, on the other hand, have nice closure properties and are well-studied.
Notice that this would also be an important feature for possible optimizations based on pattern re-writing,
since such re-writing would require us to have a mechanism determining whether two expressions are equivalent. 
Our framework provides such a mechanism.

\subsection{Symbolic Expressions and Automata}

The main idea behind symbolic automata is that each transition,
instead of being labeled with a symbol from an alphabet, 
is equipped with a unary formula from an effective Boolean algebra \cite{DBLP:conf/cav/DAntoniV17}.
A symbolic automaton can then read strings of elements and,
upon reading an element while in a given state,
can apply the predicates of this state's outgoing transitions to that element.
The transitions whose predicates evaluate to \true\ are said to be ``enabled'' and the automaton moves to their target states.

The formal definition of an effective Boolean algebra is the following: 
\begin{definition}[Effective Boolean algebra \cite{DBLP:conf/cav/DAntoniV17}]
An effective Boolean algebra is a tuple 
($\mathcal{D}$, $\Psi$, $\llbracket \_ \rrbracket$, $\bot$, $\top$, $\vee$, $\wedge$, $\neg$)
where 
\begin{itemize}
\item $\mathcal{D}$ is a set of domain elements; 
\item $\Psi$ is a set of predicates closed under the Boolean connectives; 
\item $\bot,\top\in \Psi$ ; 
\item and the component $\llbracket \_ \rrbracket : \Psi \rightarrow 2^{\mathcal{D}}$ is a function
such that 
\begin{itemize}
	\item $\llbracket \bot \rrbracket = \emptyset$
	\item $\llbracket \top \rrbracket = \mathcal{D}$
	\item and $\forall \phi,\psi \in \Psi$:
	\begin{itemize}
		\item $\llbracket \phi \vee \psi \rrbracket = \llbracket \phi \rrbracket \cup \llbracket \psi \rrbracket$
		\item $\llbracket \phi \wedge \psi \rrbracket = \llbracket \phi \rrbracket \cap \llbracket \psi \rrbracket $
		\item $\llbracket \neg \phi \rrbracket = \mathcal{D} \setminus \llbracket \phi \rrbracket$
	\end{itemize}
\end{itemize}
\end{itemize}
It is also required that checking satisfiability of $\phi$, i.e., whether $\llbracket \phi \rrbracket \neq \emptyset$,  is decidable and that the operations of $\vee$, $\wedge$ and $\neg$ are computable.
\end{definition}
Using our running example, 
such an algebra could be one consisting of AIS messages,
corresponding to $\mathcal{D}$,
along with two predicates about the speed of a vessel,
e.g., $\mathit{speed} < 5$ and $\mathit{speed} > 20$.
These predicates would correspond to $\Psi$.
The predicate $\mathit{speed} < 5$ would be mapped, 
via $\llbracket \_ \rrbracket$,
to the set of all AIS messages whose speed level is below 5 knots.
According to the definition above,
$\bot$ and $\top$ should also belong to $\Psi$,
along with all the combinations of the original two predicates constructed from the Boolean connectives,
e.g., $\neg (\mathit{speed} < 5) \wedge \neg (\mathit{speed} > 20)$.

Elements of $\mathcal{D}$ are called \emph{characters} and finite sequences of characters are called \emph{strings}. 
A set of strings $\mathcal{L}$ constructed from elements of $\mathcal{D}$ ($\mathcal{L} \subseteq \mathcal{D}^{*}$, where $^{*}$ denotes Kleene-star) is called a language over $\mathcal{D}$.

As with classical regular expressions \cite{DBLP:books/daglib/0016921},
we can use symbolic regular expressions to represent a class of languages over $\mathcal{D}$.
\begin{definition}[Symbolic regular expression]
A symbolic regular expression (\sre) over an effective Boolean algebra ($\mathcal{D}$, $\Psi$, $\llbracket \_ \rrbracket$, $\bot$, $\top$, $\vee$, $\wedge$, $\neg$) is recursively defined as follows:
\begin{itemize}
	\item The constants $\epsilon$ and $\emptyset$ are symbolic regular expressions with $\mathcal{L}(\epsilon) = \{\epsilon\}$ and $\mathcal{L}(\emptyset) = \{ \emptyset \}$;
	\item If $\psi \in \Psi$, then $R := \psi$ is a symbolic regular expression, with $\mathcal{L}(\psi) = \llbracket \psi \rrbracket$, 
	i.e., the language of $\psi$ is the subset of $\mathcal{D}$ for which $\psi$ evaluates to \true;
	\item Disjunction / Union: If $R_{1}$ and $R_{2}$ are symbolic regular expressions, then $R := R_{1} + R_{2}$ is also a symbolic regular expression, with $\mathcal{L}(R) = \mathcal{L}(R_{1}) \cup \mathcal{L}(R_2)$;
	\item Concatenation / Sequence: If $R_{1}$ and $R_{2}$ are symbolic regular expressions, then $R := R_{1} \cdot R_{2}$ is also a symbolic regular expression, with $\mathcal{L}(R) = \mathcal{L}(R_{1}) \cdot \mathcal{L}(R_2)$, where $\cdot$ denotes concatenation. $\mathcal{L}(R)$ is then the set of all strings constructed from concatenating each element of $\mathcal{L}(R_{1})$ with each element of $\mathcal{L}(R_{2})$;
	\item Iteration / Kleene-star: If $R$ is a symbolic regular expression, then $R' := R^{*}$ is a symbolic regular expression, with $\mathcal{L}(R^{*})=(\mathcal{L}(R))^{*}$, 
	where $\mathcal{L}^{*} = \bigcup\limits_{i \geq 0}{\mathcal{L}^{i}}$ and $\mathcal{L}^{i}$ is the concatenation of $\mathcal{L}$ with itself $i$ times.
\end{itemize}
\end{definition}
As an example,
if we want to detect instances of a vessel accelerating suddenly,
we could write the expression $R := (\mathit{speed < 5}) \cdot (\mathit{speed > 20})$.
The third and fourth events of the stream of Table \ref{table:stream} would then belong to the language of $R$.

Given a Boolean algebra, we can also define symbolic automata.
The definition of a symbolic automaton is the following:
\begin{definition}[Symbolic finite automaton \cite{DBLP:conf/cav/DAntoniV17}]
A symbolic finite automaton (\sfa) is a tuple $M=$($\mathcal{A}$, $Q$, $q^{s}$, $Q^{f}$, $\Delta$), 
where 
\begin{itemize}
	\item $\mathcal{A}$ is an effective Boolean algebra;
	\item $Q$ is a finite set of states;
	\item $q^{s} \in Q$ is the initial state;
	\item $Q^{f} \subseteq Q$ is the set of final states;
	\item $\Delta \subseteq Q \times \Psi_{\mathcal{A}} \times Q$ is a finite set of transitions.
\end{itemize}
\end{definition}
A string $w=a_{1}a_{2} \cdots a_{k}$  is accepted by a \sfa\ $M$ iff, 
for $1 \leq i \leq k$, there exist transitions $q_{i-1} \overset{a_{i}}{\rightarrow} q_{i}$ such that $q_{0}=q^{s}$ and $q_{k} \in Q^{f}$. 
We refer to the set of strings accepted by $M$ as the language of $M$, denoted by $\mathcal{L}(M)$ \cite{DBLP:conf/cav/DAntoniV17}.
Figure \ref{fig:sfa_speed} shows a \sfa\ that can detect the expression of sudden acceleration for our running example.

As with classical regular expressions and automata,
we can prove that every symbolic regular expression can be translated to an equivalent (i.e., with the same language) symbolic automaton.
\begin{proposition}
\label{proposition:sre2sfa}
For every symbolic regular expression $R$ there exists a symbolic finite automaton $M$ such that $\mathcal{L}(R) = \mathcal{L}(M)$.
\end{proposition}
\begin{proof}
The proof is essentially the same as that for classical expressions and automata \cite{DBLP:books/daglib/0016921}.
It is a constructive proof starting from the base case of an expression that is a single predicate (instead of a symbol, as in classical expressions) and then proceeds in a manner identical to that of the classical case.
For the sake of completeness, 
the full proof is provided in the Appendix, Section \ref{sec:proof:sre2sfa}.
\end{proof}

\begin{figure}
\centering
\begin{subfigure}[t]{0.7\textwidth}
	\includegraphics[width=0.75\textwidth]{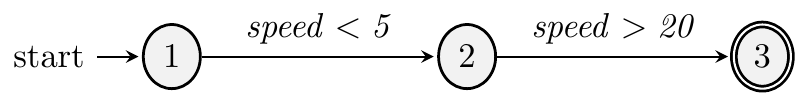}
	\caption{\sfa\ for the \sre\ $R := (\mathit{speed < 5}) \cdot (\mathit{speed > 20})$.}
	\label{fig:sfa_speed}
\end{subfigure}
\begin{subfigure}[t]{0.7\textwidth}
	\includegraphics[width=0.95\textwidth]{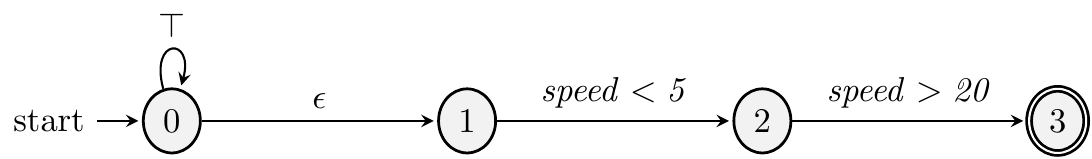}
	\caption{Streaming \sfa\ for $R := (\mathit{speed < 5}) \cdot (\mathit{speed > 20})$.
	$\top$ is a special predicate that always evaluates to \true. 
	$\top$ transitions are thus triggered for every event. 
	$\epsilon$ transitions triggered even in the absence of an event.
	}
	\label{fig:ssfa_speed}
\end{subfigure}
\caption{Examples of a symbolic automaton and streaming symbolic automaton. States with double circles are final.}
\label{fig:ssfa}
\end{figure}

\subsection{Streaming Expressions and Automata}

Our discussion thus far has focused on how \sre\ and \sfa\ can be applied to bounded strings that are known in their totality before recognition.
We feed a string to a \sfa\ and we expect an answer about whether the whole string belongs to the automaton's language or not.
However, in CER and CEF we need to handle continuously updated streams of events and detect instances of \sre\ satisfaction as soon as they appear in a stream. 
For example, the automaton of the (classical) regular expression $a \cdot b$ would accept only the string $a,b$.
In a streaming setting, 
we would like the automaton to report a match every time this string appears in a stream.
For the stream $a,b,c,a,b,c$, 
we would thus expect two matches to be reported,
one after the second symbol and one after the fifth.

In order to accommodate this scenario, 
slight modifications are required so that \sre\ and \sfa\ may work in a streaming setting.
First, 
we need to make sure that the automaton can start its recognition after every new element.
If we have a classical regular expression $R$,
we can achieve this by applying on the stream the expression $\Sigma^{*} \cdot R$,
where $\Sigma$ is the automaton's (classical) alphabet.
For example,
if we apply $R := \{a,b,c\}^{*} \cdot (a \cdot b)$ on the stream $a,b,c,a,b,c$,
the corresponding automaton would indeed reach its final state after reading the second and the fifth symbols.
In our case, 
events come in the form of tuples with both numerical and categorical values. 
Using database systems terminology we can speak of tuples from relations of a database schema \cite{DBLP:conf/icdt/GrezRU19}.
These tuples constitute the set of domain elements $\mathcal{D}$.
A stream $S$ then has the form of an infinite sequence $S=t_{1},t_{2},\cdots$, where each $t_{i}$ is a tuple ($t_{i} \in \mathcal{D}$).
Our goal is to report the indices $i$ at which a CE is detected.

More precisely, if $S_{1..k}=\cdots,t_{k-1},t_{k}$ is the prefix of $S$ up to the index $k$,
we say that an instance of a \sre\ $R$ is detected at $k$ iff there exists a suffix $S_{m..k}$ of $S_{1..k}$ such that $S_{m..k} \in \mathcal{L}(R)$.
In order to detect CEs of a \sre\ $R$ on a stream, we use a streaming version of \sre\ and \sfa.
\begin{definition}[Streaming SRE and SFA]
If $R$ is a \sre, then $R_{s}= \top^{*} \cdot R$ is called the streaming \sre\ (\ssre) corresponding to $R$.
A \sfa\ $M_{R_{s}}$ constructed from $R_{s}$ is called a streaming \sfa\ (\ssfa) corresponding to $R$.
\end{definition}
Using $R_{s}$ we can detect CEs of $R$ while reading a stream $S$,
since a stream segment $S_{m..k}$ belongs to the language of $R$ iff the prefix $S_{1..k}$ belongs to the language of $R_{s}$.
The prefix $\top^{*}$ lets us skip any number of events from the stream and start recognition at any index $m, 1 \leq m \leq k$.
\begin{proposition}
\label{proposition:streamingsre}
If $S=t_{1},t_{2},\cdots$ is a stream of domain elements from an effective Boolean algebra $\mathcal{A} = (\mathcal{D}$, $\Psi$, $\llbracket \_ \rrbracket$, $\bot$, $\top$, $\vee$, $\wedge$, $\neg$), where $t_{i} \in \mathcal{D}$, and $R$ is a symbolic regular expression over the same algebra,
then, for every $S_{m..k}$, $S_{m..k} \in \mathcal{L}(R)$ iff $S_{1..k} \in \mathcal{L}(R_{s})$ (and $S_{1..k} \in \mathcal{L}(M_{R_{s}})$).
\end{proposition}
\begin{proof}
The proof is provided in the Appendix, Section \ref{sec:proof:streamingsre}.
\end{proof}
As an example,
if $R := (\mathit{speed < 5}) \cdot (\mathit{speed > 20})$ is the pattern for sudden acceleration,
then its \ssre\ would be $R_{s} := \top^{*} \cdot (\mathit{speed < 5}) \cdot (\mathit{speed > 20})$.
After reading the fourth event of the stream of Table \ref{table:stream},
$S_{1..4}$ would belong to the language of $\mathcal{L}(R_{s})$ and $S_{3..4}$ to the language of $\mathcal{L}(R)$.
Note that \ssre\ and \ssfa\ are just special cases of \sre\ and \sfa\ respectively.
Therefore, every result that holds for \sre\ and \sfa\ also holds for \ssre\ and \ssfa\ as well.
Figure \ref{fig:ssfa_speed} shows an example \ssfa.

The streaming behavior of a \ssfa\ as it consumes a stream $S$ can be formally defined using the notion of configuration:
\begin{definition}[Configuration of sSFA]
Assume $S=t_{1},t_{2},\cdots$ is a stream of domain elements from an effective Boolean algebra,
$R$ a symbolic regular expression over the same algebra
and $M_{R_{s}}$ a \ssfa\ corresponding to $R$.
A configuration $c$ of $M_{R_{s}}$ is a tuple $[i,q]$,
where $i$ is the current position of the stream, 
i.e., the index of the next event to be consumed,
and $q$ the current state of $M_{R_{s}}$.
We say that $c'=[i',q']$ is a successor of $c$ iff:
\begin{itemize}
	\item $\exists \delta \in M_{R_{s}}.\Delta: \delta = (q,\psi,q') \wedge (t_{i} \in \llbracket \psi \rrbracket \vee \psi = \epsilon)$;
	\item $i=i'$ if $\delta = \epsilon$. Otherwise, $i' = i + 1$.
\end{itemize}
We denote a succession by $[i,q] \overset{\delta}{\rightarrow} [i',q']$.
\end{definition}

For the initial configuration $c^{s}$, before consuming any events,
we have that $i=1$ and $c^{s}.q = M_{R_{s}}.q^{s}$, 
i.e. the state of the first configuration is the initial state of $M_{R_{s}}$.
In other words, for every index $i$, we move from our current state $q$ to another state $q'$ if there is an outgoing transition from $q$ to $q'$ and the predicate on this transition evaluates to \true\ for $t_{i}$.
We then increase the reading position by 1.
Alternatively, if the transition is an $\epsilon$-transition, we move to $q'$ without increasing the reading position.

The actual behavior of a \ssfa\ upon reading a stream is captured by the notion of the run:
\begin{definition}[Run of sSFA over stream]
A run $\varrho$ of a \ssfa\ $M$ over a stream $S_{1..k}$ is a sequence of successor configurations
$[1,q_{1}=M.q^{s}] \overset{\delta_{1}}{\rightarrow} [2,q_{2}] \overset{\delta_{2}}{\rightarrow} \cdots \overset{\delta_{k}}{\rightarrow} [k+1,q_{k+1}]$.
A run is called accepting iff $q_{k+1} \in M.Q^{f}$.
\end{definition}
A run $\varrho$ of a \ssfa\ $M_{R_{s}}$ over a stream $S_{1..k}$ is accepting iff $S_{1..k} \in \mathcal{L}(R_{s})$,
since $M_{R_{s}}$, after reading $S_{1..k}$, must have reached a final state.
Therefore, for a \ssfa\ that consumes a stream, the existence of an accepting run with configuration index $k+1$ implies that a CE for the \sre\ $R$ has been detected at the stream index $k$.

As far as the temporal model is concerned,
we assume that all SDEs are instantaneous.
They all carry a \emph{timestamp} attribute which is single, unique numerical value.
We also assume that the stream of SDEs is temporally sorted.
A sequence/concatenation operator is thus satisfied if the event of its first operand precedes in time the event of its second operand.
The exception to the above assumptions is when the stream is composed of multiple partitions and the defined pattern is applied on a per-partition basis.
For example,
in the maritime domain a stream may be composed of the sub-streams generated by all vessels and we may want to detect the same pattern for each individual vessel.
In such cases,
the above assumptions must hold for each separate partition but not necessarily across all partitions.
Another general assumption is that there is no imposed limit on the time elapsed between consecutive events in a sequence operation.

\subsection{Expressive Power of Symbolic Regular Expressions}

We conclude this section with some remarks about the expressive power of \sre\ and \sfa\ and how it meets the requirements of a CER system.
As discussed in \cite{DBLP:journals/vldb/GiatrakosAADG20,DBLP:conf/icdt/GrezRU19},
besides the three operators of regular expressions that we have presented and implemented in this paper (disjunction, sequence, iteration),
there exist some extra operators which should be supported by a CER system.
\emph{Negation} is one them.
If we use $!$ to denote the negation operator,
then $R' := !R$ defines a language which is the complement of the language of $R$.
Since \sfa\ are closed under complement \cite{DBLP:conf/cav/DAntoniV17},
negation is an operator that can be supported by our framework and has also been implemented
(but not discussed further).

The same is true for the operator of \emph{conjunction}.
If we use $\wedge$ to denote conjunction,
then $R := R_{1} \wedge R_{2}$ is an expression whose language consists of concatenated elements of $\mathcal{L}(R_{1})$ and $\mathcal{L}(R_{2})$, regardless of their order,
i.e., $\mathcal{L}(R) = \mathcal{L}(R_{1}) \cdot \mathcal{L}(R_{2}) \cup \mathcal{L}(R_{2}) \cdot \mathcal{L}(R_{1})$.
This operator can thus be equivalently expressed using the already available operators of concatenation ($\cdot$) and disjunction ($+$).

Another important notion in CER is that of \emph{selection policies}.
An expression like $R := R_{1} \cdot R_{2}$ typically implies that an instance of $R_{2}$ must immediately follow an instance of $R_{1}$.
As a result,
for the stream of Table \ref{table:stream} and $R := (\mathit{speed < 5}) \cdot (\mathit{speed > 20})$,
only one match will be detected at $\mathit{timestamp}=4$.
With selection policies,
we can relax the requirement for contiguous instances.
For example,
with the so-called \skipany\ policy,
any number of events are allowed to occur between $R_{1}$ and $R_{2}$.
If we apply this policy on $R := (\mathit{speed < 5}) \cdot (\mathit{speed > 20})$,
we would detect six CEs,
since the first three events of Table \ref{table:stream} can be matched with the two events at $\mathit{timestamp}=4$ and at $\mathit{timestamp}=6$,
if we ignore all intermediate events.
Selection policies can also be accommodated by our framework and have been implemented.
For a proof, using symbolic transducers, see \cite{DBLP:conf/icdt/GrezRU19}.
Notice, for example,
that an expression $R := R_{1} \cdot R_{2}$ can be evaluated with \skipany\ by being rewritten as $R' := R_{1} \cdot \top^{*} \cdot R_{2}$,
so that any number of events may occur between $R_{1}$ and $R_{2}$.

Support for hierarchical definition of patterns,
i.e., the ability to define patterns in terms of other patterns,
is yet another important feature in many CER systems.
Since \sre\ and \sfa\ are compositional by nature,
hierarchies are supported by default in our framework.
Although we do not treat these operators and functionalities explicitly in this paper,
their incorporation is possible within the expressive limits of \sre\ and \sfa\ and the results that we present in the next sections would still hold.

\section{Building a Probabilistic Model}
\label{sec:prob}

The main idea behind our forecasting method is the following:
Given a pattern $R$ in the form of a \sre, we first construct a \ssfa\ as described in the previous section. 
For event recognition, this would already be enough,
but in order to perform event forecasting, we translate the \ssfa\ to an equivalent deterministic \sfa\ (\dsfa).
This \dsfa\ can then be used to learn a probabilistic model, 
typically a Markov chain, 
that encodes dependencies among the events in an input stream.
Note that a non-deterministic automaton cannot be directly converted to a Markov chain,
since from each state we might be able to move to multiple other target states with a given event. 
Therefore, we first determinize the automaton. 

The probabilistic model is learned from a portion of the input stream which acts as a training dataset and it is then used to derive forecasts about the expected occurrence of the CE encoded by the automaton.
The issue that we address in this paper is how to build a model which retains long-term dependencies that are useful for forecasting.

Figure \ref{fig:flow} depicts all the required steps in order to produce forecasts for a given pattern. 
We have already described steps 1 and 2. 
In Section \ref{sec:dsfa} we describe step 3.
In Sections \ref{sec:vmm} - \ref{sec:embed} we present step 4, our proposed method for constructing a probabilistic model for a pattern, 
based on prediction suffix trees.
Steps 5 and 6 are described in Section \ref{sec:forecasts}.
After learning a model,
we first need to estimate the so-called \emph{waiting-time distributions} for each state of our automaton. 
Roughly speaking, these distributions let us know the probability of reaching a final state from any other automaton state in $k$ events from now.
These distributions are then used to estimate forecasts,
which generally have the form of an interval within which a CE has a high probability of occurring.
Finally, Section \ref{sec:no-mc} discusses an optimization that allows us to bypass the explicit construction of the Markov chain 
and Section \ref{sec:complexity} presents a full complexity analysis.

\begin{figure}[t]
	\centering
	\includegraphics[width=0.99\textwidth]{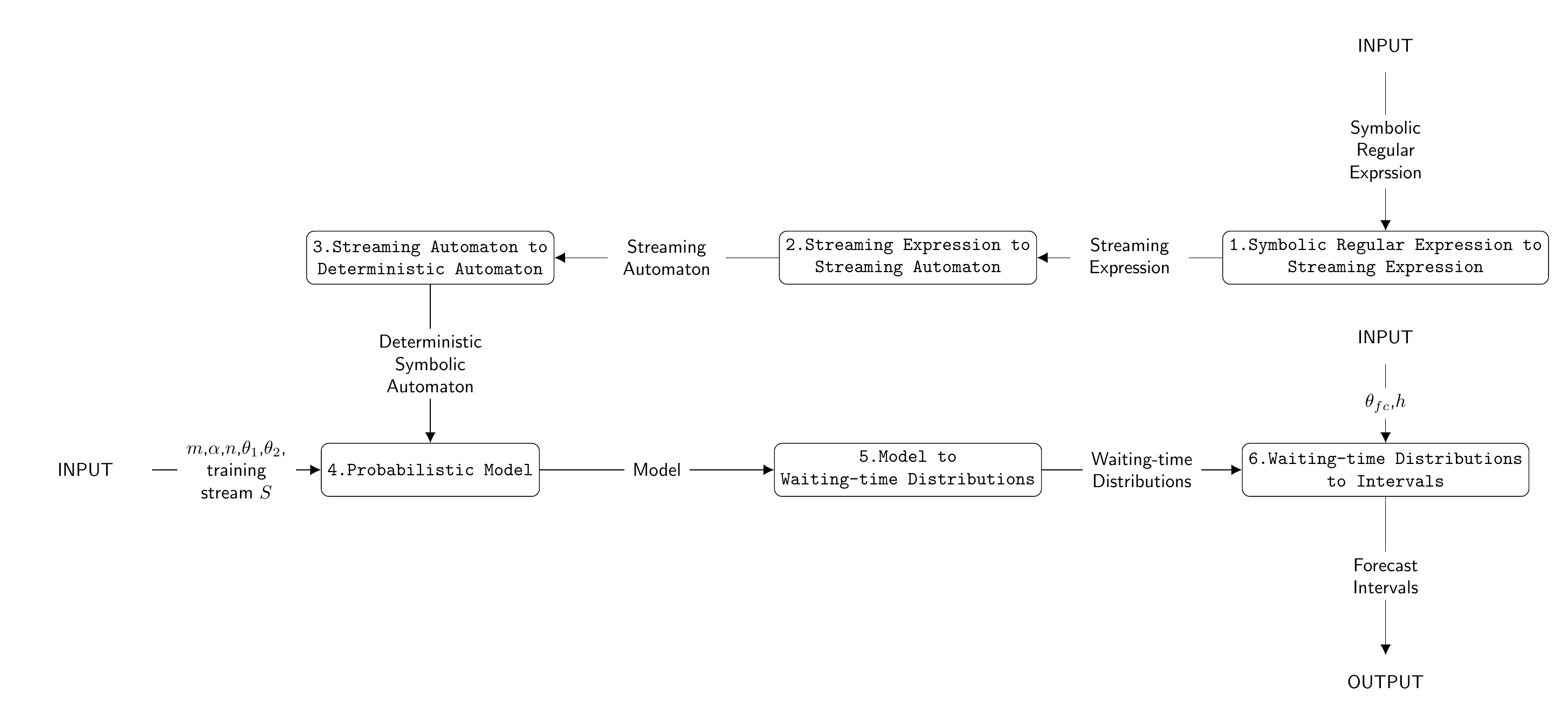}
	\caption{Steps for calculating the forecast intervals of a given pattern.}
	\label{fig:flow}
\end{figure}

\subsection{Deterministic Symbolic Automata}
\label{sec:dsfa}

The definition of \dsfa\ is similar to that of classical deterministic automata. 
Intuitively, we require that, for every state and every tuple/character, 
the \sfa\ can move to at most one next state upon reading that tuple/character.
We note though that it is not enough to require that all outgoing transitions from a state have different predicates as guards.
Symbolic automata differ from classical in one important aspect.
For the latter,
if we start from a given state and we have two outgoing transitions with different labels,
then it is not possible for both of these transition to be triggered simultaneously (i.e., with the same character).
For symbolic automata,
on the other hand,
two predicates may be different but still both evaluate to \true\ for the same tuple and thus two transitions with different predicates may both be triggered with the same tuple.
Therefore, the formal definition for \dsfa\ must take this into account: 
\begin{definition}[Deterministic SFA \cite{DBLP:conf/cav/DAntoniV17}]
A \sfa\ $M$ is deterministic if, for all transitions $(q,\psi_{1},q_{1}),(q,\psi_{2},q_{2}) \in M.\Delta$,
if $q_{1} \neq q_{2}$ then $\llbracket \psi_{1} \wedge \psi_{2} \rrbracket = \emptyset$.
\end{definition}
Using this definition for \dsfa\ it can be proven that \sfa\ are indeed closed under determinization \cite{DBLP:conf/cav/DAntoniV17}.
The determinization process first needs to create the \emph{minterms} of the predicates of a \sfa\ $M$, 
i.e., the set of maximal satisfiable Boolean combinations of such predicates,
denoted by $\mathit{Minterms}(\mathit{Predicates}(M))$,
and then use these minterms as guards for the \dsfa\ \cite{DBLP:conf/cav/DAntoniV17}.

There are two factors that can lead to a combinatorial explosion of the number of states of the resulting \dsfa:
first, the fact that the powerset of the states of the original \sfa\ must be constructed (similarly to classical automata);
second, the fact that the number of minterms (and, thus, outgoing transitions from each \dsfa\ state) is an exponential function of the number of the original \sfa\ predicates. 
In order to mitigate this doubly exponential cost, 
we follow two simple optimization techniques.
As is typically done with classical automata as well, 
instead of constructing the powerset of states of the \sfa\ and then adding transitions,
we construct the states of the \dsfa\ incrementally, starting from its initial state,
without adding states that will be inaccessible in the final \dsfa.
We can also reduce the number of minterms by taking advantage of some previous knowledge about some of the predicates that we might have.
In cases where we know that some of the predicates are mutually exclusive,
i.e., 
at most one of them can evaluate to \true,
then we can both discard some minterms and simplify some others. 
For example, 
if we have two predicates,
$\psi_{A} := \mathit{speed} < 5$ and $\psi_{B} := \mathit{speed} >20$,
then we also know that $\psi_{A}$ and $\psi_{B}$ are mutually exclusive.
As a result,
we can simplify the minterms,
as shown in Table \ref{table:minterms_simplified}.
\begin{table*}[t]
\centering
\caption{The set of simplified minterms for the predicates $\psi_{A} := \mathit{speed} < 5$ and $\psi_{B} := \mathit{speed} > 20$.}
\begin{tabular}{ccc}
\toprule
Original & Simplified & Reason\\
\midrule
$\psi_{A} \wedge \psi_{B}$ & discard & unsatisfiable \\
\midrule
$\psi_{A} \wedge \neg \psi_{B}$  & $\psi_{A}$ & $\psi_{A} \vDash \neg \psi_{B}$ \\
\midrule
$\neg \psi_{A} \wedge \psi_{B}$ & $\psi_{B}$ & $\psi_{B} \vDash \neg \psi_{B}$ \\
\midrule
$\neg \psi_{A} \wedge \neg \psi_{B}$ & $\neg \psi_{A} \wedge \neg \psi_{B}$ & for events whose speed is between 5 and 20 \\
\bottomrule
\end{tabular}
\label{table:minterms_simplified}
\end{table*}

Before moving to the discussion about how a \dsfa\ can be converted to a Markov chain,
we present a useful lemma.
We will show that a \dsfa\ always has an equivalent (through an isomorphism) deterministic classical automaton.
This result is important for two reasons:
a) it allows us to use methods developed for classical automata without having to always prove that they are indeed applicable to symbolic automata as well, and
b) it will help us in simplifying our notation, 
since we can use the standard notation of symbols instead of predicates. 

First note that $\mathit{Minterms}(\mathit{Predicates}(M))$ induces a finite set of equivalence classes on the (possibly infinite) set of domain elements of $M$ \cite{DBLP:conf/cav/DAntoniV17}. 
For example, if $\mathit{Predicates}(M)=\{\psi_{1},\psi_{2}\}$,
then $\mathit{Minterms}(\mathit{Predicates}(M))=\{\psi_{1} \wedge \psi_{2}, \psi_{1} \wedge \neg \psi_{2}, \neg \psi_{1} \wedge \psi_{2}, \neg \psi_{1} \wedge \neg \psi_{2}\}$, and we can map each domain element, which, in our case, is a tuple, to exactly one of these 4 minterms:
the one that evaluates to \true\ when applied to the element.
Similarly, the set of minterms induces a set of equivalence classes on the set of strings (event streams in our case).
For example, if $S{=}t_{1},\cdots,t_{k}$ is an event stream, 
then it could be mapped to $S'{=}a,\cdots,b$,
with $a$ corresponding to $\psi_{1} \wedge \neg \psi_{2}$ if $\psi_{1}(t_{1}) \wedge \neg \psi_{2}(t_{1}) = \true$, $b$ to $\psi_{1} \wedge \psi_{2}$, etc.
\begin{definition}[Stream induced by the minterms of a \dsfa]
If $S$ is a stream from the domain elements of the algebra of a \dsfa\ $M$ and $N=\mathit{Minterms}(\mathit{Predicates}(M))$,
then the stream $S'$ induced by applying $N$ on $S$ is the equivalence class of $S$ induced by $N$.
\end{definition}

We can now prove the lemma, 
which states that for every \dsfa\ there exists an equivalent classical deterministic automaton. 
\begin{lemma}
\label{lemma:isomorphism}
For every deterministic symbolic finite automaton (\dsfa) $M_{s}$ there exists a deterministic classical finite automaton (\dfa) $M_{c}$ such that $\mathcal{L}(M_{c})$ is the set of strings induced by applying $N=\mathit{Minterms}(\mathit{Predicates}(M_{s}))$ to $\mathcal{L}(M_{s})$.
\end{lemma}
\begin{proof}
From an algebraic point of view,
the set $N=\mathit{Minterms}(\mathit{Predicates}(M))$ may be treated as a generator of the monoid $N^{*}$, 
with concatenation as the operation.
If the cardinality of $N$ is $k$,
then we can always find a set $\Sigma=\{a_{1},\cdots,a_{k}\}$ of $k$ distinct symbols and then a morphism (in fact, an isomorphism) $\phi: N^{*} \rightarrow \Sigma^{*}$ that maps each minterm to exactly one, unique $a_{i}$.
A classical \dfa\ $M_{c}$ can then be constructed by relabeling the \dsfa\ $M_{s}$ under $\phi$,
i.e.,
by copying/renaming the states and transitions of the original \dsfa\ $M_{s}$ and by replacing the label of each transition of $M_{s}$ by the image of this label under $\phi$.
Then, the behavior of $M_{c}$ (the language it accepts) is the image under $\phi$ of the behavior of $M_{s}$ \cite{DBLP:books/daglib/0023547}.
Or, equivalently, the language of $M_{c}$ is the set of strings induced by applying $N=\mathit{Minterms}(\mathit{Predicates}(M_{s}))$ to $\mathcal{L}(M_{s})$.
\end{proof}
A direct consequence drawn from the proof of the above lemma is that, 
for every run 
$\varrho = [1,q_{1}] \overset{\delta_{1}}{\rightarrow} [2,q_{2}] \overset{\delta_{2}}{\rightarrow} \cdots \overset{\delta_{k}}{\rightarrow} [k+1,q_{k+1}]$ 
followed by a \dsfa\ $M_{s}$ by consuming a symbolic string (stream of tuples) $S$,
the run that the equivalent \dfa\ $M_{c}$ follows by consuming the induced string $S'$ is also
$\varrho' = [1,q_{1}] \overset{\delta_{1}}{\rightarrow} [2,q_{2}] \overset{\delta_{2}}{\rightarrow} \cdots \overset{\delta_{k}}{\rightarrow} [k+1,q_{k+1}]$,
i.e., $M_{c}$ follows the same copied/renamed states and the same copied/relabeled transitions.

This direct relationship between \dsfa\ and classical \dfa\ allows us to transfer techniques developed for classical \dfa\ to the study of \dsfa. 
Moreover, we can simplify our notation by employing the terminology of symbols/characters and strings/words that is typical for classical automata. 
Henceforth, 
we will be using symbols and strings as in classical theories of automata and strings (simple lowercase letters to denote symbols),
but the reader should bear in mind that,
in our case,
each symbol always corresponds to a predicate and,
more precisely,
to a minterm of a \dsfa. 
For example, the symbol $a$ may actually refer to the minterm $(\mathit{speed} < 5) \wedge (\mathit{speed} >20)$,
the symbol $b$ to $(\mathit{speed} < 5) \wedge \neg (\mathit{speed} >20)$, etc.

\subsection{Variable-order Markov Models}
\label{sec:vmm}

Assuming that we have a deterministic automaton, 
the next question is how we can build a probabilistic model that captures the statistical properties of the streams to be processed by this automaton. 
With such a model,
we could then make inferences about the automaton's expected behavior as it reads event streams.
One approach would be to map each state of the automaton to a state of a Markov chain,
then apply the automaton on a training stream of symbols,
count the number of transitions from each state to every other target state and use these counts to calculate the transition probabilities.
This is the approach followed in \cite{DBLP:conf/debs/MuthusamyLJ10}.

However, there is an important issue with the way in which this approach models transition probabilities. 
Namely, a probability is attached to the transition between two states, 
say state 1 and state 2, 
ignoring the way in which state 1 has been reached, i.e., 
failing to capture the sequence of symbols.
For example, in Figure \ref{fig:dfasr}, 
state $0$ can be reached after observing symbol $b$ or symbol $c$. 
The outgoing transition probabilities do not distinguish between the two cases. 
Instead, they just capture the probability of $a$ given that the previous symbol was $b$ or $c$. 
This introduces ambiguity and if there are many such states in the automaton, 
we may end up with a Markov chain that is first-order (with respect to its states),
but nevertheless provides no memory of the stream itself.
It may be unable to capture first-order (or higher order) dependencies in the stream of events. In the worst case (if every state can be reached with any symbol), such a Markov chain may essentially assume that the stream is composed of i.i.d. events.
\begin{figure}[t]
	\centering
	\includegraphics[width=0.7\textwidth]{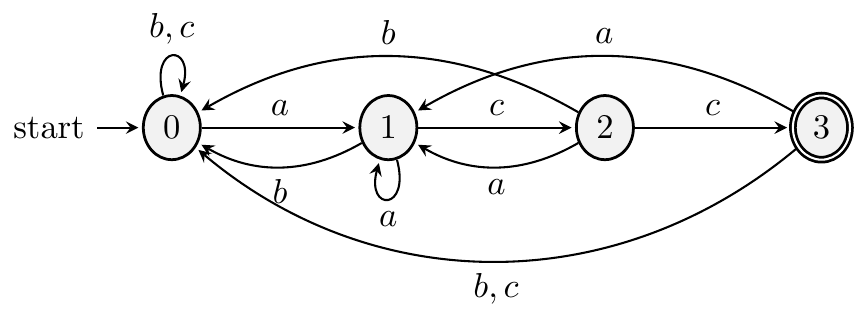}
	\caption{A classical automaton for the expression $R := a \cdot c \cdot c$ with alphabet $\Sigma=\{a,b,c\}$. State 1 can always remember the last symbol seen, since it can be reached only with $a$. State 0 can be reached with $b$ or $c$.}
	\label{fig:dfasr}
\end{figure}

An alternative approach,
followed in \cite{DBLP:conf/lpar/AlevizosAP18,DBLP:conf/debs/AlevizosAP17},
is to first set a maximum order $m$ that we need to capture
and then iteratively split each state of the original automaton into as many states as required so that each new state can remember the past $m$ symbols that have led to it. 
The new automaton that results from this splitting process is equivalent to the original,
in the sense that they recognize the same language,
but can always remember the last $m$ symbols of the stream.
With this approach,
it is indeed possible to guarantee that $m$-order dependencies can be captured.
As expected though, higher values of $m$ can quickly lead to an exponential growth of the number of states and the approach may be practical only for low values of $m$.

We propose the use of a variable-order Markov model (VMM) to mitigate the high cost of increasing the order $m$
\cite{DBLP:journals/jair/BegleiterEY04,buhlmann1999variable,DBLP:journals/ml/RonST96,DBLP:conf/nips/RonST93,DBLP:journals/tcom/ClearyW84,DBLP:journals/tit/WillemsST95}.
This allows us to increase $m$ to values not possible with the previous approaches and thus capture longer-term dependencies,
which can lead to a better accuracy.
An alternative would be to use hidden Markov models (HMMs) \cite{rabiner1989tutorial},
which are generally more expressive than bounded-order (either full or variable) Markov models.
However, HMMs often require large training datasets \cite{DBLP:journals/jair/BegleiterEY04,DBLP:journals/ml/AbeW92}.
Another problem is that it is not always obvious how a domain can be modeled through HMMs and a deep understanding of the domain may be required \cite{DBLP:journals/jair/BegleiterEY04}.
The relation between an automaton and the observed state of a HMM is not straightforward 
and it is not evident how a HMM would capture an automaton's behavior.

Different Markov models of variable order have been proposed in the literature
(see \cite{DBLP:journals/jair/BegleiterEY04} for a nice comparative study).
The general approach of such models is as follows:
let $\Sigma$ denote an alphabet, 
$\sigma \in \Sigma$ a symbol from that alphabet and $s \in \Sigma^{m}$ a string of length $m$ of symbols from that alphabet.
The aim is to derive a predictor $\hat{P}$ from the training data such that the average log-loss on a test sequence $S_{1..k}$ is minimized.
The loss is
given by 
$l(\hat{P},S_{1..k}) = - \frac{1}{T} \sum_{i=1}^{k} log \hat{P}(t_{i} \mid t_{1} \cdots t_{i-1})$.
Minimizing the log-loss is equivalent to maximizing the likelihood $\hat{P}(S_{1..k})=\prod_{i=1}^{k}\hat{P}(t_{i} \mid t_{1} \dots t_{i-1})$.
The average log-loss may also be viewed as a measure of the average compression rate achieved on the test sequence \cite{DBLP:journals/jair/BegleiterEY04}.
The mean (or expected) log-loss ($-\boldsymbol{E}_{P}\{log \hat{P}(S_{1..k}) \}$) is minimized if the derived predictor $\hat{P}$ is indeed the actual distribution $P$ of the source emitting sequences.

For full-order Markov models, 
the predictor $\hat{P}$ is derived through the estimation of conditional distributions $\hat{P}(\sigma \mid s)$,
with $m$ constant and equal to the assumed order of the Markov model. 
Variable-order Markov Models (VMMs), on the other hand, relax the assumption of $m$ being fixed.
The length of the ``context'' $s$ (as is usually called) may vary,
up to a \emph{maximum} order $m$, 
according to the statistics of the training dataset.  
By looking deeper into the past only when it is statistically meaningful,
VMMs can capture both short- and long-term dependencies.

\subsection{Prediction Suffix Trees}
\label{sec:pst}

We use Prediction Suffix Trees (\pst),
as described in \cite{DBLP:journals/ml/RonST96,DBLP:conf/nips/RonST93},
as our VMM of choice.
The reason is that,
once a \pst\ has been learned,
it can be readily converted to a probabilistic automaton.
More precisely, 
we learn a probabilistic suffix automaton (\psa),
whose states correspond to contexts of variable length.
The outgoing transitions from each state of the \psa\ encode the conditional distribution of seeing a symbol given the context of that state. 
As we will show,
this probabilistic automaton (or the tree itself)
can then be combined with a symbolic automaton in a way that allows us to infer when a CE is expected to occur.

The formal definition of a PST is the following:
\begin{definition}[Prediction Suffix Tree \cite{DBLP:journals/ml/RonST96}]
Let $\Sigma$ be an alphabet. 
A PST $T$ over $\Sigma$ is a tree
whose edges are labeled by symbols $\sigma \in \Sigma$ and each internal node has exactly one edge for every $\sigma \in \Sigma$
(hence, the degree is $\mid \Sigma \mid$).
Each node is labeled by a pair $(s,\gamma_{s})$,
where $s$ is the string associated with the walk starting from that node and ending at the root, 
and $\gamma_{s}: \Sigma \rightarrow [0,1]$ is the next symbol probability function related with $s$. 
For every string $s$ labeling a node, $\sum_{\sigma \in \Sigma} \gamma_{s}(\sigma) = 1$.
The depth of the tree is its order $m$.
\end{definition}
Figure \ref{fig:pstab1} shows an example of a \pst\ of order $m=2$.
According to this tree,
if the last symbol that we have encountered in a stream is $a$ 
and we ignore any other symbols that may have preceded it,
then the probability of the next input symbol being again $a$ is $0.7$.
However, we can obtain a better estimate of the next symbol probability by extending the context and looking one more symbol deeper into the past.
Thus, if the last two symbols encountered are $b,a$,
then the probability of seeing $a$ again is very different ($0.1$).
On the other hand,
if the last symbol encountered  is $b$,
the next symbol probability distribution is $(0.5,0.5)$ and,
since the node $b,(0.5,0.5)$ has not been expanded,
this implies that its children would have the same distribution if they had been created.
Therefore, the past does not affect the prediction and will not be used.
Note that a \pst\ whose leaves are all of equal depth $m$ corresponds to a full-order Markov model of order $m$,
as its paths from the root to the leaves correspond to every possible context of length $m$.

Our goal is to incrementally learn a \pst\ $\hat{T}$ by adding new nodes only when it is necessary
and then use $\hat{T}$ to construct a \psa\ $\hat{M}$ that will approximate the actual \psa\ $M$ that has generated the training data.
Assuming that we have derived an initial predictor $\hat{P}$
(as described in more detail in Section \ref{sec:prob_empirical}),
the learning algorithm in \cite{DBLP:journals/ml/RonST96} starts with a tree having only a single node,
corresponding to the empty string $\epsilon$.
Then, 
it decides whether to add a new context/node $s$ by checking two conditions:
\begin{itemize}
	\item First, there must exist $\sigma \in \Sigma$ such that $\hat{P}(\sigma \mid s) > \theta_{1}$ must hold, i.e., $\sigma$ must appear ``often enough'' after the suffix $s$;
	\item Second, $\frac{\hat{P}(\sigma \mid s)}{\hat{P}(\sigma \mid \mathit{suffix}(s))} > \theta_{2}$ (or $\frac{\hat{P}(\sigma \mid s)}{\hat{P}(\sigma \mid \mathit{suffix}(s))} <  \frac{1}{\theta_{2}}$) must hold, 
	i.e., it is ``meaningful enough'' to expand to $s$ because there is a significant difference in the conditional probability of $\sigma$ given $s$ with respect to the same probability given the shorter context $\mathit{suffix}(s)$,
where $\mathit{suffix}(s)$ is the longest suffix of $s$ that is different from $s$.
\end{itemize}
The thresholds $\theta_{1}$ and $\theta_{2}$ depend, among others, 
on parameters $\alpha$, $n$ and $m$, 
$\alpha$ being an approximation parameter, 
measuring how close we want the estimated \psa\ $\hat{M}$ to be compared to the actual \psa\ $M$,
$n$ denoting the maximum number of states that we allow $\hat{M}$ to have
and $m$ denoting the maximum order/length of the dependencies we want to capture. 
For example,
consider node $a$ in Figure \ref{fig:pstab1} and assume that we are at a stage of the learning process where we have not yet added its children, $aa$ and $ba$.
We now want to check whether it is meaningful to add $ba$ as a node.
Assuming that the first condition is satisfied,
we can then check the ratio  $\frac{\hat{P}(\sigma \mid s)}{\hat{P}(\sigma \mid \mathit{suffix}(s))} = \frac{\hat{P}(a \mid ba)}{\hat{P}(a \mid a)} = \frac{0.1}{0.7} \approx 0.14$.
If $\theta_{2} = 1.05$, then $\frac{1}{\theta_{2}} \approx 0.95$ and the condition is satisfied,
leading to the addition of node $ba$ to the tree.
For more details, see \cite{DBLP:journals/ml/RonST96}.
\begin{figure}
\begin{subfigure}[t]{0.54\textwidth}
	\includegraphics[width=0.99\textwidth]{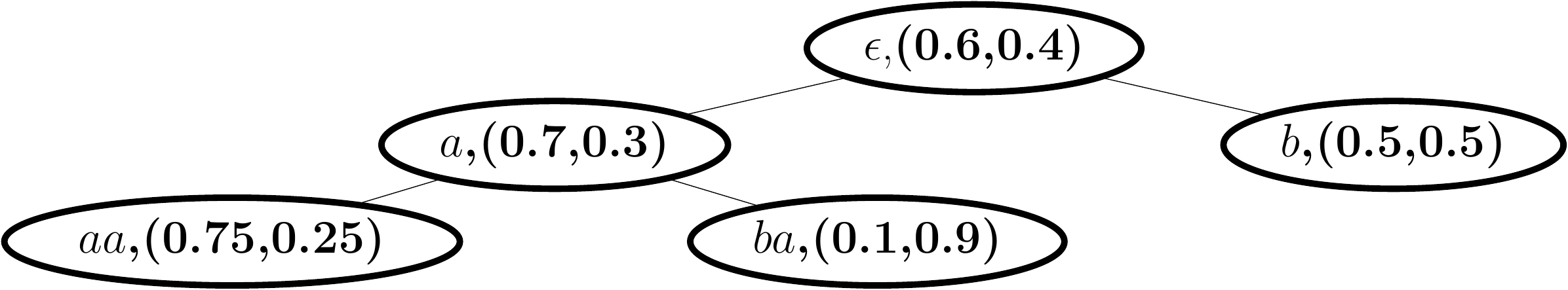}
	\caption{Example \pst\ $T$ for $\Sigma = \{a,b\}$ and $m=2$.
	Each node contains the label and the next symbol probability distribution for $a$ and $b$.}
	\label{fig:pstab1}
\end{subfigure}
\begin{subfigure}[t]{0.4\textwidth}
	\includegraphics[width=0.99\textwidth]{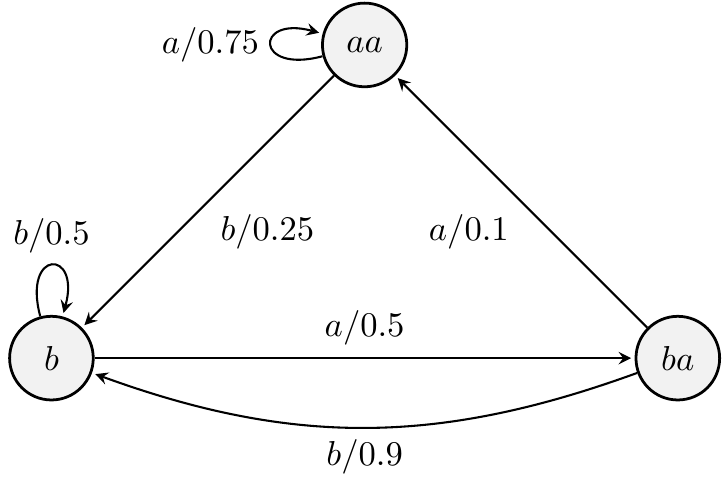}
	\caption{Example \psa\ $M_{S}$ constructed from the tree $T$. 
	Each state contains its label. 
	Each transition is composed of the next symbol to be encountered along with that symbol's probability.}
	\label{fig:psaab1}
\end{subfigure}
\caption{Example of a prediction suffix tree and its corresponding probabilistic suffix automaton.}
\label{fig:pstpsa}
\end{figure}

Once a \pst\ $\hat{T}$ has been learned,
we can convert it to a \psa\ $\hat{M}$.
The definition for \psa\ is the following:
\begin{definition}[Probabilistic Suffix Automaton \cite{DBLP:journals/ml/RonST96}]
\label{def:psa}
A Probabilistic Suffix Automaton $M$ is a tuple $(Q,\Sigma,\tau,\gamma,\pi)$, where:
\begin{itemize}
	\item $Q$ is a finite set of states;
	\item $\Sigma$ is a finite alphabet;
	\item $\tau: Q \times \Sigma \rightarrow Q$ is the transition function;
	\item $\gamma: Q \times \Sigma \rightarrow [0,1]$ is the next symbol probability function;
	\item $\pi: Q \rightarrow [0,1]$ is the initial probability distribution over the starting states; 
\end{itemize}
The following conditions must hold: 
\begin{itemize}
	\item For every $q \in Q$, it must hold that $\sum_{\sigma \in \Sigma} \gamma(q,\sigma) = 1$ and $\sum_{q \in Q} \pi(q) = 1$;
	\item Each $q \in Q$ is labeled by a string $s \in \Sigma^{*}$ and the set of labels is suffix free,
	i.e., no label $s$ is a suffix of another label $s'$;
	\item For every two states $q_{1},q_{2} \in Q$ and for every symbol $\sigma \in \Sigma$, 
	if $\tau(q_{1},\sigma)=q_{2}$ and $q_{1}$ is labeled by $s_{1}$, 
	then $q_{2}$ is labeled by $s_{2}$, such that $s_{2}$ is a suffix of $s_{1} \cdot \sigma$;
	\item For every $s$ labeling some state $q$, 
	and every symbol $\sigma$ for which $\gamma(q,\sigma) > 0$, 
	there exists a label which is a suffix of $s \cdot \sigma$;
	\item Finally, the graph of $M$ is strongly connected.
\end{itemize}
\end{definition}
Note that a \psa\ is a Markov chain.
$\tau$ and $\gamma$ can be combined into a single function, 
ignoring the symbols,
and this function, 
together with the first condition of Definition \ref{def:psa}, 
would define the transition matrix of a Markov chain.
The last condition about $M$ being strongly connected also ensures that the Markov chain is
composed of a single recurrent class of states.
Figure \ref{fig:psaab1} shows an example of a \psa,
the one that we construct from the \pst\ of Figure \ref{fig:pstab1},
using the leaves of the tree as automaton states.
A full-order \psa\ for $m=2$ would require a total of 4 states,
given that we have two symbols.
If we use the \pst\ of Figure \ref{fig:pstab1},
we can construct the \psa\ of Figure \ref{fig:psaab1} which has 3 states.
State $b$ does not need to be expanded to states $bb$ and $ab$,
since the tree tells us that such an expansion is not statistically meaningful. 

Using a \psa\,
we can process a stream of symbols and at every point be able to provide an estimate about the next symbols that will be encountered along with their probabilities.
The state of the \psa\ at every moment corresponds to a suffix of the stream.
For example, 
according to the \psa\ of Figure \ref{fig:psaab1},
if the last symbol consumed from the stream is $b$,
then the \psa\ would be in state $b$ and the probability of the next symbol being $a$ would be $0.5$.
If the last symbol in the stream is $a$,
we would need to expand this suffix to look at one more symbol in the past.
If the last two symbols are $aa$,
then the \psa\ would be in state $aa$ and the probability of the next symbol being $a$ again would be $0.75$.

Note that a \psa\ does not act as an acceptor (there are no final states),
but can act as a generator of strings.
It can use $\pi$, its initial distribution on states, to select an initial state and generate its label as a first string
and then continuously use $\gamma$ to generate a symbol, move to a next state and repeat the same process.
At every time, the label of its state is always a suffix of the string generated thus far. 
A \psa\ may also be used to read a string or stream of symbols.
In this mode,
the state of the \psa\ at every moment corresponds again to a suffix of the stream and the \psa\ can be used to calculate the probability of seeing any given string in the future, 
given the label of its current state. 
Our intention is to use this derived \psa\ to process streams of symbols,
so that, 
while consuming a stream $S_{1..k}$,
we can know what its meaningful suffix and use that suffix for any inferences.

However, there is a subtle technical issue about the convertibility of a \pst\ to a \psa.
Not every \pst\ can be converted to a \psa\ (but every \pst\ can be converted to a larger class of so-call probabilistic automata).
This is achievable under a certain condition. 
If this condition does not hold,
then the \pst\ can be converted to an automaton that is composed of a \psa\ as usual,
with the addition of some extra states.
These states, viewed as states of a Markov chain, are transient.
This means that the automaton will move through these states for some transitions, 
but it will finally end into the states of the \psa,
stay in that class and never return to any of the transient states.
In fact, if the automaton starts in any of the transient states,
then it will enter the single, recurrent class of the \psa\ in at most $m$ transitions.
Given the fact that in our work we deal with streams of infinite length,
it is certain that,
while reading a stream,
the automaton will have entered the \psa\ after at most $m$ symbols.
Thus, instead of checking this condition,
we prefer to simply construct only the \psa\ and wait (for at most $m$ symbols)
until the first $k \leq m$ symbols of a stream have been consumed and are equal to a label of the \psa.
At this point, we set the current state of the \psa\ to the state with that label and start processing.

The above discussion seems to suggest that a \psa\ is constructed from the leaves of a \pst.
Thus, it should be expected that the number of states of a \psa\ should always be smaller than the total number of nodes of its \pst.
However, this is not true in the general case.
In fact, in some cases the \pst\ nodes might be significantly less than the \psa\ states.
The reason is that a \pst,
as is produced by the learning algorithm described previously,
might not be sufficient to construct a \psa.
To remedy this situation,
we need to expand the original \pst\ $\hat{T}$ by adding more nodes in order to get a suitable \pst\ $\hat{T}'$ and then construct the \psa\ from $\hat{T}'$.
The leaves of $\hat{T}'$ (and thus the states of the \psa) could be significantly more than the leaves of $\hat{T}$.
This issue is further discussed in Section \ref{sec:no-mc}.

\subsection{Emitting Forecasts}
\label{sec:forecasts}

Our ultimate goal is to use the statistical properties of a stream,
as encoded in a \pst\ or a \psa,
in order to infer when a Complex Event (CE) with a given Symbolic Regular Expression (\sre) $R$ will be detected.
Equivalently, we are interested in inferring when the \sfa\ of $R$ will reach one of its final states. 
To achieve this goal, we work as follows.
We start with a \sre\ $R$ and a training stream $S$.
We first use $R$ to construct an equivalent \ssfa\ and then determinize this \ssfa\ into a \dsfa\ $M_{R}$.
$M_{R}$ can be used to perform recognition on any given stream,
but cannot be used for probabilistic inference.
Next, we use the minterms of $M_{R}$
(acting as ``symbols'', see Lemma \ref{lemma:isomorphism})
and the training stream $S$ to learn a \pst\ $T$ and (if required) a \psa\ $M_{S}$ which encode the statistical properties of $S$.
These probabilistic models do not yet have any knowledge of the structure of $R$ (they only know its minterms), 
are not acceptors (the \psa\ does not have any final states) and cannot be used for recognition.
We therefore need to combine the learned probabilistic model ($T$ or $M_{S}$) with the automaton used for recognition ($M_{R}$).

At this point,
there is a trade-off between memory and computation efficiency.
If the online performance of our system is critical and we are not willing to make significant sacrifices in terms of computation efficiency,
then we should combine the recognition automaton $M_{R}$ with the \psa\ $M_{S}$.
Using the \psa\, we can have a very efficient solution with minimal overhead on throughput.
The downside of this approach is its memory footprint, which may limit the order of the model. 
Although we may increase the order beyond what is possible with full-order models, 
we may still not achieve the desired values,
due to the significant memory requirements.
Hence, 
if high accuracy and thus high order values are necessary,
then we should combine the recognition automaton $M_{R}$ directly with the \pst\ $T$,
bypassing the construction of the \psa.
In practice prediction suffix trees often turn out to be more compact and memory efficient than probabilistic suffix automata,
but trees need to be constantly traversed from root to leaves whereas an automaton simply needs to find the triggered transition and immediately jump to the next state.
In the remainder of this Section,
we present these two alternatives.

\subsubsection{Using a Probabilistic Suffix Automaton (\psa)}
\label{sec:embed}

We can combine a recognition automaton $M_{R}$ and a \psa\ $M_{S}$ into a single automaton $M$ that has the power of both and can be used for recognizing and for forecasting occurrences of CEs of the expression $R$. 
We call $M$ the \emph{embedding} of $M_{S}$ in $M_{R}$.
The reason for merging the two automata is that we need to know at every point in time the state of $M_{R}$ in order to estimate which future paths might actually lead to a final state (and thus a complex event).
If only SDE forecasting was required, 
this merging would not be necessary.
We could use $M_{R}$ for recognition and then $M_{S}$ for SDE forecasting. 
In our case,
we need information about the structure of the pattern automaton and its current state to determine if and when it might reach a final state.
The formal definition of an embedding is given below,
where,
in order to simplify notation,
we use Lemma \ref{lemma:isomorphism} and represent \dsfa\ as classical deterministic automata.
\begin{definition}[Embedding of a \psa\ in a \dsfa]
Let $M_{R}$ be a \dsfa\ (actually its mapping to a classical automaton) and $M_{S}$ a \psa\ with the same alphabet.
An embedding of $M_{S}$ in $M_{R}$ is a tuple $M=(Q,Q^{s},Q^{f},\Sigma,\Delta,\Gamma,\pi)$, where:
\begin{itemize}
	\item $Q$ is a finite set of states;
	\item $Q^{s} \subseteq Q$ is the set of initial states;
	\item $Q^{f} \subseteq Q$ is the set of final states;
	\item $\Sigma$ is a finite alphabet;
	\item $\Delta: Q \times \Sigma \rightarrow Q$ is the transition function;
	\item $\Gamma: Q \times \Sigma \rightarrow [0,1]$ is the next symbol probability function;
	\item $\pi: Q \rightarrow [0,1]$ is the initial probability distribution.
\end{itemize}
The language $\mathcal{L}(M)$ of $M$ is defined, as usual, as the set of strings that lead $M$ to a final state.
The following conditions must hold,
in order for $M$ to be an embedding of $M_{S}$ in $M_{R}$:
\begin{itemize}
	\item $\Sigma = M_{R}.\Sigma = M_{S}.\Sigma$;
	\item $\mathcal{L}(M) = \mathcal{L}(M_{R})$;
	\item For every string/stream $S_{1..k}$, $P_{M}(S_{1..k}) = P_{M_{S}}(S_{1..k})$,
where $P_{M}$ denotes the probability of a string calculated by $M$ (through $\Gamma$) and $P_{M_{S}}$ the probability calculated by $M_{S}$ (through $\gamma$). 
\end{itemize}
\end{definition} 
The first condition ensures that all automata have the same alphabet.
The second ensures that $M$ is equivalent to $M_{R}$ by having the same language.
The third ensures that $M$ is also equivalent to $M_{S}$,
since both automata return the same probability for every string.

It can be shown that such an equivalent embedding can indeed be constructed for every \dsfa\ and \psa.
\begin{theorem}
\label{theorem:embedding}
For every \dsfa\ $M_{R}$ and \psa\ $M_{S}$ constructed using the minterms of $M_{R}$,
there exists an embedding of $M_{S}$ in $M_{R}$.
\end{theorem}
\begin{proof}
Construct an embedding in the following straightforward manner:
First let its states be the Cartesian product $M_{R}.Q \times M_{S}.Q$,
i.e., for every $q \in Q$, $q = (r,s)$ and $r \in M_{R}.Q$, $s \in M_{S}.Q$.
Set the initial states of $M$ as follows:
for every $q = (r,s)$ such that $r = M_{R}.q^{s}$, set $q \in Q^{s}$.
Similarly, for the final states,
for every $q = (r,s)$ such that $r \in M_{R}.Q^{f}$, set $q \in Q^{f}$.
Then let the transitions of $M$ be defined as follows:
A transition $\delta((r,s), \sigma) = (r^{'},s^{'})$ is added to $M$
if there exists a transition $\delta_{R}(r, \sigma) = r^{'}$ in $M_{R}$
and a transition $\tau(s, \sigma) = s^{'}$ in $M_{S}$.
Let  also $\Gamma$ be defined as follows: $\Gamma((r,s), \sigma) = \gamma(s,\sigma)$.
Finally, for the initial state distribution, we set:
\[ \pi((r,s)) =
  \begin{cases}
    M_{S}.\pi(s) & \quad \text{if } r = M_{R}.q^{s}   \\
    0  & \quad \text{otherwise} \\
  \end{cases}
\]

Proving that $\mathcal{L}(M) = \mathcal{L}(M_{R})$ is done with induction on the length of strings.
The inductive hypothesis is that, for strings $S_{1..k} = t_{1} \cdots t_{k}$ of length $k$,
if $q=(r,s)$ is the state reached by $M$ and $q_{R}$ the state reached by $M_{R}$,
then $r=q_{R}$.
Note that both $M_{R}$ and $M$ are deterministic and complete automata and thus only one state is reached for every string (only one run exists).
If a new element $t_{k+1}$ is read, $M$ will move to a new state $q^{'}=(r^{'},s^{'})$ and $M_{R}$ to $q_{R}^{'}$.
From the construction of the transitions of $M$, we see that $r^{'} = q_{R}^{'}$.
Thus, the induction hypothesis holds for $S_{1..k+1}$ as well. 
It also holds for $k=0$,
since, 
for every $q = (r,s) \in Q^{s}$, $r = M_{R}.q^{s}$.
Therefore, it holds for all $k$.
As a result, if $M$ reaches a final state $(r,s)$, 
$r$ is reached by $M_{R}$.
Since $r \in M_{R}.Q^{f}$, $M_{R}$ also reaches a final state.

For proving probabilistic equivalence,
first note that the probability of a string given by a predictor $P$ is 
$P(S_{1..k})=\prod_{i=1}^{k}P(t_{i} \mid t_{1} \dots t_{i-1})$.
Assume now that a \psa\ $M_{S}$ reads a string $S_{1..k}$ and follows a run 
$\varrho = [l,q_{l}] \overset{t_{l}}{\rightarrow} [l+1,q_{l+1}] \overset{t_{l+1}}{\rightarrow} \cdots \overset{t_{k}}{\rightarrow} [k+1,q_{k+1}]$.
We define a run in a manner similar to that for runs of a \dsfa.
The difference is that a run of a \psa\ may begin at an index $l>1$,
since it may have to wait for $l$ symbols before it can find a state $q_{l}$ whose label is equal to $S_{1..l}$.
We also treat the \psa\ as a reader (not a generator) of strings for which we need to calculate their probability.
The probability of $S_{1..k}$ is then given by 
$P_{M_{S}}(S_{1..k}) = M_{S}.\pi(q_{l}) \cdot \prod_{i=l}^{k} M_{S}.\gamma(q_{i},t_{i})$.
Similarly, for the embedding $M$,
assume it follows the run
$\varrho^{'} = [l,q^{'}_{l}] \overset{t_{l}}{\rightarrow} [l+1,q^{'}_{l+1}] \overset{t_{l+1}}{\rightarrow} \cdots \overset{t_{k}}{\rightarrow} [k+1,q^{'}_{k+1}]$.
Then,
$P_{M}(S_{1..k}) = M.\pi(q^{'}_{l}) \cdot \prod_{i=l}^{k} M.\Gamma(q^{'}_{i},t_{i})$.
Now note that $M$ has the same initial state distribution as $M_{S}$,
i.e.,
the number of the initial states of $M$ is equal to the number of states of $M_{S}$ and they have the same distribution.
With an inductive proof, as above, we can prove that whenever $M$ reaches a state $q=(r,s)$ and $M_{S}$ reaches $q_{S}$,
$s = q_{S}$.
As a result, for the initial states of $M$ and $M_{S}$,
$M.\pi(q^{'}_{l}) = M_{S}.\pi(q_{l})$.
From the construction of the embedding,
we also know that $M_{S}.\gamma(s,\sigma) = M.\Gamma(q,\sigma)$ for every $\sigma \in \Sigma$.
Therefore, 
$M_{S}.\gamma(q_{i},t_{i}) = M.\Gamma(q^{'}_{i},t_{i})$ for every $i$ and 
$P_{M}(S_{1..k}) = P_{M_{S}}(S_{1..k})$.
\end{proof}

\begin{figure}
\begin{subfigure}[t]{0.44\textwidth}
	\includegraphics[width=0.99\textwidth]{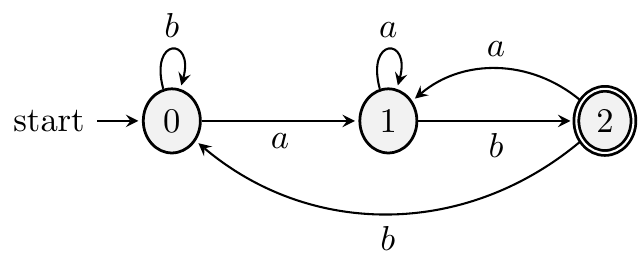}
	\caption{\dsfa\ $M_{R}$ for $R:=a \cdot b$ and $\Sigma = \{a,b\}$.}
	\label{fig:dsfaab}
\end{subfigure}
\begin{subfigure}[t]{0.5\textwidth}
	\includegraphics[width=0.99\textwidth]{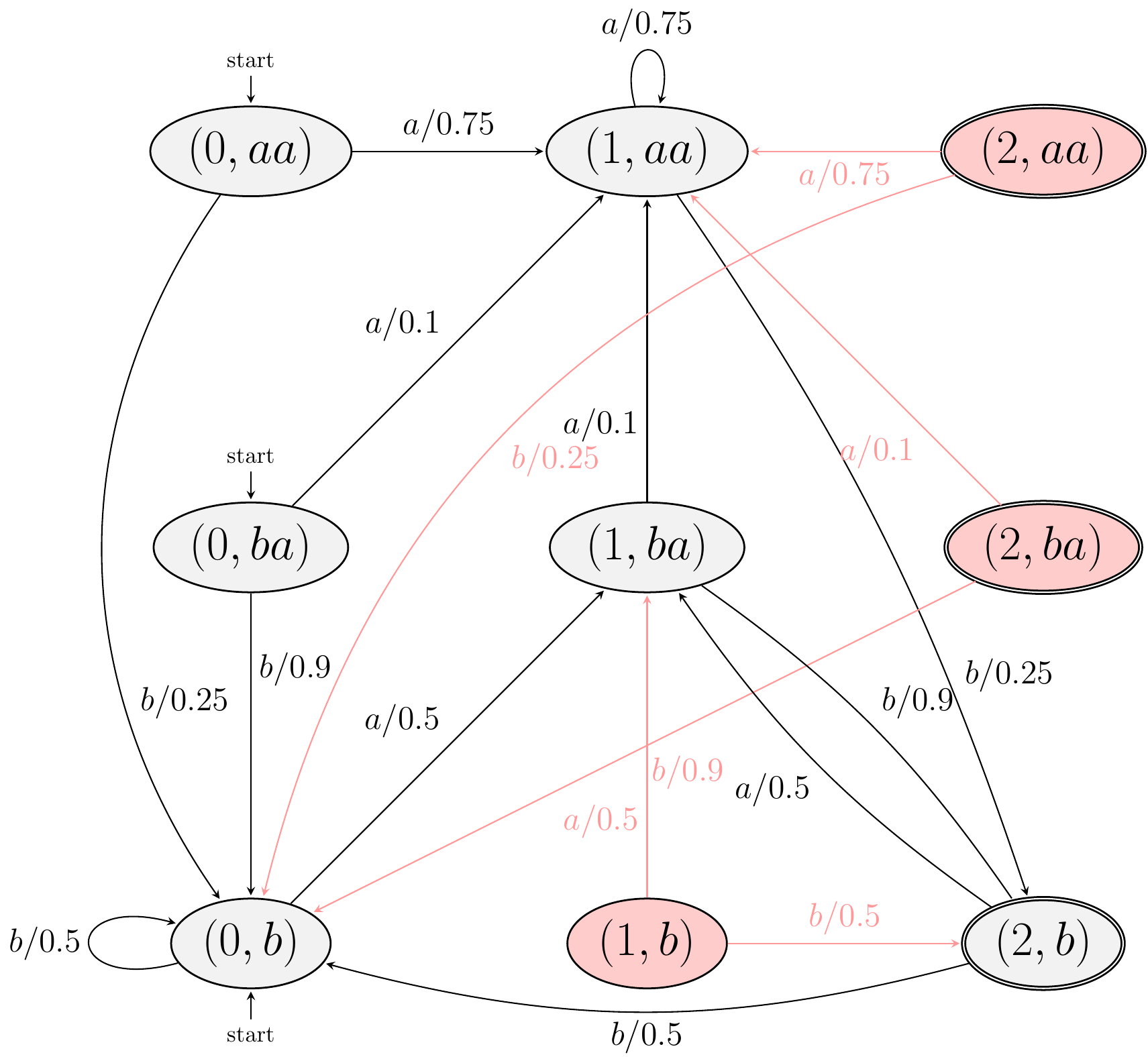}
	\caption{Embedding of $M_{S}$ of Figure \ref{fig:psaab1} in $M_{R}$ of Figure \ref{fig:dsfaab}.}
	\label{fig:merged}
\end{subfigure}
\caption{Embedding example.}
\label{fig:dsfapsa}
\end{figure}

As an example,
consider the \dsfa\ $M_{R}$ of Figure \ref{fig:dsfaab}
for the expression $R = a \cdot b$ with $\Sigma = \{a,b\}$.
We present it as a classical automaton, 
but we remind readers that symbols in $\Sigma$ correspond to minterms.
Figure \ref{fig:pstab1} depicts a possible \pst\ $T$ that could be learned from a training stream composed of symbols from $\Sigma$. 
Figure \ref{fig:psaab1} shows the \psa\ $M_{S}$ constructed from $T$.
Figure \ref{fig:merged} shows the embedding $M$ of $M_{S}$ in $M_{R}$ that would be created,
following the construction procedure of the proof of Theorem \ref{theorem:embedding}.
Notice, however, that this embedding has some redundant states and transitions;
namely the states indicated with red that have no incoming transitions and are thus inaccessible.
The reason is that some states of $M_{R}$ in Figure \ref{fig:dsfaab} have a ``memory'' imbued to them from the structure of the automaton itself.
For example, state 2 of $M_{R}$ has only a single incoming transition with $b$ as its symbol. 
Therefore, there is no point in merging this state with all the states of $M_{S}$,
but only with state $b$.
If we follow a straightforward construction,
as described above,
the result will be the automaton depicted in Figure \ref{fig:merged},
including the redundant red states.
To avoid the inclusion of such states,
we can merge $M_{R}$ and $M_{S}$ in an incremental fashion 
(see Algorithm \ref{algorithm:merging}).
The resulting automaton would then consist only of the black states and transitions of Figure \ref{fig:merged}.
In a streaming setting,
we would thus have to wait at the beginning of the stream for some input events to arrive before deciding the start state with which to begin.
For example,
if $b$ were the first input event,
we would then begin with the bottom left state $(0,b)$.
On the other hand, 
if $a$ were the first input event,
we would have to wait for yet another event.
If another $a$ arrived as the second event,
we would begin with the top left state $(0,aa)$.
In general, 
if $m$ is our maximum order,
we would need to wait for at most $m$ input events before deciding.

\begin{algorithm}[!h]
\KwIn{A \dsfa\ $M_{R}$ and a \psa\ $M_{S}$ learnt with the minterms of $M_{R}$}
\KwOut{An embedding $M$ of $M_{S}$ in $M_{R}$ equivalent to both $M_{R}$ and $M_{S}$}
\tcc{{\footnotesize First create the initial states of the merged automaton by combining the initial state of $M_{R}$ with all the states of $M_{S}$.}}\
$Q^{s} \leftarrow \emptyset$\;
\ForEach{$s \in M_{S}.Q$}{
	$q \leftarrow \mathit{CreateNewState}(M_{R}.q^{s},s)$\;
	\tcc{{\footnotesize $q$ is a tuple $(r,s)$}}\
	$Q^{s} \leftarrow Q^{s} \cup \{q\}$\;
}
\tcc{{\footnotesize A frontier of states is created, including states of $M$ that have no outgoing transitions yet. First frontier consists of the initial states of $M$}}\
$\mathit{Frontier} \leftarrow Q^{s}$;
$\mathit{Checked} \leftarrow \emptyset$;
$\Delta \leftarrow \emptyset$;
$\Gamma \leftarrow \emptyset$\;
\While{$\mathit{Frontier} \neq \emptyset$}{
	$q \leftarrow$ pick an element from $\mathit{Frontier}$\;
	\ForEach{$\sigma \in M_{S}.\Sigma$}{
		$s^{next} \leftarrow M_{S}.\tau(q.s,\sigma)$;
		$r^{next} \leftarrow M_{R}.\delta(q.r,\sigma)$\;
		\uIf{$(r^{next},s^{next}) \notin \mathit{Checked}$}{
			$q^{next} \leftarrow \mathit{CreateNewState}(r^{next},s^{next})$\;
			$\mathit{Frontier} \leftarrow \mathit{Frontier} \cup q^{next}$\;
		}
		\uElse {
			$q^{next} \leftarrow (r^{next},s^{next})$\;
		}
		\tcc{{\footnotesize Both the symbol $\sigma$ and its probability are added to the transition.}}\
		$\delta \leftarrow \mathit{CreateNewTransition}(q,\sigma,q^{next})$\;
		$\gamma \leftarrow \mathit{CreateNewProbability(q,\sigma,M_{S}.\gamma(q.s,\sigma))}$\;
		$\Delta \leftarrow \Delta \cup \delta$;
		$\Gamma \leftarrow \Gamma \cup \gamma$\;		
	}
	$\mathit{Checked} \leftarrow \mathit{Checked} \cup \{q\}$;
	$\mathit{Frontier} \leftarrow \mathit{Frontier} \setminus \{q\}$\;
}
$Q \leftarrow \mathit{Checked}$\;
\tcc{{\footnotesize Create the final states of $M$ by gathering all states of $M$ whose second element is a final state of $M_{R}$.}}\
$Q^{f} \leftarrow \emptyset$\;
\ForEach{$q \in Q$}{
	\uIf{$q.q_{R} \in M_{R}.Q^{f}$}{
		$Q^{f} \leftarrow Q^{f} \cup \{q\}$\;
	}
}
$\Sigma \leftarrow M_{S}.\Sigma$\;
return $M=(Q,Q^{s},Q^{f},\Sigma,\Delta,\Gamma)$\;
\caption{Embedding of a \psa\ in a \dsfa\ (incremental).}
\label{algorithm:merging}
\end{algorithm}

After constructing an embedding $M$ from a \dsfa\ $M_{R}$ and a \psa\ $M_{S}$,
we can use $M$ to perform forecasting on a test stream.
Since $M$ is equivalent to $M_{R}$,
it can also consume a stream and detect the same instances of the expression $R$ as $M_{R}$ would detect.
However, 
our goal is to use $M$ to forecast the detection of an instance of $R$.
More precisely, 
we want to estimate the number of transitions from any state in which $M$ might be until it reaches for the first time one of its final states.
Towards this goal,
we can use the theory of Markov chains.
Let $N$ denote the set of non-final states of $M$ and $F$ the set of its final states.
We can organize the transition matrix of $M$ in the following way
(we use bold symbols to refer to matrices and vectors and normal ones to refer to scalars or sets):
\begin{equation}
\label{eq:matrix}
\boldsymbol{\Pi} = 
\begin{pmatrix} 
\boldsymbol{N} & \boldsymbol{N_{F}}  \\ 
\boldsymbol{F_{N}} & \boldsymbol{F}
\end{pmatrix}
\end{equation}
where $\boldsymbol{N}$ is the sub-matrix containing the probabilities of transitions from non-final to non-final states,
$\boldsymbol{F}$ the probabilities from final to final states,
$\boldsymbol{F_{N}}$ the probabilities from final to non-final states
and $\boldsymbol{N_{F}}$ the probabilities from non-final to final states.
By partitioning the states of a Markov chain into two sets,
such as $N$ and $F$,
the following theorem can be used to estimate the probability of reaching a state in $F$ starting from a state in $N$:
\begin{theorem}[\cite{fu2003distribution}]
\label{theorem:non-finals}
Let $\boldsymbol{\Pi}$ be the transition probability matrix of a homogeneous Markov chain $Y_{t}$ in the form of Equation \eqref{eq:matrix}
and $\boldsymbol{\xi}_{init}$ its initial state distribution.
The probability for the time index $n$ when the system first enters the set of states $F$,
starting from a state in $N$, 
can be obtained from
\begin{equation}
\label{eq:wtd:non-finals}
P(Y_{n} \in F, Y_{n-1} \in N, \cdots, Y_{2} \in N, Y_{1} \in N \mid \boldsymbol{\xi_{init}}) =
\boldsymbol{\xi_{N}}^{T}\boldsymbol{N}^{n-1}(\boldsymbol{I}-\boldsymbol{N})\boldsymbol{1}
\end{equation}
where $\xi_{N}$ is the vector consisting of the elements of $\xi_{init}$ corresponding to the states of $N$.
\end{theorem}
In our case,
the sets $N$ and $F$ have the meaning of being the non-final and final states of $M$.
The above theorem then gives us the desired probability of reaching a final state.

However, notice that this theorem assumes that we start in a non-final state ($Y_{1} \notin F$).
A similar result can be given if we assume that we start in a final state.
\begin{theorem}
\label{theorem:finals}
Let $\boldsymbol{\Pi}$ be the transition probability matrix of a homogeneous Markov chain $Y_{t}$ in the form of Equation \eqref{eq:matrix}
and $\boldsymbol{\xi}_{init}$ its initial state distribution.
The probability for the time index $n$ when the system first enters the set of states $F$,
starting from a state in $F$, 
can be obtained from
\begin{equation}
\label{eq:wtd:finals}
P(Y_{n} \in F, Y_{n-1} \in N, \cdots, Y_{2} \in N, Y_{1} \in F \mid \boldsymbol{\xi_{init}}) =
  \begin{cases}
    \boldsymbol{\xi_{F}}^{T} \boldsymbol{F}  \boldsymbol{1} & \quad \text{if } n=2   \\
    \boldsymbol{\xi_{F}}^{T}  \boldsymbol{F_{N}} \boldsymbol{N}^{n-2}(\boldsymbol{I}-\boldsymbol{N})\boldsymbol{1} & \quad \text{otherwise} \\
  \end{cases}
\end{equation}
where $\xi_{F}$ is the vector consisting of the elements of $\xi_{init}$ corresponding to the states of $F$.
\end{theorem}
\begin{proof}
The proof may be found in the Appendix, Section \ref{sec:proof:finals}.
\end{proof}

Note that the above formulas do not use $\boldsymbol{N_{F}}$,
as it is not needed when dealing with probability distributions.
As the sum of the probabilities is equal to $1$, 
we can derive $\boldsymbol{N_{F}}$ from $\boldsymbol{N}$.
This is the role of the term $(\boldsymbol{I}-\boldsymbol{N})\boldsymbol{1}$ in the formulas,
which is equal to $\boldsymbol{N_{F}}$ when there is only a single final state and equal to the sum of the columns of $\boldsymbol{N_{F}}$ when there are multiple final states,
i.e., each element of the matrix corresponds to the probability of reaching one of the final states from a given non-final state.

Using Theorems \ref{theorem:non-finals} and \ref{theorem:finals},
we can calculate the so-called waiting-time distributions for any state $q$ of the automaton,
i.e.,
the distribution of the index $n$, 
given by the waiting-time variable
$W_{q}=inf\{n: Y_{0},Y_{1},...,Y_{n}, Y_{0}=q, q \in Q \backslash F, Y_{n} \in F\}$.
Theorems \ref{theorem:non-finals} and \ref{theorem:finals} provide a way to calculate the probability of reaching a final state,
given an initial state distribution $\boldsymbol{\xi_{init}}$.
In our case,
as the automaton is moving through its various states,
$\boldsymbol{\xi_{init}}$ takes a special form. 
At any point in time, 
the automaton is (with certainty) in a specific state $q$.
In that state,
$\boldsymbol{\xi_{init}}$ is a vector of $0$,
except for the element corresponding to the current state of the automaton,
which is equal to $1$.

Figure \ref{fig:wtdfas} shows an example of an automaton 
(its exact nature is not important,
as long as it can also be described as a Markov chain),
along with the waiting-time distributions for its non-final states. 
For this example,
if the automaton is in state 2,
then the probability of reaching the final state 4 for the first time in 2 transitions is $\approx 50\%$.
However, it is $0\%$ for 3 transitions,
as the automaton has no path of length 3 from state 2 to state 4.

\begin{figure}[t]
    \centering
    \begin{subfigure}[b]{0.55\textwidth}
        \includegraphics[width=\textwidth]{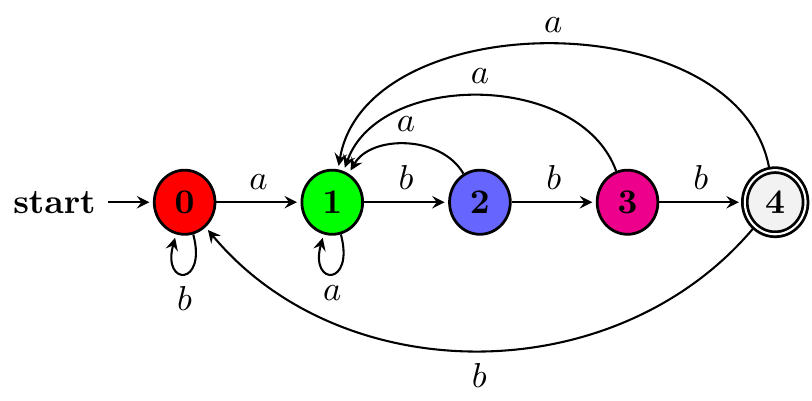}
        \caption{DFA.}\label{fig:dfaabbb}
    \end{subfigure}\\
    \begin{subfigure}[b]{0.63\textwidth}
        \includegraphics[width=\textwidth]{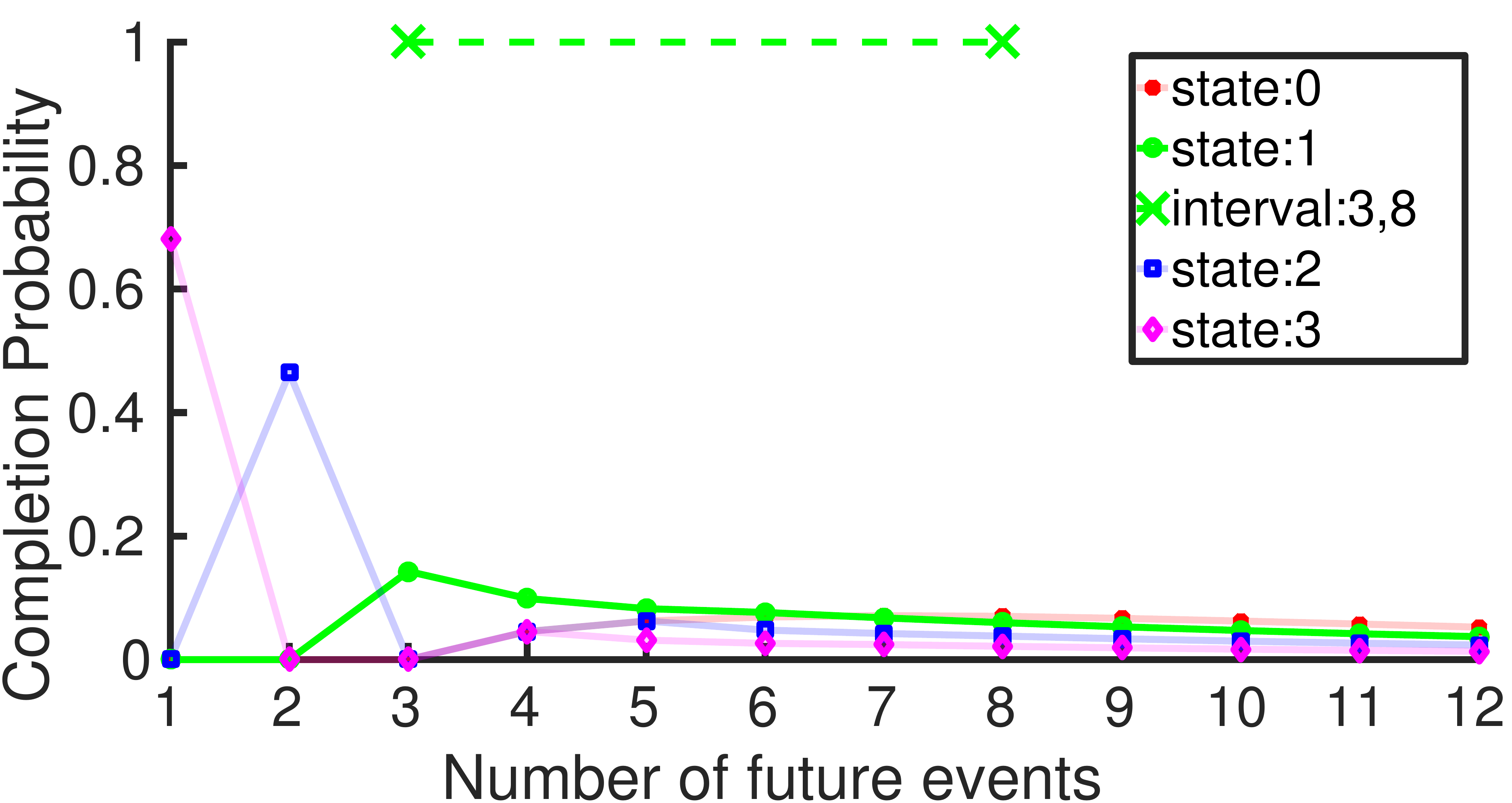}
        \caption{Waiting-time distributions and shortest interval, i.e. $[3,8]$, exceeding a confidence threshold $\theta_{fc} = 50\%$ for state 1.}\label{fig:wt1}
    \end{subfigure}
    \caption{Automaton and waiting-time distributions for $R=a\cdot b\cdot b\cdot b$, $\Sigma=\{a,b\}$.}\label{fig:wtdfas}
\end{figure}

We can use the waiting-time distributions to produce various kinds of forecasts.
In the simplest case,
we can select the future point with the highest probability and return this point as a forecast.
We call this type of forecasting \emph{REGRESSION-ARGMAX}.
Alternatively, 
we may want to know how likely it is that a CE will occur within the next $w$ input events.
For this,
we can sum the probabilities of the first $w$ points of a distribution
and if this sum exceeds a given threshold
we emit a ``positive'' forecast (meaning that a CE is indeed expected to occur);
otherwise a ``negative'' (no CE is expected) forecast is emitted.
We call this type of forecasting \emph{CLASSIFICATION-NEXTW}.
These kinds of forecasts are easy to compute.

There is another kind of useful forecasts,
which are however more computationally demanding.
Given that we are in a state $q$, 
we may want to forecast whether the automaton,
with confidence at least $\theta_{fc}$, 
will have reached its final state(s) in $n$ transitions, 
where $n$ belongs to a future interval $I=[\mathit{start},\mathit{end}]$.
The confidence threshold $\theta_{fc}$ is a parameter set by the user.
The forecasting objective is to select the shortest possible interval $I$ that satisfies $\theta_{fc}$.
Figure \ref{fig:wt1} shows the forecast interval produced for state 1 of the automaton of Figure \ref{fig:dfaabbb}, 
with $\theta_{fc} = 50\%$.
We call this third type of forecasting \emph{REGRESSION-INTERVAL}.
We have implemented all of the above types of forecasting.

\begin{algorithm}[t]
\KwIn{A waiting-time distribution $P$ with horizon $h$ and a threshold $\theta_{fc} < 1.0$}
\KwOut{The smallest interval $I=(s,e)$ such that $1 \leq s,e \leq h$, $s \leq e$ and $P(I) \geq \theta_{fc}$}
$s \leftarrow -1$;
$e \leftarrow -1$;
$i \leftarrow 1$;
$j \leftarrow 1$;
$p \leftarrow P(1)$\;
\While{$j \neq h$}{
\tcc{{\footnotesize Loop invariant: $(s,e)$ is the smallest interval with $P((s,e))>\theta_{fc}$ among all intervals with $e \leq j$ (or $s=e=-1$ in the first iteration).}}\
\tcc{{\footnotesize Expansion phase.}}\
	\While{$(p < \theta_{fc}) \wedge (j < h)$}{
		$j \leftarrow j + 1$\;
		$p \leftarrow p + P(j)$\;
	}
\tcc{{\footnotesize Shrinking phase.}}\
	\While{$p \geq \theta_{fc}$}{
		$i \leftarrow i + 1$\;
		$p \leftarrow p - P(i)$\;
	}
	$i \leftarrow i - 1$\;
	\tcc{{\footnotesize $s=-1$ indicates that no interval has been found yet, i.e., that this is the first iteration.}}\
	\If{$(\mathit{spread}((i,j)) < \mathit{spread}((s,e))) \vee (s = -1)$}{
		$s \leftarrow i$\; 
		$e \leftarrow j$\;
	}
}
return $(s,e)$\;

\caption{Estimating a forecast interval from a waiting-time distribution.}
\label{algorithm:interval}
\end{algorithm}

A naive way to estimate the forecast interval from a waiting-time distribution whose domain is $[1,h]$
(we call $h$, the maximum index of the distribution, its \emph{horizon})
is to first enumerate all possible intervals $(\mathit{start},\mathit{end})$,
such that $1 \leq \mathit{start},\mathit{end} \leq h$ and $\mathit{start} \leq \mathit{end}$, 
and then calculate each interval's probability by summing the probabilities of all of its points.
The complexity of such an exhaustive algorithm is $O(h^{3})$.
To prove this,
first note that the algorithm would have to check 1 interval of length $h$, 2 intervals of length $h-1$, etc., and $h$ intervals of length 1. 
Assuming that the cost of estimating the probability of an interval is proportional to its length $l$
($l$ points need to be retrieved and $l-1$ additions be performed),
the total cost would thus be:
\begin{equation*}
\begin{aligned}
1h + 2(h-1) + 3(h-2) + \cdots + h1 =  & \sum_{i=1}^{h}i(h - (i-1))  \\
= &  \sum_{i=1}^{h}(ih - i^{2} + i) \\
= &  h\sum_{i=1}^{h}i - \sum_{i=1}^{h}i^{2} + \sum_{i=1}^{h}i \\
= & h \frac{h(h+1)}{2} - \frac{h(h+1)(2h+1)}{6} + \frac{h(h+1)}{2} \\
= & \cdots \\
= & \frac{1}{6}h(h+1)(h+2) \\
= & O(h^{3})
\end{aligned}
\end{equation*}
Note that this is just the cost of estimating the probabilities of the intervals,
ignoring the costs of actually creating them first and then searching for the best one,
after the step of probability estimation.

We can find the best forecast interval with a more efficient algorithm that has a complexity linear in $h$
(see Algorithm \ref{algorithm:interval}).
We keep two pointers $i,j$ that we initially set them equal to the first index of the distribution.
We then repeatedly move $i,j$ in the following manner:
We first move $j$ to the right by incrementing it by 1 until $P(i,j)$ exceeds $\theta_{fc}$,
where each $P(i,j)$ is estimated incrementally by repeatedly adding $P(j)$ to an accumulator.
We then move $i$ to the right by $1$ until $P(i,j)$ drops below $\theta_{fc}$,
where $P(i,j)$ is estimated by incremental subtractions.
If the new interval $(i,j)$ is smaller than the smallest interval exceeding $\theta_{fc}$ thus far,
we discard the old smallest interval and keep this new one.
This wave-like movement of $i,j$ stops when $j=h$.
This algorithm is more efficient (linear in the $h$, see Proposition \ref{proposition:complexity6} in Section \ref{sec:complexity}) by avoiding intervals that cannot possibly exceed $\theta_{fc}$.
The proof for the algorithm's correctness is presented in the Appendix, Section \ref{sec:proof:interval}.

Note that the domain of a waiting-time distribution is not composed of timepoints and thus a forecast does not explicitly refer to time.
Each value of the index $n$ on the $x$ axis essentially refers to the number of transitions that the automaton needs to take before reaching a final state, or, equivalently, to the number of future input events to be consumed. 
If we were required to output forecasts referring to time,
we would need to convert these basic event-related forecasts to time-related ones.
If input events arrive at regular time intervals,
then this conversion is a straightforward multiplication of the forecast by the time interval.
However, in the general case where the intervals between input events are not regular and fixed,
we would need to build another probabilistic model describing the time that elapses between events and use this model to convert event-related to time-related forecasts.
Building such a time model might not always be possible or might be prohibitively expensive.
In this paper we decided to focus on the number of steps for two reasons:
a) Sometimes it might not be desirable to give time-related forecasts. 
Event-related forecasts might be more suitable, 
as is the case, for example, in the domain of credit card fraud management, 
where we need to know whether or not the next transaction(s) will be fraudulent.
We examine this use case in Section \ref{sec:cards}.
b) Time-related forecasts might be very difficult (or almost impossible) to produce if the underlying process exhibits a high degree of randomness.
For example,
this is the case in the maritime domain,
where the intervals between vessel position (AIS) messages are wildly random and depend on many (even human-related) factors,
e.g., the crew of a vessel simply forgetting to switch on the AIS equipment.
In such cases,
it might be preferable to perform some form of sampling or interpolation on the original stream of input events in order to derive another stream similar to the original one but with regular intervals.
This is the approach we follow in our experiments in the maritime domain (Section \ref{sec:maritime}).  
For these reasons, we initially focused on event-related forecasts.
This, however, does not exclude the option of using event-related forecasts as a first step in order to subsequently produce time-related ones,
whenever this is possible. 
For example, 
a simple solution would be to try to model the time elapsed between events via a Poisson process.
We intend to pursue this line of work in the future.

\subsubsection{Using a Prediction Suffix Tree (\pst)}
\label{sec:no-mc}

The reason for constructing an embedding of the \psa\ $M_{S}$ learned from the data into the automaton $M_{R}$ used for recognition, 
as described in the previous section,
is that the embedding is based on a variable-order model that will consist on average of much fewer states than a full-order model.
There is, however, one specific step in the process of creating an embedding that may act as a bottleneck and prevent us from increasing the order to desired values:
the step of converting a \pst\ to a \psa.
The number of nodes of a \pst\ is often order of magnitudes smaller than the number of states of the \psa\ constructed from that \pst.
Motivated by this observation,
we devised a way to estimate the required waiting-time distributions without actually constructing the embedding.
Instead, we make direct use of the \pst,
which is more memory efficient. 
Thus, given a \dsfa\ $M_{R}$ and its \pst\ $T$,
we can estimate the probability for $M_{R}$ to reach for the first time one of its final states
in the following manner.

As the system processes events from the input stream,
besides feeding them to $M_{R}$,
it also stores them in a buffer that holds the $m$ most recent events,
where $m$ is equal to the maximum order of the \pst\ $T$.
After updating the buffer with a new event,
the system traverses $T$ according to the contents of the buffer and arrives at a leaf $l$ of $T$.
The probability of any future sequence of events can be estimated with the use of the probability distribution at $l$.
In other words,
if $S_{1..k}=\cdots,t_{k-1},t_{k}$ is the stream seen thus far,
then the next symbol probability for $t_{k+1}$, 
i.e., $P(t_{k+1} \mid t_{k-m+1},\cdots,t_{k})$, 
can be directly retrieved from the distribution of the leaf $l$.
If we want to look further into the future,
e.g., into $t_{k+2}$,
we can repeat the same process as necessary.
Namely, if we fix $t_{k+1}$, 
then the probability for $t_{k+2}$,
$P(t_{k+2} \mid t_{k-m+2},\cdots,t_{k+1})$,
can be retrieved from $T$,
by retrieving the leaf $l^{'}$ reached with $t_{k+1},\cdots,t_{k-m+2}$.
In this manner, we can estimate the probability of any future sequence of events.
Consequently,
we can also estimate the probability of any future sequence of states of the \dsfa\ $M_{R}$,
since we can simply feed these future event sequences to $M_{R}$ and let it perform ``forward'' recognition with these projected events.
In other words,
we can let $M_{R}$ ``generate'' a sequence of future states,
based on the sequence of projected events, 
in order to determine when $M_{R}$ will reach a final state.
Finally, since we can estimate the probability for any future sequence of states of $M_{R}$,
we can use the definition of the waiting-time variable 
(${W_{q}=inf\{n: Y_{0},Y_{1},...,Y_{n}, Y_{0}=q, q \in Q \backslash F, Y_{n} \in F\}}$)
to calculate the waiting-time distributions.
Figure \ref{fig:nomc} shows an example of this process for the automaton $M_{R}$ of Figure \ref{fig:dsfaab}.
Figure \ref{fig:pstab} displays an example \pst\ $T$ learned with the minterms/symbols of $M_{R}$.
\begin{figure}
\centering
\begin{subfigure}[t]{0.75\textwidth}
	\includegraphics[width=0.99\textwidth]{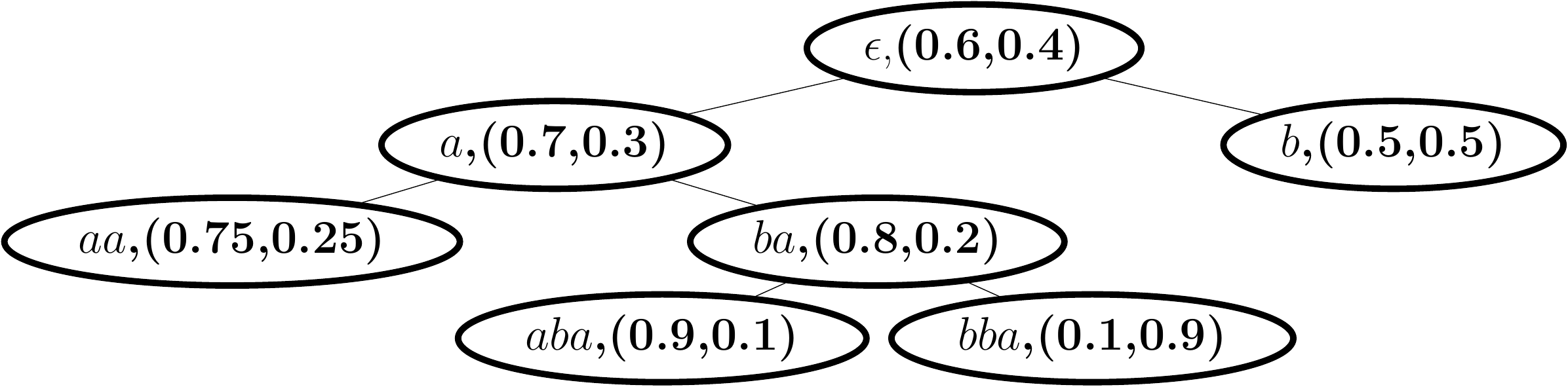}
	\caption{The \pst\ $T$ for the automaton $M_{R}$ of Figure \ref{fig:dsfaab}.}
	\label{fig:pstab}
\end{subfigure}\\
\begin{subfigure}[t]{0.7\textwidth}
	\includegraphics[width=0.99\textwidth]{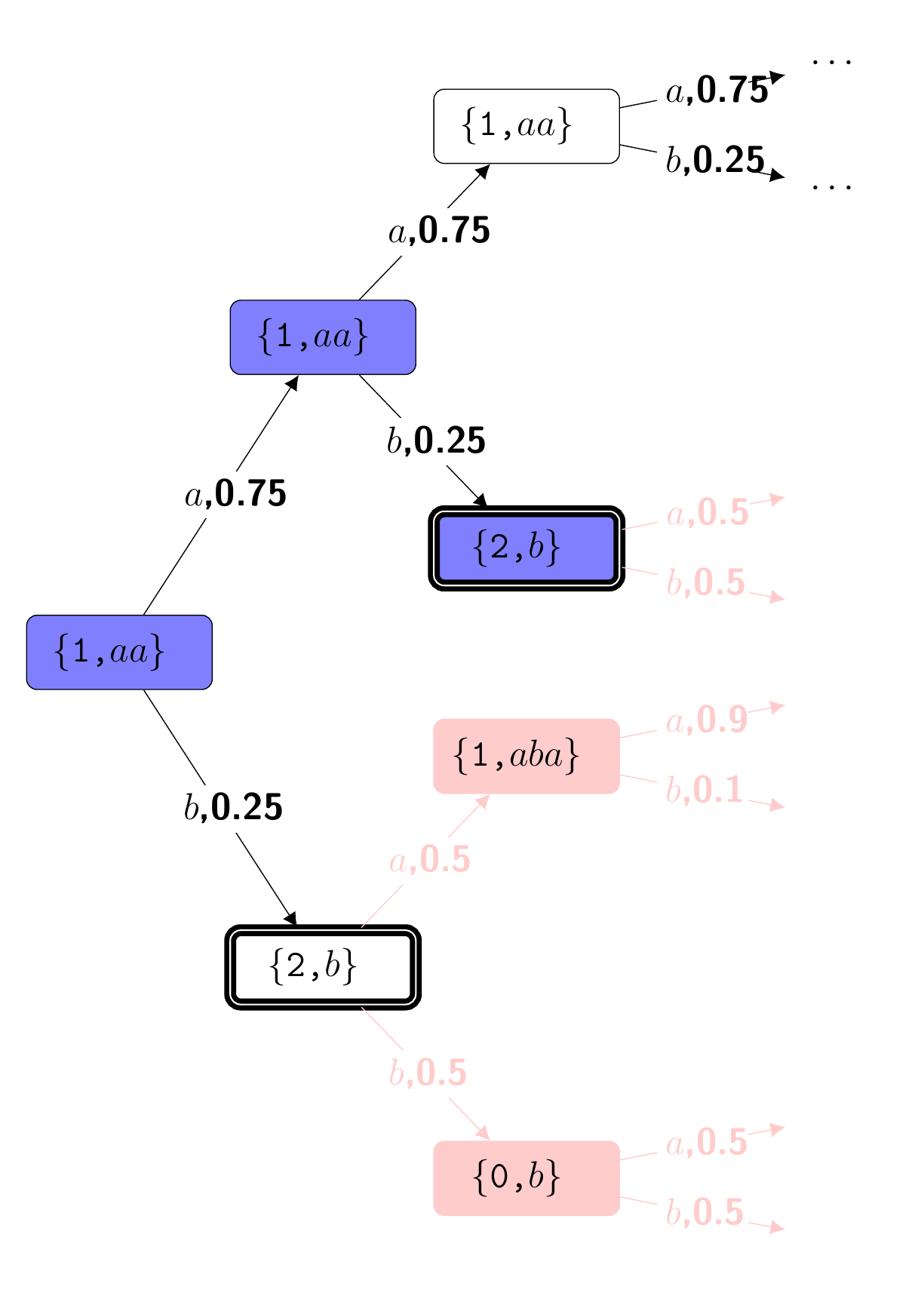}
	\caption{Future paths followed by $M_{R}$ and $T$ starting from state $1$ of $M_{R}$ and node $aa$ of $T$. Purple nodes correspond to the only path of length $k=2$ that leads to a final state. Pink nodes are pruned. Nodes with double borders correspond to final states of $M_{R}$.}
	\label{fig:future}
\end{subfigure}
\caption{Example of estimating a waiting-time distribution without a Markov chain.}
\label{fig:nomc}
\end{figure}

One remark should be made at this point in order to showcase how an attempt to convert $T$ to a \psa\ could lead to a blow-up in the number of states.
The basic step in such a conversion is to take the leaves of $T$ and use them as states for the \psa.
If this was sufficient,
the resulting \psa\ would always have fewer states than the \pst. 
As this example shows,
this is not the case.
Imagine that the states of the \psa\ are just the leaves of $T$ and that we are in the right-most state/node,
$b,(0.5,0.5)$.
What will happen if an $a$ event arrives?
We would be unable to find a proper next state.
The state $aa,(0.75,0.25)$ is obviously not the correct one,
whereas states $aba,(0.9,0.1)$ and $bba,(0.1,0.9)$ are both ``correct'',
in the sense that $ba$ is a suffix of both $aba$ and $bba$.
In order to overcome this ambiguity regarding the correct next state,
we would have to first expand node $b,(0.5,0.5)$ of $T$ and then use the children of this node as states of the \psa.
In this simple example,
this expansion of a single problematic node would not have serious consequences.
But for deep trees and large alphabets,
the number of states generated by such expansions are far more than the number of the original leaves.
For this reason,
the size of the \psa\ is far greater than that of the original, unexpanded \pst. 

Figure \ref{fig:future} illustrates how we can estimate the probability for any future sequence of states of the \dsfa\ $M_{R}$, 
using the distributions of the \pst\ $T$.
Let us assume that,
after consuming the last event, 
$M_{R}$ is in state $1$ and $T$ has reached its left-most node, $aa,(0.75,0.25)$.
This is shown as the left-most node also in Figure \ref{fig:future}.
Each node in this figure has two elements:
the first one is the state of $M_{R}$ and the second the node of $T$,
starting with $\{1,aa\}$ as our current ``configuration''.
Each node has two outgoing edges, one for $a$ and one for $b$,
indicating what might happen next and with what probability.
For example,
from the left-most node of Figure \ref{fig:future},
we know that, 
according to $T$, 
we might see $a$ with probability $0.75$ and $b$ with probability $0.25$.
If we do encounter $b$, 
then $M_{R}$ will move to state 2 and $T$ will reach leaf $b,(0.5,0.5)$.
This is shown in Figure \ref{fig:future} as the white node $\{2,b\}$.
This node has a double border to indicate that $M_{R}$ has reached a final state.

In a similar manner,
we can keep expanding this tree into the future
and use it to estimate the waiting-time distribution for its node $\{1,aa\}$.
In order to estimate the probability of reaching a final state for the first time in $k$ transitions,
we first find all the paths of length $k$ which start from the original node 
and end in a final state without including another final state.
In our example of Figure \ref{fig:future},
if $k=1$,
then the path from $\{1,aa\}$ to $\{2,b\}$ is such a path and its probability is $0.25$.
Thus, $P(W_{\{1,aa\}}=1)=0.25$.
For $k=2$,
the path with the purple nodes leads to a final state after 2 transitions.
Its probability is $0.75*0.25=0.1875$,
i.e., the product of the probabilities on the path edges.
Thus, $P(W_{\{1,aa\}}=2)=0.1875$.
If there were more such alternative paths,
we would have to add their probabilities.

Note that the tree-like structure of Figure \ref{fig:future} is not an actual data structure that we need to construct and maintain. 
It is only a graphical illustration of the required computation steps.
The actual computation is performed recursively on demand.
At each recursive call,
a new frontier of virtual future nodes at level $k$ is generated.
We thus do not maintain all the nodes of this tree in memory, 
but only access the \pst\ $T$,
which is typically much more compact than a \psa.
Despite this fact though, 
the size of the frontier after each recursive call grows exponentially as we try to look deeper into the future.
This cost can be significantly reduced by employing the following optimizations.
First, note in Figure \ref{fig:future}, 
that the paths starting from the two $\{2,b\}$ nodes are pink.
This indicates that these paths do not actually need to be generated, 
as they start from a final state. 
We are only interested in the first time $M_{R}$ reaches a final state and not in the second, third, etc.
As a result,
paths with more than one final states are not useful.
With this optimization,
we can still do an exact estimation of the waiting-time distribution.
Another useful optimization is to prune paths that we know will have a very low probability,
even if they are necessary for an exact estimation of the distributions. 
The intuition is that such paths will not contribute significantly to the probabilities of our waiting-time distribution, 
even if we do expand them. 
We can prune such paths,
accepting the risk that we will have an approximate estimation of the waiting-time distribution.
This pruning can be done without generating the paths in their totality.
As soon as a partial path has a low probability,
we can stop its expansion,
since any deeper paths will have even lower probabilities.
We have found this optimization to be very efficient while having negligible impact on the distribution for a wide range of cut-off thresholds.

\subsection{Estimation of Empirical Probabilities}
\label{sec:prob_empirical}

We have thus far described how an embedding of a \psa\ $M_{S}$ in a \dsfa\ $M_{R}$ can be constructed
and how we can estimate the forecasts for this embedding.
We have also presented how this can be done directly via a \pst, 
without going through a \psa.
However,
before learning the \pst,
as described in Section \ref{sec:pst},
we first need to estimate the empirical probabilities for the various symbols.
We describe here this extra initial step.
In \cite{DBLP:journals/ml/RonST96},
it is assumed that,
before learning a \pst,
the empirical probabilities of symbols given various contexts are available.
The suggestion in \cite{DBLP:journals/ml/RonST96} is that these empirical probabilities can be calculated either by repeatedly scanning the training stream or by using a more time-efficient algorithm that keeps pointers to all occurrences of a given context in the stream.
We opt for a variant of the latter choice.

First, note that the empirical probabilities of the strings ($s$) and the expected next symbols ($\sigma$) observed in a stream are given by the following formulas \cite{DBLP:journals/ml/RonST96}:
\begin{equation}
\label{eq:emprob1}
\hat{P}(s) = \frac{1}{k-m}\sum_{j=m}^{k-1}\chi_{j}(s)
\end{equation} 
\begin{equation}
\label{eq:emprob2}
\hat{P}(\sigma \mid s) = \frac{\sum_{j=m}^{k-1}\chi_{j+1}(s \cdot \sigma)}{\sum_{j=m}^{k-1}\chi_{j}(s)}
\end{equation} 
where $k$ is the length of the training stream $S_{1..k}$,
$m$ is the maximum length of the strings ($s$) that will be considered 
and
\begin{equation}
\label{eq:counters}
\chi_{j}(s) =
  \begin{cases}
    1 & \quad \text{if } S_{(j - \lvert s \rvert + 1) \cdots j} = s   \\
    0 & \quad \text{otherwise} \\
  \end{cases}
\end{equation}
In other words,
we need to count the number of occurrences of the various candidate strings $s$ in $S_{1..k}$.
The numerators and denominators in Eq. \eqref{eq:emprob1} and \eqref{eq:emprob2} are essentially counters for the various strings.

In order to keep track of these counters,
we can use a tree data structure which allows to scan the training stream only once.
We call this structure a \emph{Counter Suffix Tree} (\cst).
Each node in a \cst\ is a tuple $(\sigma,c)$ where $\sigma$ is a symbol from the alphabet (or $\epsilon$ only for the root node) and $c$ a counter.
For each level $k$ of the tree, 
it always holds that 
$\mathit{SumOfCountersAtK} \leq \mathit{ParentCounter}$
and  
$\mathit{SumOfCountersAtK} \geq \mathit{ParentCounter} - (k-1)$.
By following a path from the root to a node,
we get a string $s=\sigma_{0} \cdot \sigma_{1} \cdots \sigma_{n}$,
where $\sigma_{0} = \epsilon$ corresponds to the root node.
The property maintained as  a \cst\ is built from a stream $S_{1..k}$ is that the counter of the node $\sigma_{n}$ that is reached with $s$ 
gives us the number of occurrences of the string $\sigma_{n} \cdot \sigma_{n-1} \cdots \sigma_{1}$ (the reversed version of $s$) in $S_{1..k}$.
As an example,
see Figure \ref{fig:cst},
which depicts the \cst\ of maximum depth 2 for the stream $S=aaabaabaaa$.
If we want to retrieve the number of occurrences of the string $b \cdot a$ in $S$,
we follow the left child $(a,7)$ of the root and then the right child of this.
We thus reach $(b,2)$ and indeed $b \cdot a$ occurs twice in $S$.
\begin{figure}[t]
	\centering
	\includegraphics[width=0.5\textwidth]{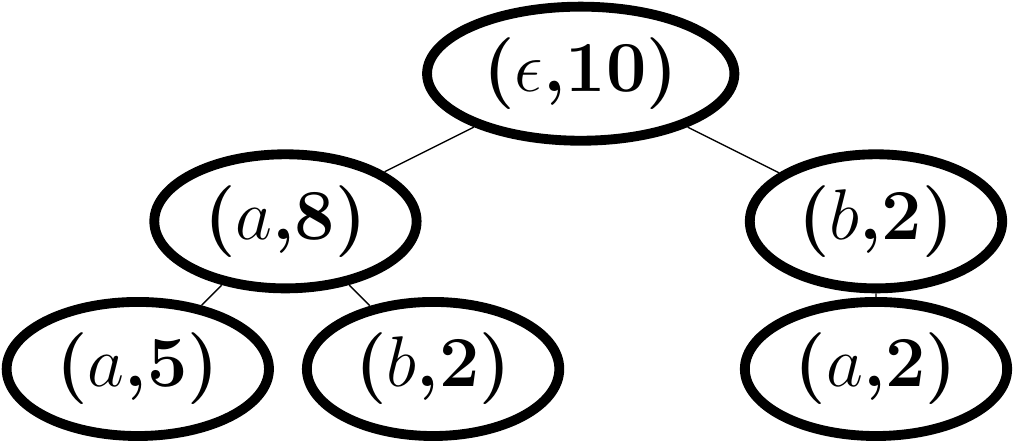}
	\caption{Example of a Counter Suffix Tree with $m=2$ and $S=aaabaabaaa$.}
	\label{fig:cst}
\end{figure}

A \cst\ can be incrementally constructed by maintaining a buffer of size $m$ that always holds the last $m$ symbols of $S$.
The contents of the buffer are fed into the \cst\ after the arrival of a new symbol.
The update algorithm follows a path through the \cst\ according to the whole string provided by the buffer.
For every node that already exists,
its counter is incremented by 1. 
If a node does not exist,
it is created and its counter is set to 1.
At any point, 
having been updated with the training stream, 
the \cst\ can be used to retrieve the necessary counters
and estimate the empirical probabilities of Equations \eqref{eq:emprob1} and \eqref{eq:emprob2} that are subsequently used in the \pst\ construction.

\section{Complexity Analysis}
\label{sec:complexity}

Figure \ref{fig:vmmflow} depicts the steps required for estimating forecasts,
along with the input required for each of them. 
The first step (box $1$) takes as input the minterms of a \dsfa, the maximum order $m$ of dependencies to be captured and a training stream.
Its output is a \cst\ of maximum depth $m$ (Section \ref{sec:prob_empirical}).
In the next step (box $2$), 
the \cst\ is converted to a \pst,
using an approximation parameter $\alpha$ and a parameter $n$ for the maximum number of states for the \psa\ to be constructed subsequently (Section \ref{sec:pst}).
For the third step,
we have two options:
we can either use the \pst\ to directly estimate the waiting-time distributions (box $3b$, Section \ref{sec:no-mc})
or we can convert the \pst\ to a \psa,
by using the leaves of the \pst\ as states of the \psa\ (box $3a$, Section \ref{sec:pst}).
If we follow the first path,
we can then move on directly to the last step of estimating the actual forecasts,
using the confidence threshold $\theta_{fc}$ provided by the user (box $6$).
If we follow the alternative path,
the \psa\ is merged with the initial \dsfa\ to create the embedding of the \psa\ in the \dsfa\ (box $4$, Section \ref{sec:embed}).
From the embedding we can calculate the waiting-time distributions (box $5$),
which can be used to derive the forecasts (box $6$).

The learning algorithm of step $2$,
as presented in \cite{DBLP:journals/ml/RonST96},
is polynomial in $m$, $n$, $\frac{1}{\alpha}$ and the size of the alphabet (number of minterms in our case).
Below, 
we give complexity results for the remaining steps.

\begin{figure}[t]
	\centering
	\includegraphics[width=0.99\textwidth]{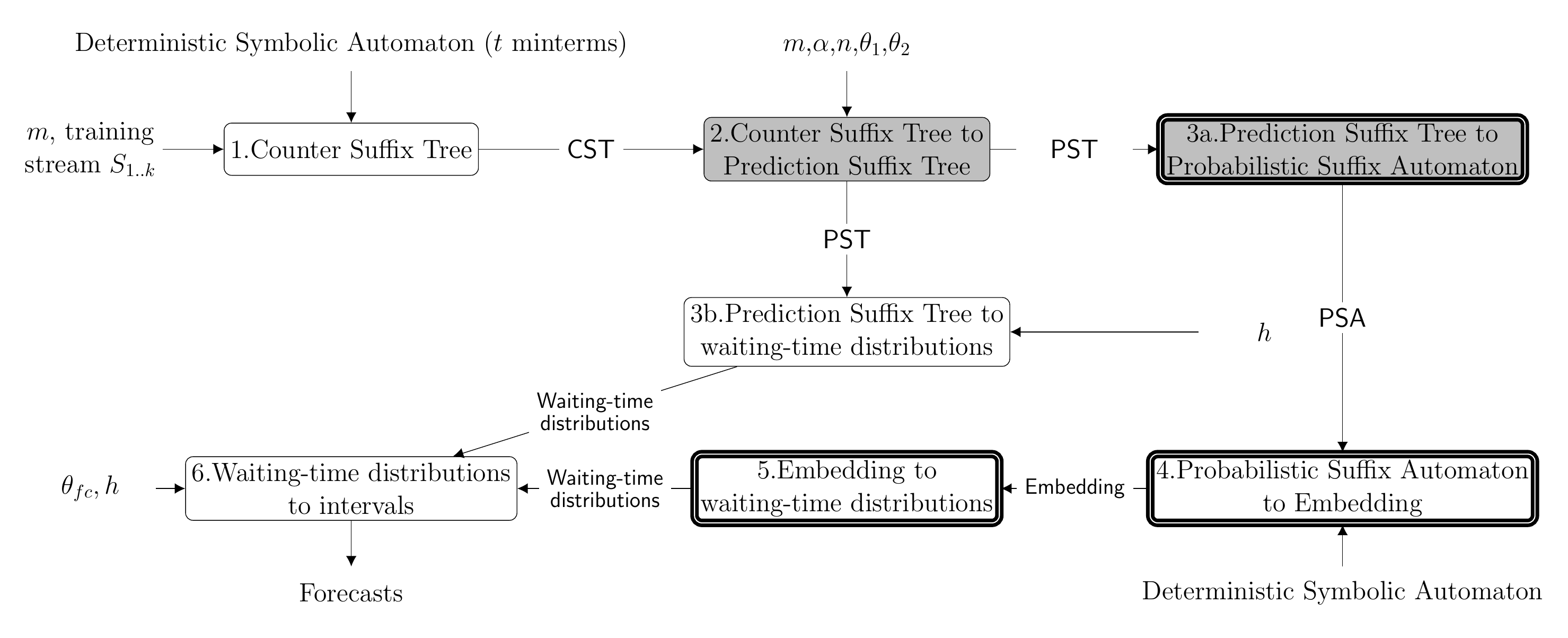}
	\caption{Steps for calculating forecasts. $S_{1..k}$: training stream, $m$: maximum assumed order, $\alpha$: approximation parameter, $n$: maximum number of states for the \psa, $\theta_{fc}$: confidence threshold, $\theta_{1},\theta_{2}$: thresholds for learning a \pst, $k$: size of training stream, $t$: number of minterms, $w$: forecasting window. Gray blocks indicate steps described in \cite{DBLP:journals/ml/RonST96}. White blocks indicate steps introduced in this paper. Blocks with double borders represent the steps required exclusively when going through a PSA embedding, while single-border blocks represent steps bypassing PSA construction.}
	\label{fig:vmmflow}
\end{figure}

\begin{proposition}[Step 1 in Figure \ref{fig:vmmflow}]
\label{proposition:complexity1}
Let $S_{1..k}$ be a stream and $m$ the maximum depth of the Counter Suffix Tree $T$ to be constructed from $S_{1..k}$.
The complexity of constructing $T$ is $O(m(k-m))$.
\end{proposition}
\begin{proof}
See Appendix, Section \ref{sec:proof:complexity1}.
\end{proof}

\begin{proposition}[Step 3a in Figure \ref{fig:vmmflow}]
\label{proposition:complexity3}
Let $T$ be a \pst\ of maximum depth $m$, learned with the $t$ minterms of a \dsfa\ $M_{R}$.
The complexity of constructing a \psa\ $M_{S}$ from $T$ is $O(t^{m+1} \cdot m)$.
\end{proposition}
\begin{proof}
See Appendix, Section \ref{sec:proof:complexity3}.
\end{proof}

\begin{proposition}[Step 3b in Figure \ref{fig:vmmflow}]
\label{proposition:complexity3prime}
Let $T$ be a \pst\ of maximum depth $m$, learned with the $t$ minterms of a \dsfa\ $M_{R}$.
The complexity of estimating the waiting-time distribution for a state of $M_{R}$ and a horizon of length $h$ directly from $T$ is $O((m+3) \frac{t - t^{h+1}}{1 - t})$.
\end{proposition}
\begin{proof}
See Appendix, Section \ref{sec:proof:complexity3prime}.
\end{proof}

\begin{proposition}[Step 4 in Figure \ref{fig:vmmflow}]
\label{proposition:complexity4}
Let $M_{R}$ be a \dsfa\ with $t$ minterms and $M_{S}$ a \psa\ learned with the minterms of $M_{R}$.
The complexity of constructing an embedding $M$ of $M_{S}$ in $M_{S}$ with Algorithm \ref{algorithm:merging} is $O(t \cdot \lvert M_{R}.Q \times M_{S}.Q\rvert)$.
\end{proposition}
\begin{proof}
See Appendix, Section \ref{sec:proof:complexity4}.
\end{proof}

Notice that the cost of learning a \psa\ might be exponential in $m$.
In the worst case,
all permutations of the $t$ minterms of length $m$ need to be added to the \pst\ and the \psa.
This may happen if the statistical properties of the training stream are such that all these permutations are deemed as important. 
In this case, 
the final embedding in the \dsfa\ $M_{R}$ might have up to $t^{m} \cdot \lvert M_{R}.Q \rvert$ states.
This is also the upper bound of the number of states of the automaton the would be created using the method of \cite{DBLP:conf/debs/AlevizosAP17},
where every state of an initial automaton is split into at most $t^{m}$ sub-states,
regardless of the properties of the stream.
Thus, 
in the worst case,
our approach would create an automaton of size similar to an automaton encoding a full-order Markov chain.
Our approach provides an advantage when the statistical properties of the training stream allow us to retain only some of the dependencies of order up to $m$.

\begin{proposition}[Step 5 in Figure \ref{fig:vmmflow}]
\label{proposition:complexity5}
Let $M$ be the embedding of a \psa\ $M_{S}$ in a \dsfa\ $M_{R}$.
The complexity of estimating the waiting-time distribution for a state of $M$ and a horizon of length $h$ using Theorem \ref{theorem:non-finals} is $O((h-1) k^{2.37})$, where $k$ is the dimension of the square matrix $\boldsymbol{N}$. A similar result may be obtained for Theorem \ref{theorem:finals}.
\end{proposition}
\begin{proof}
See Appendix, Section \ref{sec:proof:complexity5}.
\end{proof}

\begin{proposition}[Step 6 in Figure \ref{fig:vmmflow}]
\label{proposition:complexity6}
For a waiting-time distribution with a horizon of length $h$,
the complexity of finding the smallest interval that exceeds a confidence threshold $\theta_{fc}$ with Algorithm \ref{algorithm:interval} is $O(h)$. 
\end{proposition}
\begin{proof}
See Appendix, Section \ref{sec:proof:complexity6}.
\end{proof}

The complexity of the last step ($6$), 
when the forecasts are ``classification'' decisions about whether a CE will occur within the next $w$ input events, 
is $O(w)$.
In order to make such a positive or negative decision,
we can simply sum the probabilities of the first $w$ points of a waiting-time distribution and compare this sum to the given threshold $\theta_{fc}$.
If this sum exceeds the given threshold,
then the decision is positive.
The cost of the summation is $O(w)$. 

\section{Measuring the Quality of Forecasts}
\label{sec:metrics}

As described in Section \ref{sec:forecasts},
there are various types of forecasts that could be produced from the waiting-time distributions of an automaton.
In this section,
we discuss in more detail these different forecasting tasks and how the quality of the produced forecasts can be quantified in each case.
We distinguish three different types of forecasting tasks:
a) SDE forecasting, where our goal is to forecast the next SDE in the input stream;
b) regression CE forecasting, where our goal is to forecast when a CE will occur (either \emph{REGRESSION-ARGMAX} or \emph{REGRESSION-INTERVAL});
c) classification CE forecasting, where our goal is to forecast whether or not a CE will occur within a short future window (\emph{CLASSIFICATION-NEXTW}).  

\subsection{SDE Forecasting}

Although our main focus is on forecasting occurrences of CEs,
we can start with a simpler task,
targeting SDEs.
Not only does this allow us to establish a baseline with some more easily interpretable results,
but it also enables us to show the differences between SDE and CE forecasting.
As we will show,
CE forecasting is more challenging than SDE forecasting,
in terms of the feasibility of looking deep into the past.
Another reason for running experiments for SDE forecasting is to find the best values for the hyperparameters used for learning a prediction suffix tree. 
Since it is much faster to run this type of experiments, 
compared to CE forecasting experiments,
we can use a hypergrid of hyperparameter values and for each hypergrid point we run SDE forecasting.

In this type of experiments,
our goal is to investigate how well our proposed framework can forecast the next SDE to appear in the stream,
given that we know the last $m$ SDEs.  
This task is the equivalent of \emph{next symbol prediction} in the terminology of the compression community \cite{DBLP:journals/jair/BegleiterEY04}.
As explained in Section \ref{sec:vmm},
the metric that we use to estimate the quality of a predictor $\hat{P}$ is the average log-loss with respect to a test sequence $S_{1..k}=t_{1},t_{2},\cdots,t_{k}$,
given by 
$l(\hat{P},S_{1..k}) = - \frac{1}{T} \sum_{i=1}^{k} log \hat{P}(t_{i} \mid t_{1} \cdots t_{i-1})$.
The lower the average log-loss, the better the predictor is assumed to be.

We remind that the ``symbols'' which we try to predict in these experiments are essentially the minterms of our \dsfa\ in our case.
In other words, 
we do not try to predict exactly what the next SDE will be,
but which minterm the next SDE will satisfy.
For example,
if we have the minterms of Table \ref{table:minterms_simplified},
then our task is to predict whether the next SDE will satisfy $\psi_{A}$ (i.e., the speed of the vessel will be below 5 knots),
$\psi_{B}$ (i.e., the speed will be above 20 knots) or
$\neg \psi_{A} \wedge \neg \psi_{B}$ (i.e., the speed will be between 5 and 20 knots).

\subsection{Regression CE Forecasting}
\label{sec:regression}

After SDE forecasting, 
we move on to regression CE forecasting.
Our goal in this task to forecast when a CE will occur. 
We call them \emph{regression} experiments due to the fact that the forecasts are ``continuous'' values,
in the form of forecast intervals/points. 
This is in contrast to the next forecasting task,
where each forecast is a binary value,
indicating whether a CE will occur or not and is thus called a \emph{classification} task.

One important difference between SDE and CE forecasting (both regression and classification) is that,
in SDE forecasting, a forecast is emitted after every new SDE is consumed.
On the other hand, in CE forecasting,
emitting a forecast after every new SDE is feasible in principle,
but not very useful and can also produce results that are misleading.
By their very nature,
CEs are relatively rare within a stream of input SDEs.
As a result,
if we emit a forecast after every new SDE,
some of these forecasts (possibly even the vast majority) will have a significant temporal distance from the CE to which they refer.
As an example,
consider a pattern from the maritime domain which detects the entrance of a vessel in the port of Tangiers.
We can also try to use this pattern for forecasting, 
with the goal of predicting when the vessel will arrive at the port of Tangiers.
However, the majority of the vessel's messages may lie in areas so distant from the port (e.g., in the Pacific ocean) that it would be practically useless to emit forecasts when the vessel is in these areas.
Moreover, if we do emit forecasts from these distant areas,
the scores and metrics that we use to evaluate the quality of the forecasts will be dominated by these, necessarily low-quality, distant forecasts.

For these reasons,
before running a regression experiment,
we must first go through a preprocessing step.  
We must find the timepoints within a stream where it is ``meaningful'' to emit forecasts.
We call these timepoints the \emph{checkpoints} of the stream.
To do so,
we must first perform recognition on the stream to find the timepoints where CEs are detected.
We then set a required distance $d$ that we want to separate a forecast from its CE,
in the sense that we require each forecast to be emitted $d$ events before the CE.
After finding all the CEs in a stream and setting a value for $d$,
we set as checkpoints all the SDEs which occur $d$ events before the CEs.
This typically means that we end up with as many checkpoints as CEs for a given value of $d$
(unless the distance between two consecutive CEs is smaller than $d$,
in which case no checkpoint for the second CE exists).
We can then show results for various values of $d$,
starting from the smallest possible value of $1$ 
(i.e., emitting forecasts from the immediately previous SDE before the CE).

At each checkpoint,
a forecast interval is produced, 
as per Section \ref{sec:forecasts}.
Some of the metrics we can use to assess the quality of the forecasts assume that forecasts are in the form of points.
Such point metrics are the following:
the Root Mean Squared Error (RMSE) and
the Mean Absolute Error (MAE)
(the latter is less sensitive than RMSE to outliers).
If we denote by $C$ the set of all checkpoints,
by $y_{c}$ the actual distance (in number of events) between a checkpoint $c$ and its CE (which is always $d$)
and by $\hat{y}_{c}$ our forecast,
then the definitions for RMSE and MAE are as follows:
\begin{equation}
\mathit{RMSE} = \sqrt{ \frac{1}{\lvert C \rvert} \sum_{c \in C}{    {\lvert \hat{y}_{c} - y_{c} \rvert}}^{2} }
\end{equation}
\begin{equation}
\label{eq:mae}
\mathit{MAE} = \frac{1}{\lvert C \rvert} \sum_{c \in C}{\lvert \hat{y}_{c} - y_{c} \rvert}
\end{equation}
When these metrics are used,
we need to impose an extra constraint on the forecasts,
requiring that the maximum spread of each forecast is $0$,
i.e., we produce point (instead of interval) forecasts.

Besides these points metrics,
we can also use an interval metric,
the so-called \emph{negatively oriented interval score} \cite{gneiting2007strictly}.
If $\hat{y}_{c}=(l_{c},u_{c})$ is an interval forecast produced with confidence $b=1-a$ and $y_{c}$ the actual observation (distance), 
then the negatively oriented interval score (NOIS) for this forecast is given by:
\begin{equation}
\label{eq:nois}
\mathit{NOIS}_{c} = (u_{c}-l_{c}) + \frac{2}{a}(l_{c}-y_{c})I_{x<l_{c}} + \frac{2}{a}(y_{c}-u_{c})I_{x>u_{c}}
\end{equation}
We can then estimate the average negatively oriented interval score (ANOIS) as follows:
\begin{equation}
\mathit{ANOIS} = \frac{1}{\lvert C \rvert} \sum_{c \in C}{\mathit{NOIS}_{c}}
\end{equation}
The best possible value for ANOIS is $0$ and is achieved only when all forecasts are point forecasts
(so that $u_{c}-l_{c}$ is always $0$) and all of them are also correct (so that the last two terms in Eq. \ref{eq:nois} are always $0$).
In every other case,
a forecast is penalized if its interval is long 
(so that focused intervals are promoted),
regardless of whether it is correct.
If it is indeed correct,
no other penalty is added.
If it is not correct,
then an extra penalty is added,
which is essentially the deviation of the forecast from the actual observation,
weighted by a factor that grows with the confidence of the forecast.
For example, if the confidence is $100\%$, 
then $b=1.0$, $a=0.0$ and the extra penalty,
according to Eq. \ref{eq:nois},
grows to infinity.
Incorrect forecasts produced with high confidence are thus severely penalized.
Note that if we emit only point forecasts ($\hat{y}_{c}=l_{c}=u_{c}$),
then NOIS and ANOIS could be written as follows:
\begin{equation}
\label{eq:nois0}
\mathit{NOIS}_{c} = \frac{2}{a} \lvert \hat{y}_{c}-y_{c} \rvert
\end{equation}
\begin{equation}
\mathit{ANOIS} = \frac{1}{\lvert C \rvert} \sum_{c \in C}{  \frac{2}{a}\lvert \hat{y}_{c}-y_{c} \rvert }
\end{equation}
ANOIS then becomes a weighted version of MAE.

\subsection{Classification CE Forecasting}

The last forecasting task is the most challenging. 
In contrast to regression experiments, 
where we emit forecasts only at a specified distance before each CE,
in classification experiments we emit forecasts regardless of whether a CE occurs or not.
The goal is to predict the occurrence of CEs within a short future window or provide a ``negative'' forecast if our model predicts that no CE is likely to occur within this window.
In practice,
this task could be the first one performed at every new event arrival.
If a positive forecast is emitted,
then the regression task could also be performed in order to pinpoint more accurately when the CE will occur.

One issue with classification experiments is that it is not so straightforward to establish checkpoints in the stream.
In regression experiments,
CEs provide a natural point of reference.
In classification experiments,
we do not have such reference points,
since we are also required to predict the absence of CEs.
As a result,
instead of using the stream to find checkpoints,
we can use the structure of the automaton itself.
We may not know the actual distance to a CE,
but the automaton can provide us with an ``expected'' or ``possible'' distance,
as follows.
For an automaton that is in a final state,
it can be said that the distance to a CE is $0$.
More conveniently,
we can say that the ``process'' which it describes has been completed or, equivalently, that there remains $0\%$ of the process until completion.
For an automaton that is in a non-final state but separated from a final state by $1$ transition,
it can be said that the ``expected'' distance is $1$.
We use the term ``expected'' because we are not interested in whether the automaton will actually take the transition to a final state.
We want to establish checkpoints both for the presence and the absence of CEs.
When the automaton fails to take the transition to a final state (and we thus have an absence of a CE),
this ``expected'' distance is not an actual distance,
but a ``possible'' one that failed to materialize.
We also note that there might also exist other walks from this non-final state to a final one whose length could be greater than $1$
(in fact, there might exist walks with ``infinite length'', in case of loops).
In order to estimate the ``expected'' distance of a non-final state, 
we only use the shortest walk to a final state.

After estimating the expected distances of all states,
we can then express them as percentages by dividing them by the greatest among them.
A $0\%$ distance will thus refer to final states,
whereas a $100\%$ distance to the state(s) that are the most distant to a final state,
i.e.,
the automaton has to take the most transitions to reach a final state.
These are the start states. 
We can then determine our checkpoints by specifying the states in which the automaton is permitted to emit forecasts,
according to their ``expected'' distance.
For example,
we may establish checkpoints by allowing only states with a distance between $40\%$ and $60\%$ to emit forecasts.
The intuition here is that,
by increasing the allowed distance,
we make the forecasting task more difficult. 
Another option for measuring the distance of a state to a possible future CE would be to use the waiting-time distribution of the state and set its expectation as the distance.
However, this assumes that we have gone through the training phase first and learned the distributions.
For this reason, we avoid using this way to estimate distances.

The evaluation task itself consists of the following steps.
At the arrival of every new input event,
we first check whether the distance of the new automaton state falls within the range of allowed distances,
as explained above.
If the new state is allowed to emit a forecast,
we use its waiting-time distribution to produce the forecast.
Two parameters are taken into account: 
the length of the future window $w$ within which we want to know whether a CE will occur and the confidence threshold $\theta_{fc}$.
If the probability of the first $w$ points of the distribution exceeds the threshold $\theta_{fc}$,
we emit a positive forecast, essentially affirming that a CE will occur within the next $w$ events;
otherwise, we emit a negative forecast, essentially rejecting the hypothesis that a CE will occur.
We thus obtain a binary classification task.

As a consequence,
we can make use of standard classification measures,
like precision and recall.
Each forecast is evaluated: 
a) as a \emph{true positive} (TP) if the forecast is positive and the CE does indeed occur within the next $w$ events from the forecast;
b) as a \emph{false positive} (FP) if the forecast is positive and the CE does not occur;
c) as a \emph{true negative} (TN) if the forecast is negative and the CE does not occur and
d) as a \emph{false negative} (FN) if the forecast is negative and the CE does occur;
Precision is then defined as $\mathit{Precision}=\frac{TP}{TP + FP}$ and recall (also called sensitivity or true positive rate) as $\mathit{Recall}=\frac{TP}{TP + FN}$.
As already mentioned,
CEs are relatively rare in a stream.
It is thus important for a forecasting engine to be as specific as possible in identifying the true negatives.
For this reason,
besides precision and recall,
we also use \emph{specificity} (also called true negative rate),
defined as $\mathit{Specificity}=\frac{TN}{TN + FP}$.

A classification experiment is performed as follows.
For various values of the ``expected'' distance and the confidence threshold $\theta_{fc}$,
we estimate precision, recall and specificity on a test dataset.
For a given distance, 
$\theta_{fc}$ acts as a cut-off parameter.
For each value of $\theta_{fc}$,
we estimate the recall (sensitivity) and specificity scores and we plot these scores as a ROC curve.
For each distance,
we then estimate the area under curve (AUC) for the ROC curves.
The higher the AUC value,
the better the model is assumed to be.

The setting described above is the most suitable for evaluation purposes,
but might not be the most appropriate when such a system is actually deployed.
For deployment purposes,
another option would be to simply set a best, fixed confidence threshold 
(e.g., by selecting, after evaluation, the threshold with the highest F1-score or Matthews correlation coefficient)
and emit only positive forecasts,
regardless of their distance.
Forecasts with low probabilities (i.e., negative forecasts) will thus be ignored/suppressed.
This is justified by the fact that a user would typically be more interested in positive forecasts.
For evaluation purposes, 
this would not be an appropriate experimental setting,
but it would suffice for deployment purposes,
where we would then be focused on fine-tuning the confidence threshold.
In this paper,
we focus on evaluating our system and thus do not discuss further any deployment solution.

\section{Empirical Evaluation}
\label{sec:experiments}

We now present experimental results on two datasets, 
a synthetic one (Section \ref{sec:cards}) and a real-world one (Section \ref{sec:maritime}).
We first briefly discuss in Section \ref{sec:test-models} the models that we evaluated 
and present our software and hardware settings in Section \ref{sec:settings}.
We show results for SDE forecasting and classification CE forecasting.
Due to space limitations,
we omit results for regression CE forecasting .
We intend to present them in future work.

\subsection{Models Tested}
\label{sec:test-models}

In the experiments that we present,
we evaluated the variable-order Markov model that we have presented in this paper in its two versions:
the memory efficient one that bypasses the construction of a Markov chain and makes direct use of the \pst\ learned from a stream (Section \ref{sec:no-mc}) and the computationally efficient one that constructs a \psa\ (Section \ref{sec:embed}).
We compared these against four other models inspired by the relevant literature.

The first, described in \cite{DBLP:conf/debs/AlevizosAP17,DBLP:conf/lpar/AlevizosAP18}, 
is the most similar in its general outline to our proposed method.
It is a previous version of our system presented in this paper and is also based on automata and Markov chains.
The main difference is that it attempts to construct full-order Markov models of order $m$
and is thus typically restricted to low values for $m$.

The second model is presented in \cite{DBLP:conf/debs/MuthusamyLJ10},
where automata and Markov chains are used once again.
However, the automata are directly mapped to Markov chains and no attempt is made to ensure that the Markov chain is of a certain order.
Thus, in the worst case,
this model essentially makes the assumption that SDEs are i.i.d. and $m=0$.

As a third alternative,
we evaluated a model that is based on Hidden Markov Models (HMM),
similar to the work presented in \cite{DBLP:conf/colcom/PandeyNC11}.
That work uses the Esper event processing engine \cite{esper} and attempts to model a business process as a HMM.
For our purposes,
we use a HMM to describe the behavior of an automaton,
constructed from a given symbolic regular expression.
The observation variable of the HMM corresponds to the states of the automaton.
Thus, the set of possible values of the observation variable is the set of automaton states.
An observation sequence of length $l$ for the HMM consists of the sequence of $l$ states visited by the automaton after consuming $l$ SDEs.
The $l$ SDEs (symbols) are used as values for the hidden variable.
The last $l$ symbols are the last $l$ values of the hidden variable. 
Therefore, this HMM always has $l$ hidden states,
whose values are taken from the SDEs,
connected to $l$ observations,
whose values are taken from the automaton states.  
We can train such a HMM with the Baum-Welch algorithm,
using the automaton to generate a training observation sequence from the original training stream.
We can then use this learned HMM to produce forecasts on a test dataset.
We produce forecasts in an online manner as follows:
as the stream is consumed,
we use a buffer to store the last $l$ states visited by the pattern automaton.
After every new event,
we ``unroll'' the HMM using the contents of the buffer as the observation sequence
and the transition and emission matrices learned during the training phase.
We can then use the forward algorithm to estimate the probability of all possible future observation sequences (up to some length),
which, in our case, correspond to future states visited by the automaton.
Knowing the probability of every future sequence of states allows us to estimate the waiting-time distribution for the current state of the automaton and thus build a forecast,
as already described.
Note that, 
contrary to the previous approaches,
the estimation of the waiting-time distribution via a HMM must be performed online.
We cannot pre-compute the waiting-time distributions and store the forecasts in a look-up table,
due to the possibly large number of entries.
For example, 
assume that $l=5$ and the size of the ``alphabet'' (SDE symbols) of our automaton is $10$.
For each state of the automaton,
we would have to pre-compute $10^{5}$ entries. 
In other words,
as with Markov chains,
we still have a problem of combinatorial explosion.
We try to ``avoid'' this problem by estimating the waiting-time distributions online.

Our last model is inspired by the work presented in \cite{DBLP:journals/is/AalstSS11}.
This method comes from the process mining community and has not been previously applied to CEF.
However,
due to its simplicity,
we use it here as a baseline method.
We again use a training dataset to learn the model.
In the training phase,
every time the pattern automaton reaches a certain state $q$,
we simply count how long (how many transitions) we have to wait until it reaches a final state.
After the training dataset has been consumed,
we end up with a set of such ``waiting times'' for every state.
The forecast to be produced by each state is then estimated simply by calculating the average ``waiting time''.

As far as the Markov models are concerned,
we try to increase their order to the highest possible value, 
in order to determine if and how high-order values offer an advantage.
We have empirically discovered that our system can efficiently handle automata and Markov chains that have up to about  $1200$ states.
Beyond this point, 
it becomes almost prohibitive (with our hardware) to create and handle transition matrices with more than $1200^{2}$ elements.
We have thus set this number as an upper bound and increased the order of a model until this number is reached.
This restriction is applied both to full-order models and variable-order models that use a \psa\ and an embedding,
since in both of these cases we need to construct a Markov chain.
For the variable-order models that make direct use of a \pst,
no Markov chain is constructed.
We thus increase their order until their performance scores seem to reach a stable number or a very high number,
beyond which it makes little sense to continue testing.

\subsection{Hardware and Software Settings}
\label{sec:settings}
All experiments were run on a 64-bit Debian 10 (buster) machine with Intel Core i7-8700 CPU @ 3.20GHz × 12  processors and 16 GB of memory.
Our framework was implemented in Scala 2.12.10.
We used Java 1.8,
with the default values for the heap size. 
For the HMM models,
we relied on the Smile machine learning library \cite{smile}.
All other models were developed by us.
No attempt at parallelization was made.

\subsection{Credit Card Fraud Management}
\label{sec:cards}

The first dataset used in our experiments is a synthetic one,
inspired by the domain of credit card fraud management \cite{DBLP:conf/debs/ArtikisKCBMSFP17}.
We start with a synthetically generated dataset in order to investigate how our method performs under conditions that are controlled and produce results more readily interpretable.
The data generator was developed in collaboration with Feedzai\footnote{\url{https://feedzai.com}}, 
our partner in the SPEEDD project\footnote{\url{http://speedd-project.eu}}. 

In this dataset,
each event is supposed to be a credit card transaction,
accompanied by several arguments,
such as the time of the transaction, the card ID, the amount of money spent, the country where the transaction took place, etc.
In the real world,
a very small proportion of such transactions are fraudulent and the goal of a CER system would be to detect, 
with very low latency,
fraud instances.
To do so, 
a set of fraud patterns must be provided to the engine.
For typical cases of such patterns in a simplified form,
see \cite{DBLP:conf/debs/ArtikisKCBMSFP17}.
In our experiments,
we use one such pattern,
consisting of a sequence of consecutive transactions, 
where the amount spent at each transaction is greater than that of the previous transaction.
Such a trend of steadily increasing amounts constitutes a typical fraud pattern.
The goal in our forecasting experiments is to predict if and when such a pattern will be completed,
even before it is detected by the engine (if in fact a fraud instance occurs),
so as to possibly provide a wider margin for action to an analyst.

We generated a dataset consisting of 1,000,000 transactions in total from 100 different cards.
About $20\%$ of the transactions are fraudulent.
Not all of these instances of fraud belong to the pattern of increasing amounts. 
We actually inject seven different types of known fraudulent patterns in the dataset, 
including, for instance, a decreasing trend.
Each fraudulent sequence for the increasing trend consists of eight consecutive transactions with increasing amounts,
where the amount is increased each time by $100$ monetary units or more.
We additionally inject sequences of transactions with increasing amounts,
which are not fraudulent.
In those cases,
we randomly interrupt the sequence before it reaches the eighth transaction.
In the legitimate sequences
the amount is increased each time by $0$ or more units.
With this setting,
we want to test the effect of long-term dependencies on the quality of the forecasts.
For example,
a sequence of six transactions with increasing amounts,
where all increases are $100$ or more units is very likely to lead to a fraud detection.
On the other hand,
a sequence of just two transactions with the same characteristics,
could still possibly lead to a detection,
but with a significantly reduced probability.
We thus expect that models with deeper memories will perform better.
We used $75\%$ of the dataset for training and the rest for testing.
No k-fold cross validation is performed,
since each fold would have exactly the same statistical properties.

Formally, the symbolic regular expression that we use to capture the pattern of an increasing trend in the amount spent is the following:
\begin{equation}
\label{exp:amount}
\begin{aligned}
R := \ & \ (\mathit{amountDiff > 0}) \cdot (\mathit{amountDiff > 0}) \cdot (\mathit{amountDiff > 0}) \cdot (\mathit{amountDiff > 0}) \cdot \\
\ &\ (\mathit{amountDiff > 0}) \cdot (\mathit{amountDiff > 0}) \cdot (\mathit{amountDiff > 0}) 
\end{aligned}
\end{equation}
$\mathit{amountDiff}$ is an extra attribute 
(besides the card ID, the amount spent, the transaction country and the other standard attributes) 
with which we enrich each event and is equal to the difference between the amount spent by the current transaction and that spent by the immediately previous transaction from the same card.
The expression consists of seven terminal sub-expressions, 
in order to capture eight consecutive events.
The first terminal sub-expression captures an increasing amount between the first two events in a fraudulent pattern.

If we attempted to perform forecasting based solely on Pattern \eqref{exp:amount},
then the minterms that would be created would be based only on the predicate $\mathit{amountDiff}>0$:
namely, the predicate itself, along with its negation $\neg (\mathit{amountDiff}>0)$.
As expected,
such an approach does not yield good results,
as the language is not expressive enough to differentiate between fraudulent and legitimate transaction sequences.
In order to address this lack of informative (for forecasting purposes) predicates,
we have incorporated a mechanism in our system that allows us to incorporate extra predicates when building a probabilistic model,
without affecting the semantics of the initial expression (exactly the same matches are detected).
We do this by using any such extra predicates during the construction of the minterms.
For example,
if $\mathit{country}=\mathit{MA}$ is such an extra predicate that we would like included,
then we would construct the following minterms for Pattern \eqref{exp:amount}:
a) $m_{1} = (\mathit{amountDiff > 0}) \wedge (\mathit{country}=\mathit{MA})$;
b) $m_{2} = (\mathit{amountDiff > 0}) \wedge \neg (\mathit{country}=\mathit{MA})$;
c) $m_{3} = \neg (\mathit{amountDiff > 0}) \wedge (\mathit{country}=\mathit{MA})$;
d) $m_{4} = \neg (\mathit{amountDiff > 0}) \wedge \neg (\mathit{country}=\mathit{MA})$).
We can then use these enhanced minterms as guards on the automaton transitions in a way that does not affect the semantics of the expression.
For example,
if an initial transition has the guard $\mathit{amountDiff > 0}$,
then we can split it into two derived transitions,
one for $m_{1}$ and one for $m_{2}$.
The derived transitions would be triggered exactly when the initial one is triggered,
the only difference being that the derived transitions also have information about the country.
For our experiments and for Pattern \eqref{exp:amount},
if we include the extra predicate
$\mathit{amountDiff}>100$,
we expect the model to be able to differentiate between sequences involving genuine transactions 
(where the difference in the amount can by any value above $0$)
and fraudulent sequences 
(where the difference in the amount is always above $100$ units).

We now present results for SDE forecasting. 
As already mentioned in Section \ref{sec:test-models},
for this type of experiments we do not use the automaton created by Pattern \eqref{exp:amount}.
We instead use only its minterms which will constitute our ``alphabet''.
In our case,
there are four minterms:
a) $\mathit{amountDiff}>0 \wedge \mathit{amountDiff}>100$;
b) $\mathit{amountDiff}>0 \wedge \neg (\mathit{amountDiff}>100)$;
c) $\neg (\mathit{amountDiff}>0) \wedge \mathit{amountDiff}>100$;
d) $\neg (\mathit{amountDiff}>0) \wedge \neg (\mathit{amountDiff}>100)$.
Thus, the goal is to predict,
as the stream is consumed,
which one of these minterms will be satisfied.
Notice that, for every possible event,
exactly one minterm is satisfied
(the third one, $\neg (\mathit{amountDiff}>0) \wedge \mathit{amountDiff}>100$, is actually unsatisfiable). 
We use $75\%$ of the original dataset 
(which amounts to 750.000 transactions)
for training and the rest (250.000) for testing.
We do not employ cross-validation,
as the dataset is synthetic and the statistical properties of its folds would not differ.

\begin{figure}[t]
\begin{subfigure}[t]{0.45\textwidth}
	\includegraphics[width=0.95\textwidth]{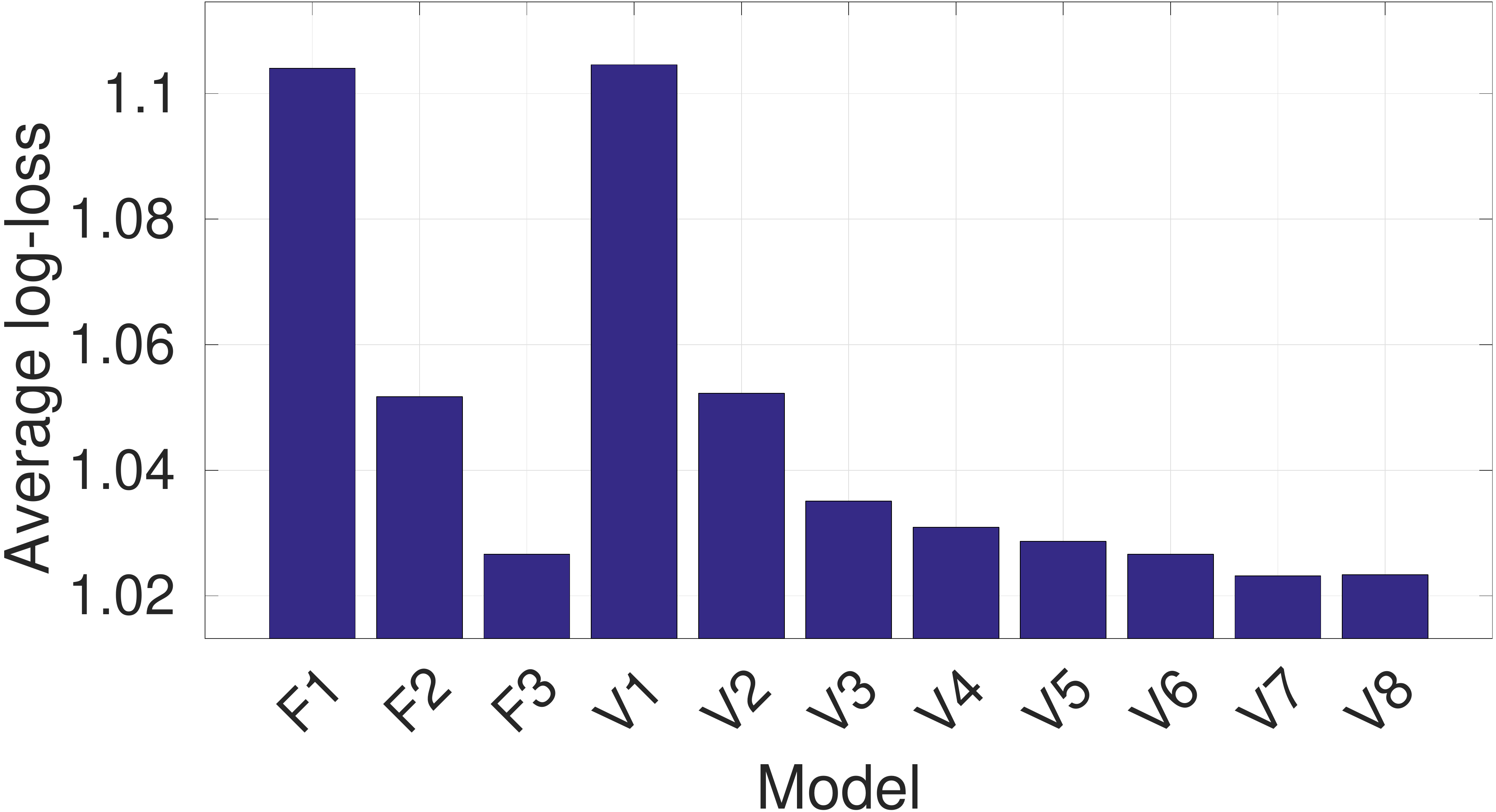}
	\caption{Average log-loss.}
	\label{fig:exp-feedzai-sde-logloss}
\end{subfigure}
\begin{subfigure}[t]{0.45\textwidth}
	\includegraphics[width=0.95\textwidth]{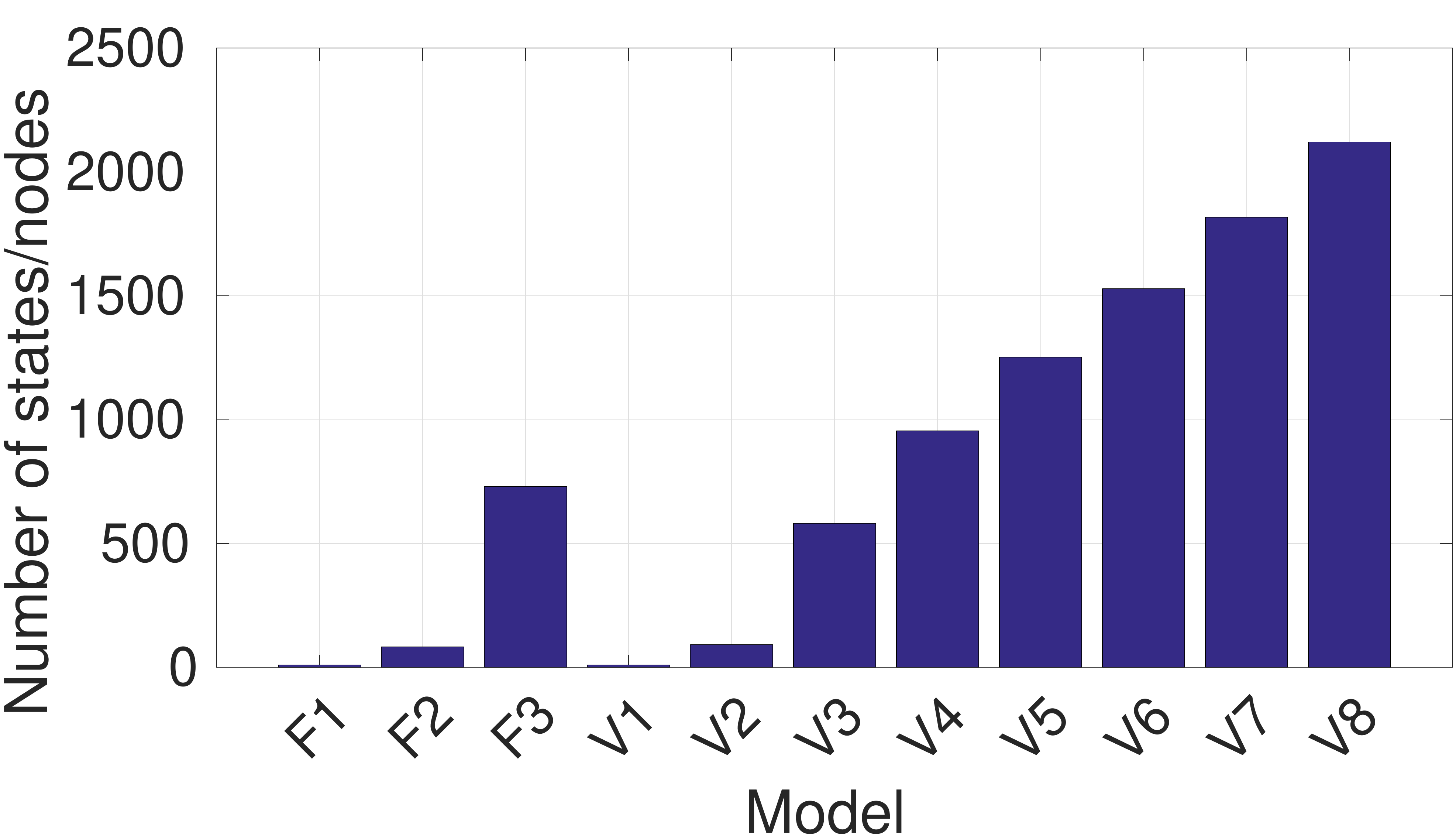}
	\caption{Number of states/nodes.}
	\label{fig:exp-feedzai-sde-states}
\end{subfigure}
\caption{Results for SDE forecasting from the domain of credit card fraud management. Fx stands for a Full-order Markov Model of order $x$. Vx stands for a Variable-order Markov Model (a prediction suffix tree) of maximum order $x$.}
\label{fig:exp-feedzai-sde}
\end{figure}

\begin{figure}[t]
\centering
\begin{subfigure}[t]{0.48\textwidth}
	\includegraphics[width=0.99\textwidth]{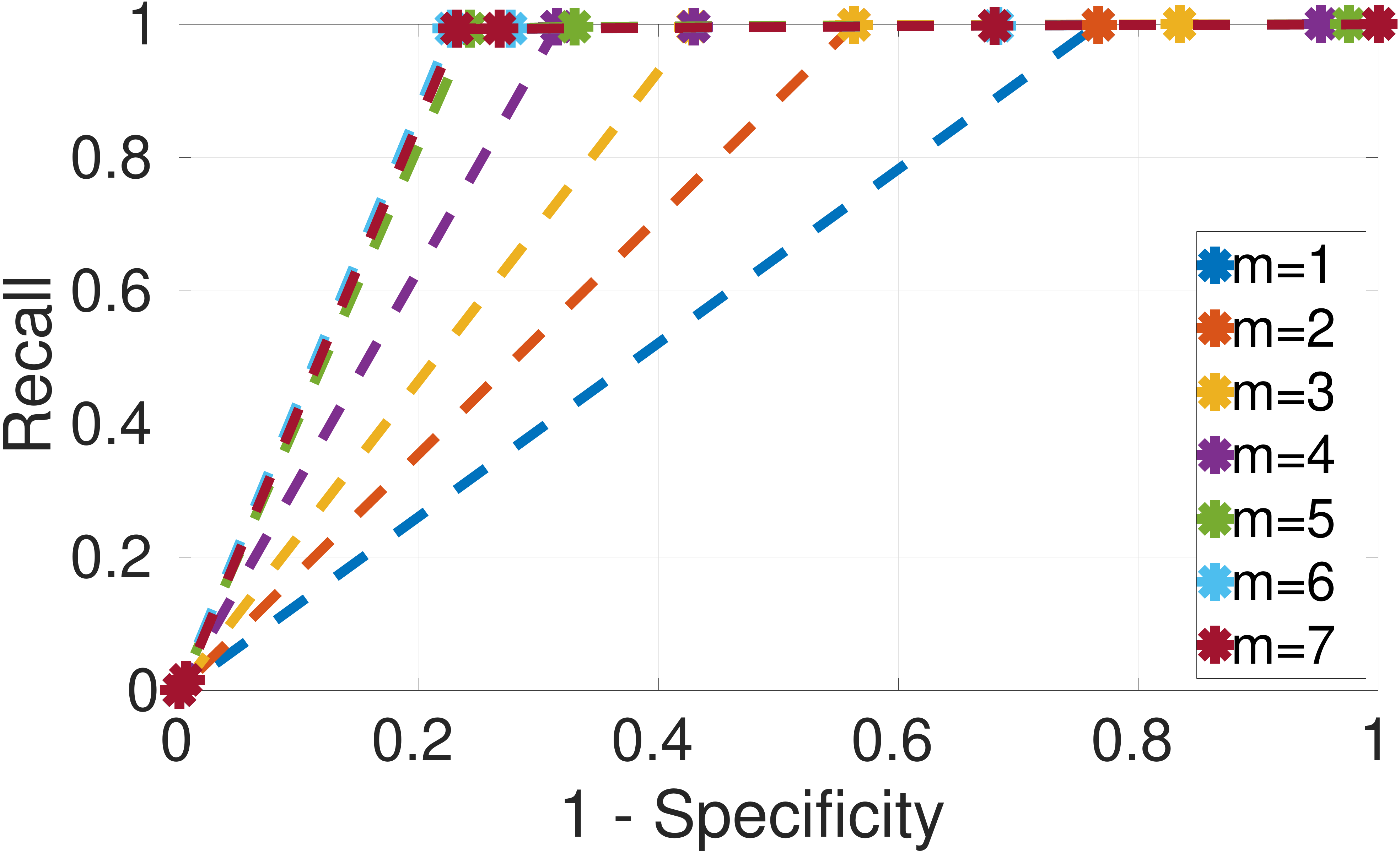}
	\caption{ROC curves for the variable-order model using the \pst\ for various values of the maximum order $m$. $\mathit{distance} \in [0.2,0.4]$.}
	\label{fig:exp-feedzai-classification-roc1}
\end{subfigure}
\begin{subfigure}[t]{0.48\textwidth}
	\includegraphics[width=0.99\textwidth]{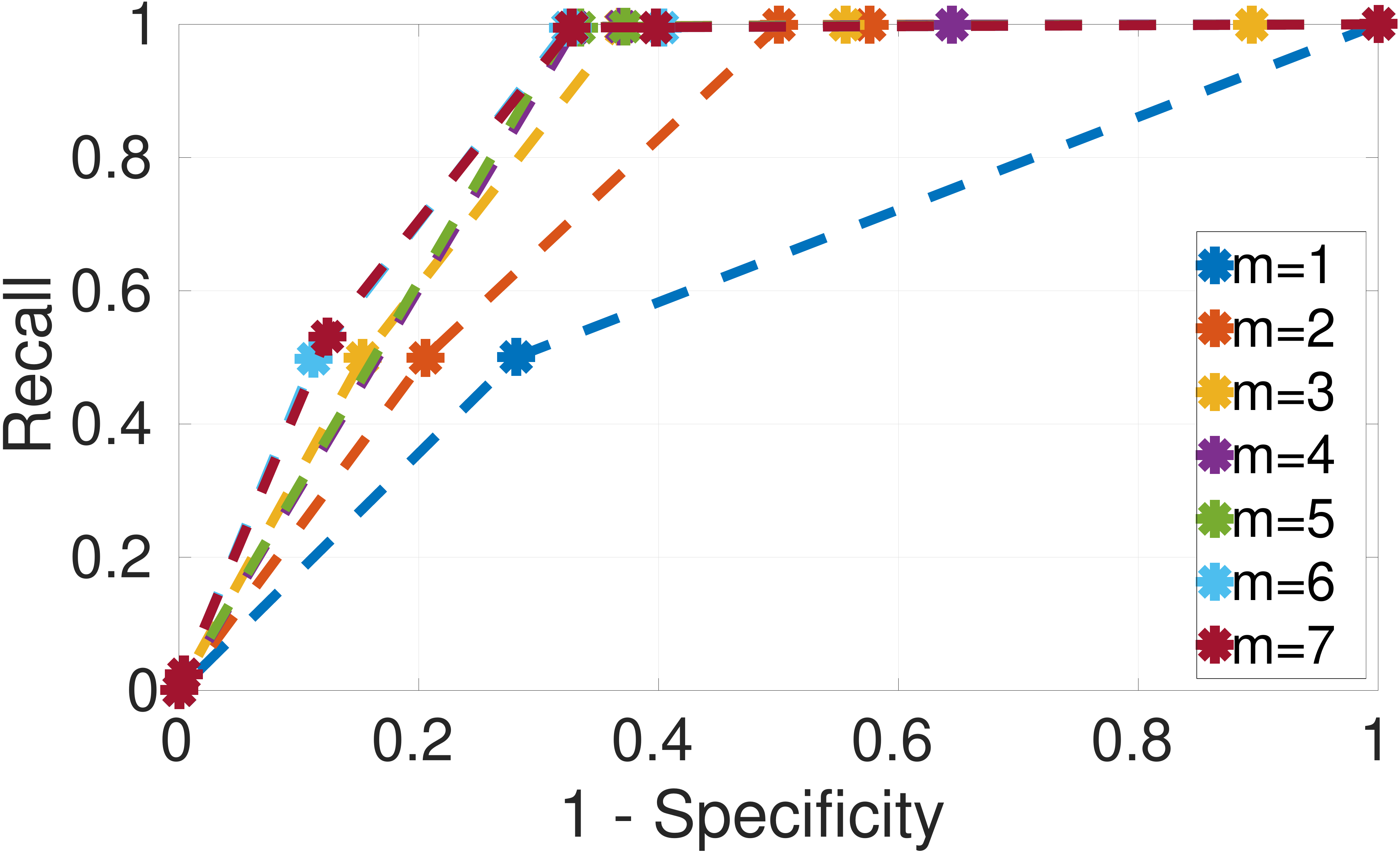}
	\caption{ROC curves for the variable-order model using the \pst\ for various values of the maximum order. $\mathit{distance} \in [0.4,0.6]$.}
	\label{fig:exp-feedzai-classification-roc2}
\end{subfigure}\\
\begin{subfigure}[t]{0.65\textwidth}
	\includegraphics[width=0.99\textwidth]{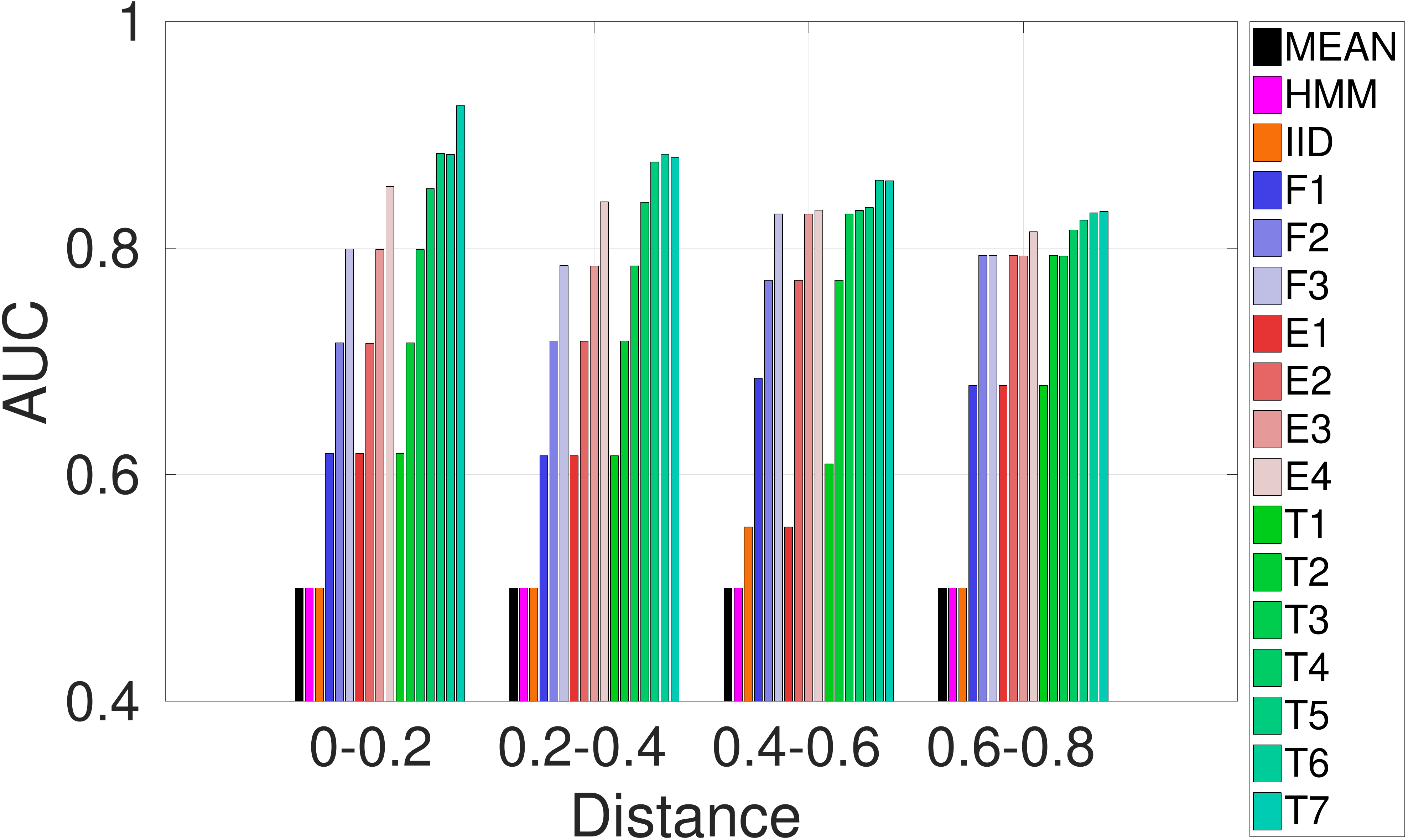}
	\caption{AUC for ROC curves for all models.}
	\label{fig:exp-feedzai-classification-aucroc}
\end{subfigure}
\caption{Results for CE forecasting from the domain of credit card fraud management. Fx stands for a Full-order Markov Model of order $x$. Ex stands for a Variable-order Markov Model of maximum order $x$ that uses a \psa\ and creates an embedding. Tx stands for a Variable-order Markov Model of maximum order $x$ that is constructed directly from a \pst. MEAN stands for the method of estimating the mean of ``waiting-times''. HMM stands for Hidden Markov Model. IID stands for the method assuming (in the worst case) that SDEs are i.i.d. Ex and Tx models are the ones proposed in this paper.}
\label{fig:exp-feedzai-classification-roc}
\end{figure}

Figure \ref{fig:exp-feedzai-sde-logloss} shows the average log-loss obtained for various models and orders $m$
and Figure \ref{fig:exp-feedzai-sde-states} shows the number of states for the full-order models or nodes for the variable-order models, 
which are prediction suffix trees.
The best result is achieved with a variable-order Markov model of maximum order 7.
The full-order Markov models are slightly better than their equivalent (same order) variable-order models. 
This is an expected behavior,
since the variable-order models are essentially approximations of the full-order ones.
We increase the order of the full-order models until $m=3$,
in which case the Markov chain has $\approx\ 750$ states.
We avoid increasing the order any further,
because Markov chains with more than 1000 states become significantly difficult to manage in terms of memory usage 
(and in terms of the computational cost of estimating the waiting-time distributions for the experiments of CE forecasting that follow).
Note that a Markov chain with $1000$ states would require a transition matrix with $1000^{2}$ entries.
On the contrary,
we can increase the maximum order of the variable-order model until we find the best one,
i.e., the order after which the average log-loss starts increasing again. 
The size of the prediction suffix tree can be allowed to increase to more than $1000$ nodes,
since we are not required to build a transition matrix.

\begin{figure*}[t]
\centering
\begin{subfigure}[t]{0.58\textwidth}
	\includegraphics[width=0.99\textwidth]{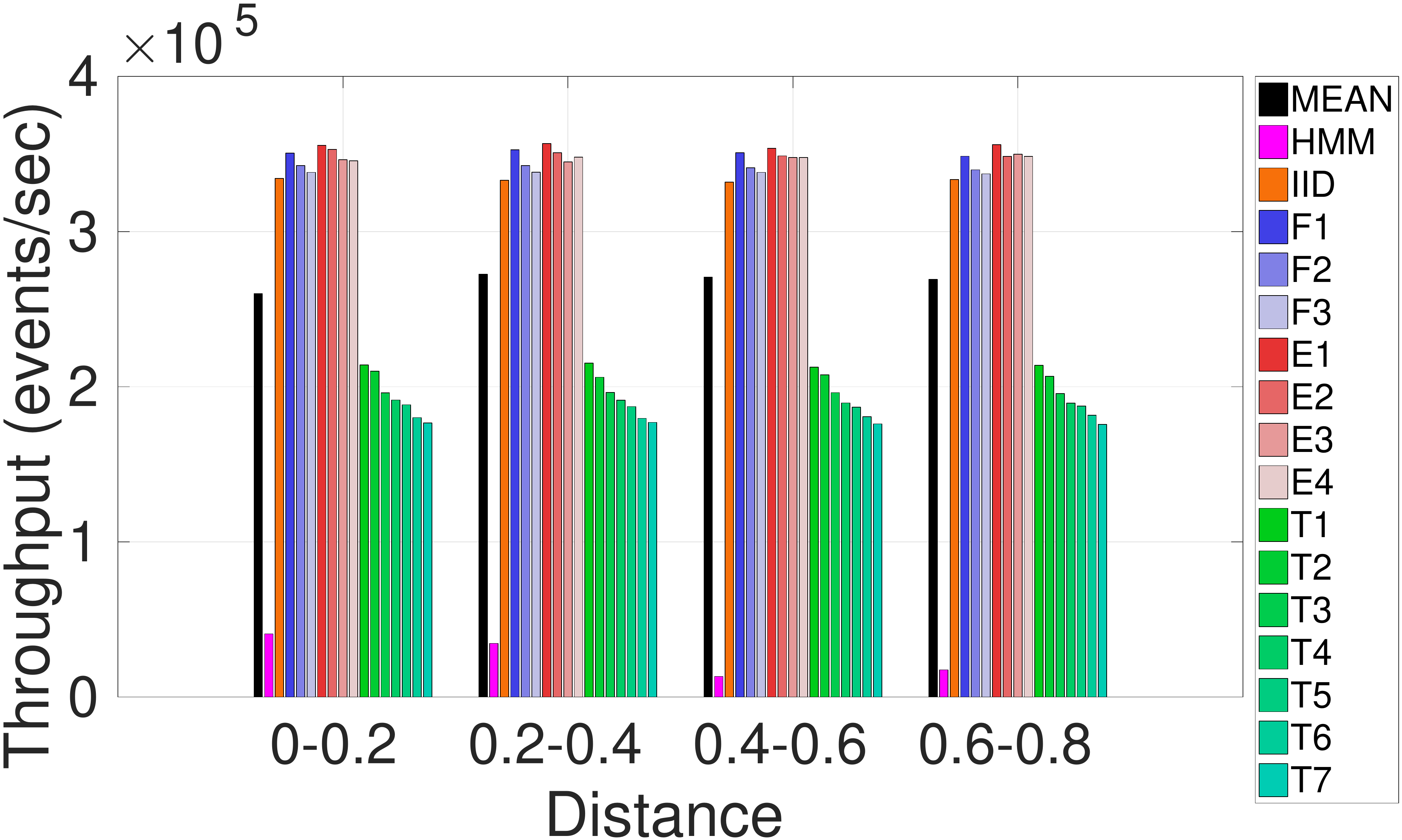}
	\caption{Throughput.}
	\label{fig:exp-feedzai-classification-throughput}
\end{subfigure} 
\begin{subfigure}[t]{0.45\textwidth}
	\includegraphics[width=0.99\textwidth]{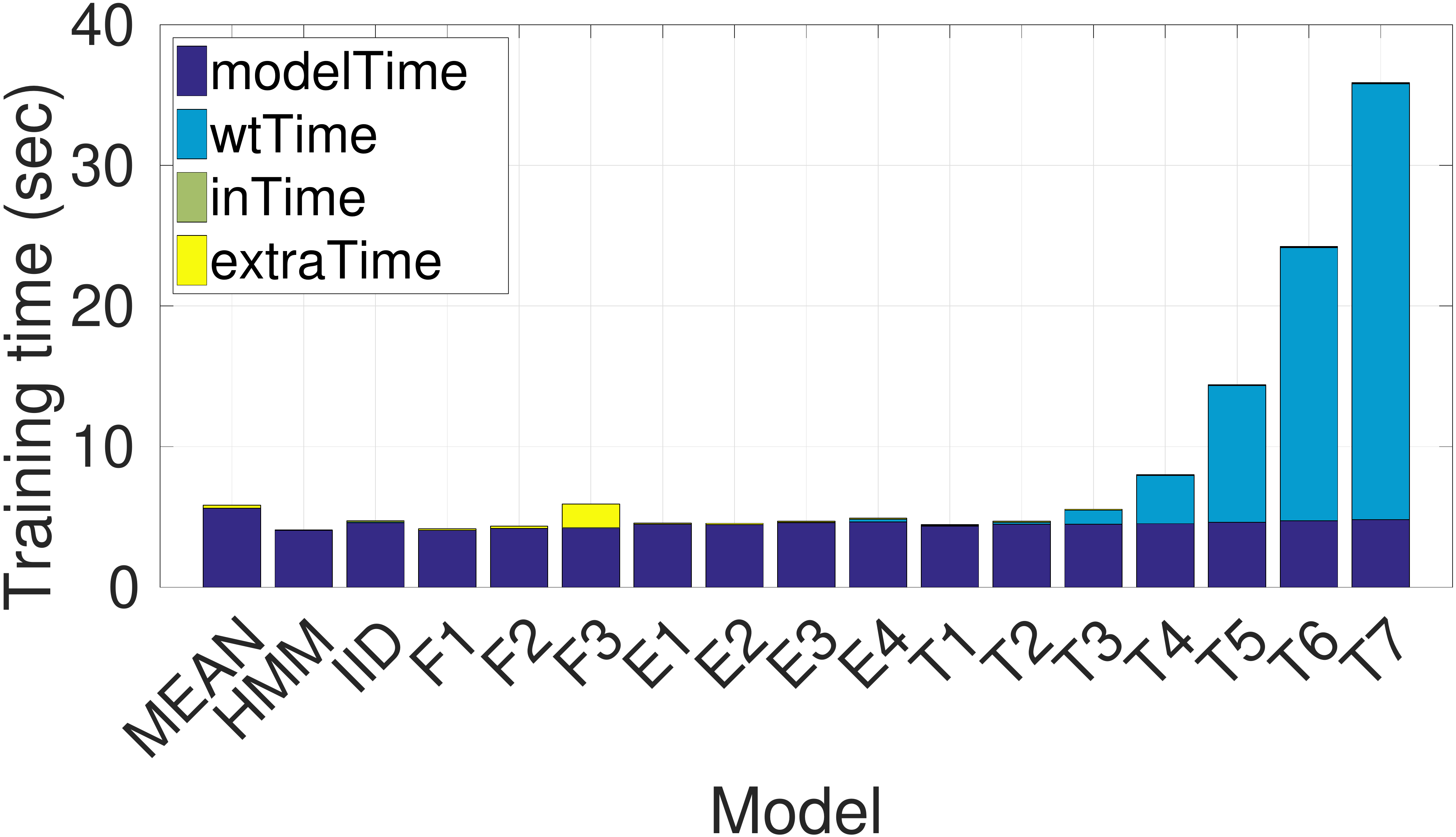}
	\caption{Training time.}
	\label{fig:exp-feedzai-classification-training}
\end{subfigure}
\begin{subfigure}[t]{0.45\textwidth}
	\includegraphics[width=0.99\textwidth]{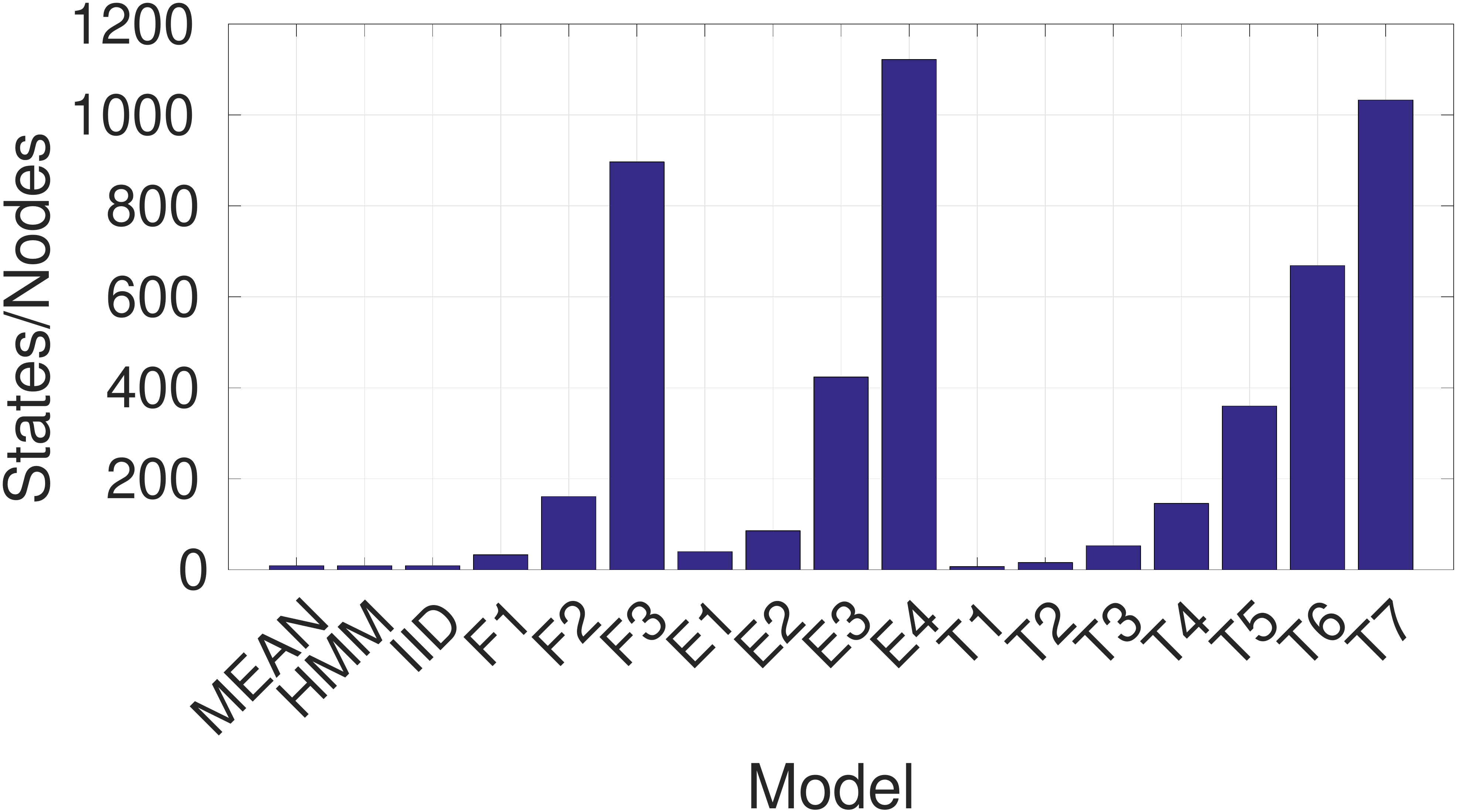}
	\caption{Number of states/nodes.}
	\label{fig:exp-feedzai-classification-states}
\end{subfigure}
\caption{Results for throughput, training time and number of automaton states/tree nodes for classification CE forecasting from the domain of credit card fraud management. modelTime = time to construct the model. wtTime = time to estimate the waiting-time distributions for all states. inTime = time to estimate the forecast of all states from their waiting-time distributions. extraTime = time to determinize an automaton (+ disambiguation time for full-order models).}
\label{fig:exp-feedzai-classification-performance}
\end{figure*}

We now move on to the classification experiments.
Figure \ref{fig:exp-feedzai-classification-roc} shows the ROC curves of the variable-order model that directly uses a \pst.
We show results for two different ``expected'' distance ranges:
$\mathit{distance} \in [0.2,0.4]$ in Figure \ref{fig:exp-feedzai-classification-roc1}
and $\mathit{distance} \in [0.4,0.6]$ in Figure \ref{fig:exp-feedzai-classification-roc2}.
The ideal operating point in the ROC is the top-left corner and thus, 
the closer to that point the curve is, the better. 
Thus, the first observation is that by increasing the maximum order we obtain better results.
Notice, however, that in Figure \ref{fig:exp-feedzai-classification-roc2},
where the distance is greater and the forecasting problem is harder,
increasing the order from 6 to 7 yields only a marginal improvement.

Figure \ref{fig:exp-feedzai-classification-aucroc} displays ROC results for different distances and all models, 
in terms of the Area Under the ROC Curve (AUC), 
which is a measure of the models' classification accuracy.
The first observation is that the MEAN and HMM methods consistently underperform, 
compared to the Markov models.
Focusing on the Markov models,
as expected,
the task becomes more challenging and the ROC scores decrease,
as the distance increases.
It is also evident that higher orders lead to better results.
The advantage of increasing the order becomes less pronounced 
(or even non-existent)
as the distance increases.
The variable-order models that use an embedding are only able to go as far as $m=4$,
due to increasing memory requirements,
whereas the tree-based versions can go up to $m=7$ 
(and possibly even further, but we did not try to extend the order beyond this point).
Although the embedding (\psa) can indeed help achieve better scores than full-order models  by reaching higher orders, 
this is especially true for the tree-based models which bypass the embedding.
We can thus conclude that full-order models are doing well up to the order that they we can achieve with them. 
\psa\ models can reach roughly the same levels, 
as they are also practically restricted. 
The performance of \pst\ models is similar to that of the other models for the same order, 
but the fact that they can use higher orders allows them to finally obtain superior performance.

We show performance results in Figure \ref{fig:exp-feedzai-classification-performance},
in terms of computation and memory efficiency.
Figure \ref{fig:exp-feedzai-classification-throughput} displays throughput results.
We can observe the trade-off between the high forecasting accuracy of the tree-based high-order models and the performance penalty that these models incur.
The models based on \pst\ have a throughput figure that is almost half that of the full-order models and the embedding-based variable-order ones.
In order to emit a forecast, 
the tree-based models
need to traverse a tree after every new event arrives at the system,
as described in Section \ref{sec:no-mc}.
The automata-based full- and variable-order models, 
on the contrary,
only need to evaluate the minterms on the outgoing transitions of their current state and simply jump to the next state.
It would be possible to improve the throughput of the tree-based models, 
by using caching techniques,
so that we can reuse some of the previously estimated forecasts,
but we reserve such optimizations for future work.
By far the worst throughput, however,
is observed for the HMM models.
The reason is that the waiting-time distributions and forecasts are always estimated online,
as explained in Section \ref{sec:test-models}.

Figure \ref{fig:exp-feedzai-classification-training} shows training times as a stacked, bar plot.
For each model,
the total training time is broken down into 4 different components,
each corresponding to a different phase of the forecast building process.
\emph{modelTime} is the time required to actually construct the model from the training dataset.
\emph{wtTime} is the time required to estimate the waiting-time distributions, 
once the model has been constructed.
\emph{inTime} measures the time required to estimate the forecast of each waiting-time distribution.
Finally, \emph{extraTime} measures the time required to determinize the automaton of our initial pattern.
For the full-order Markov models,
it also includes the time required to convert the deterministic automaton into its equivalent, disambiguated automaton. 
We observe that the tree-based models exhibit significantly higher times than the rest,
for high orders.
The other models have similar training times,
almost always below $5$ seconds.
Thus, if we need high accuracy,
we again have to pay a price in terms of training time.
Even in the case of high-order tree-based models though,
the training time is almost half a minute for a training dataset composed of 750,000 transactions,
which allows us to be confident that training could be performed online.

Figure \ref{fig:exp-feedzai-classification-states} shows the memory footprint of the models in terms of the size of their basic data structures.
For automata-based methods, 
we show the number of states,
whereas for the tree-based methods we show the number of nodes.
We see that variable-order models, 
especially the tree-based ones,
are significantly more compact than the full-order ones, 
for the same order.
We also observe that the tree-based methods,
for the same order,
are much more compact (fewer nodes) than the ones based on the embedding (more states).
This allows us to increase the order up to $7$ with the tree-based approach,
but only up to $4$ with the embedding.

\subsection{Maritime Situational Awareness}
\label{sec:maritime}

The second dataset that we used in our experiments is a real--world dataset coming from the field of maritime monitoring.
It is composed of a set of trajectories from ships sailing at sea,
emitting AIS (Automatic Identification System) messages that relay information about their position, heading, speed, etc.,
as described in the running example of Section \ref{sec:example}.
These trajectories can be analyzed, 
using the techniques of Complex Event Recognition,
in order to detect interesting patterns in the behavior of vessels \cite{DBLP:journals/geoinformatica/PatroumpasAAVPT17}.
The dataset that we used is publicly available, contains AIS kinematic messages from vessels sailing in the Atlantic Ocean around the port of Brest, France, and spans a period from 1 October 2015 to 31 March 2016 \cite{ray_cyril_2018_1167595}.
We used a derivative dataset that contains clean and compressed trajectories, 
consisting only of critical points \cite{patroumpas_2018_2563256}.
Critical points are the important points of a trajectory that indicate a significant change in the behavior of a vessel.
Using critical points, 
one can reconstruct quite accurately the original trajectory \cite{DBLP:journals/geoinformatica/PatroumpasAAVPT17}.
We further processed the dataset by interpolating between the critical points in order to produce trajectories where two consecutive points have a temporal distance of exactly 60 seconds. 
The reason for this pre-processing step is that AIS messages typically arrive at unspecified time intervals.
These intervals can exhibit a very wide variation,
depending on many factors (e.g., human operators may turn on/off the AIS equipment),
without any clear pattern that could be encoded by our probabilistic model.
Consequently, our system performs this interpolation as a first step.

\begin{figure}[t]
	\centering
	\includegraphics[width=0.8\textwidth]{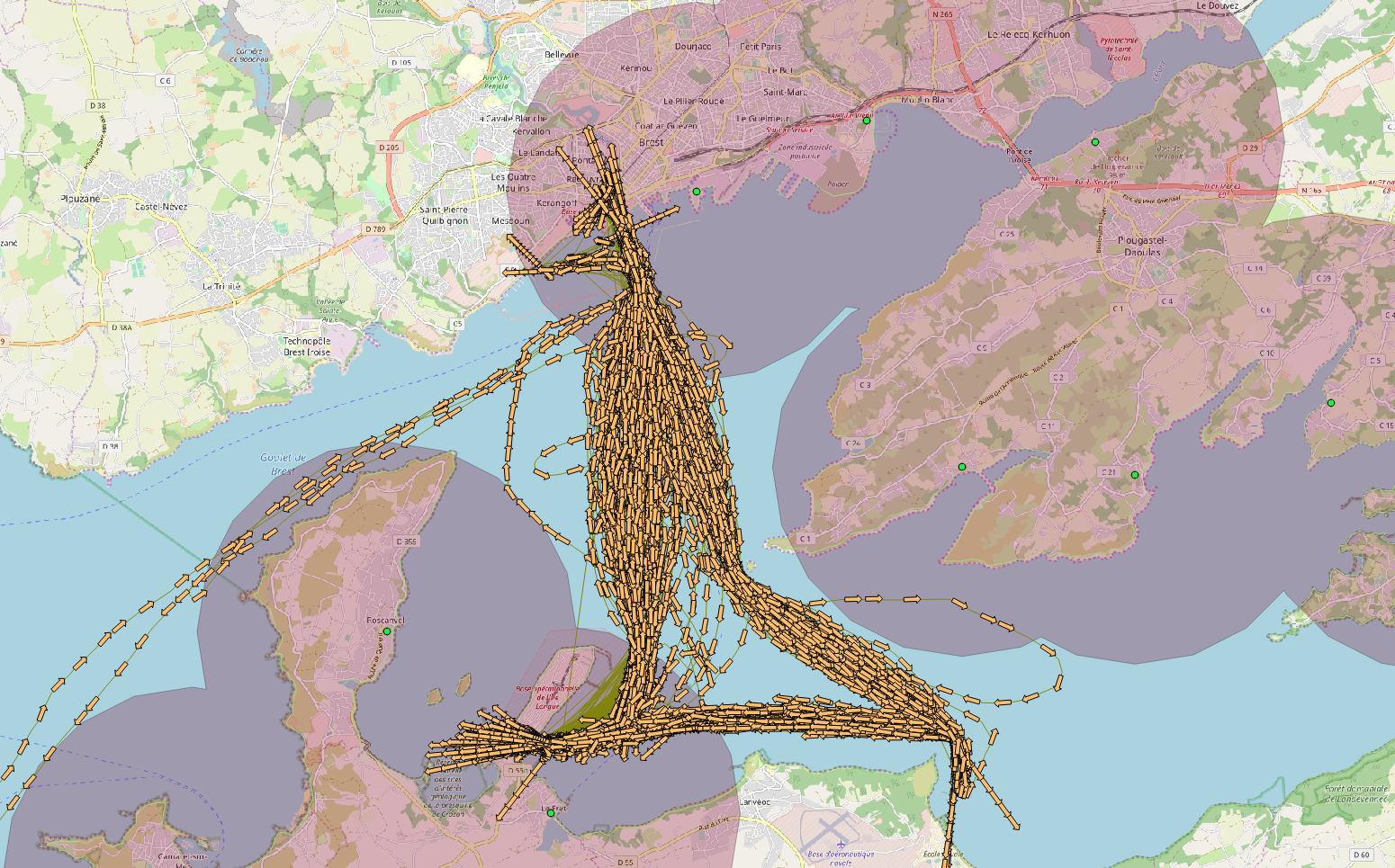}
	\caption{Trajectories of the vessel with the most matches for Pattern \eqref{exp:port} around the port of Brest. 
	Green points denote ports. Red, shaded circles show the areas covered by each port. They are centered around each port and have a radius of 5 km.}
	\label{fig:port-traffic}
\end{figure}

\begin{figure}[t]
\begin{subfigure}[t]{0.45\textwidth}
	\includegraphics[width=0.95\textwidth]{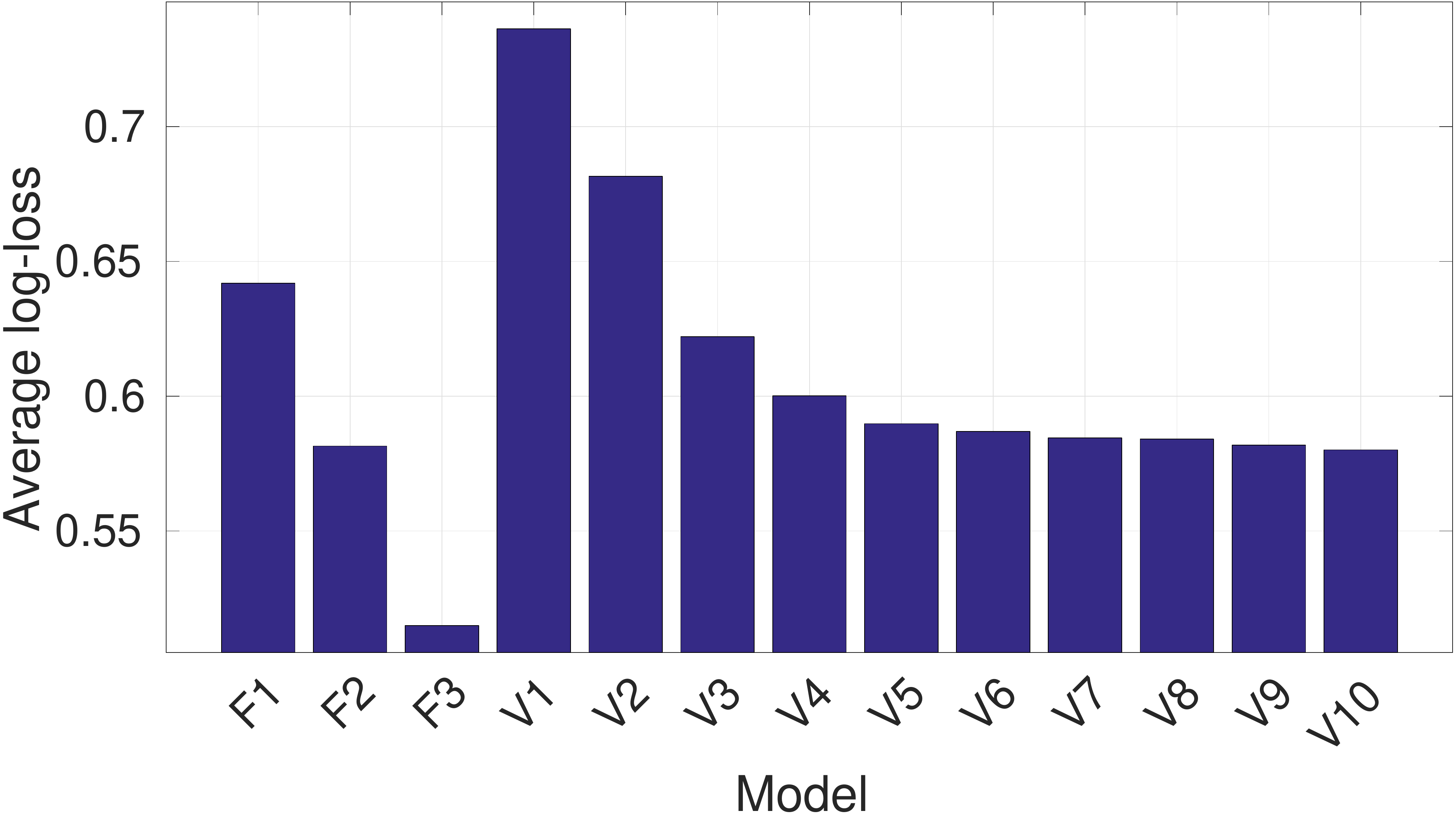}
	\caption{Average log-loss.}
	\label{fig:exp-maritime-sde-logloss}
\end{subfigure}
\begin{subfigure}[t]{0.45\textwidth}
	\includegraphics[width=0.95\textwidth]{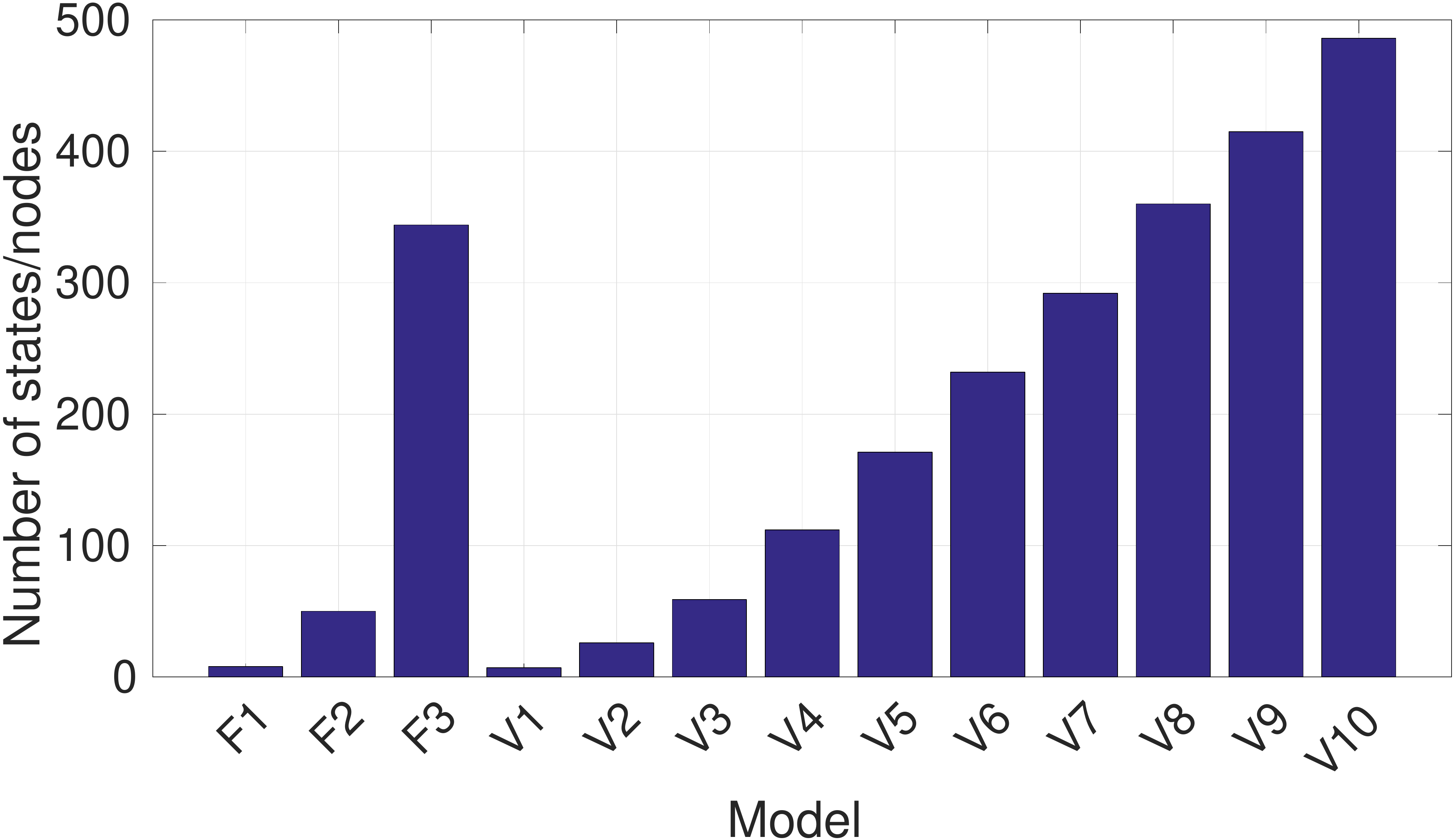}
	\caption{Number of states/nodes.}
	\label{fig:exp-maritime-sde-states}
\end{subfigure}
\caption{Results for SDE forecasting from the domain of maritime monitoring. Fx stands for a Full-order Markov Model of order $x$. Vx stands for a Variable-order Markov Model (a prediction suffix tree) of maximum order $x$.}
\label{fig:exp-maritime-sde}
\end{figure}

The pattern that we used  in the experiments is a movement pattern in which a vessel approaches the main port of Brest.
The goal is to forecast when a vessel will enter the port.
This way, 
port traffic management may be optimized, 
in order to reduce the carbon emissions of vessels waiting to enter the port.
The symbolic regular expression for this pattern is the following:
\begin{equation}
\label{exp:port}
\begin{aligned}
R := \ & \ (\neg \mathit{InsidePort(Brest)})^{*} \cdot (\neg \mathit{InsidePort(Brest)}) \cdot  \\
\ &\ (\neg \mathit{InsidePort(Brest)}) \cdot (\mathit{InsidePort(Brest)})
\end{aligned}
\end{equation}
The intention is to detect the entrance of a vessel in the port of Brest.
The predicate $\mathit{InsidePort(Brest)}$ evaluates to \true\ whenever a vessel has a distance of less than 5 km from the port of Brest (see Figure \ref{fig:port-traffic}).
In fact,
the predicate is generic and takes as arguments the longitude and latitude of any point,
but we show here a simplified version, 
using the port of Brest, 
for reasons of readability.
The pattern defines the entrance to the port as a sequence of at least 3 consecutive events, 
only the last of which satisfies the $\mathit{InsidePort(Brest)}$ predicate.
In order to detect an entrance,
we must first ensure that the previous event(s) indicated that the vessel was outside the port.
For this reason,
we require that,
before the last event,
there must have occurred at least 2 events where the vessel was outside the port.
We require 2 or more such events to have occurred (instead of just one), 
in order to avoid detecting ``noisy'' entrances.

\begin{figure}[t]
\centering
\begin{subfigure}[t]{0.48\textwidth}
	\includegraphics[width=0.99\textwidth]{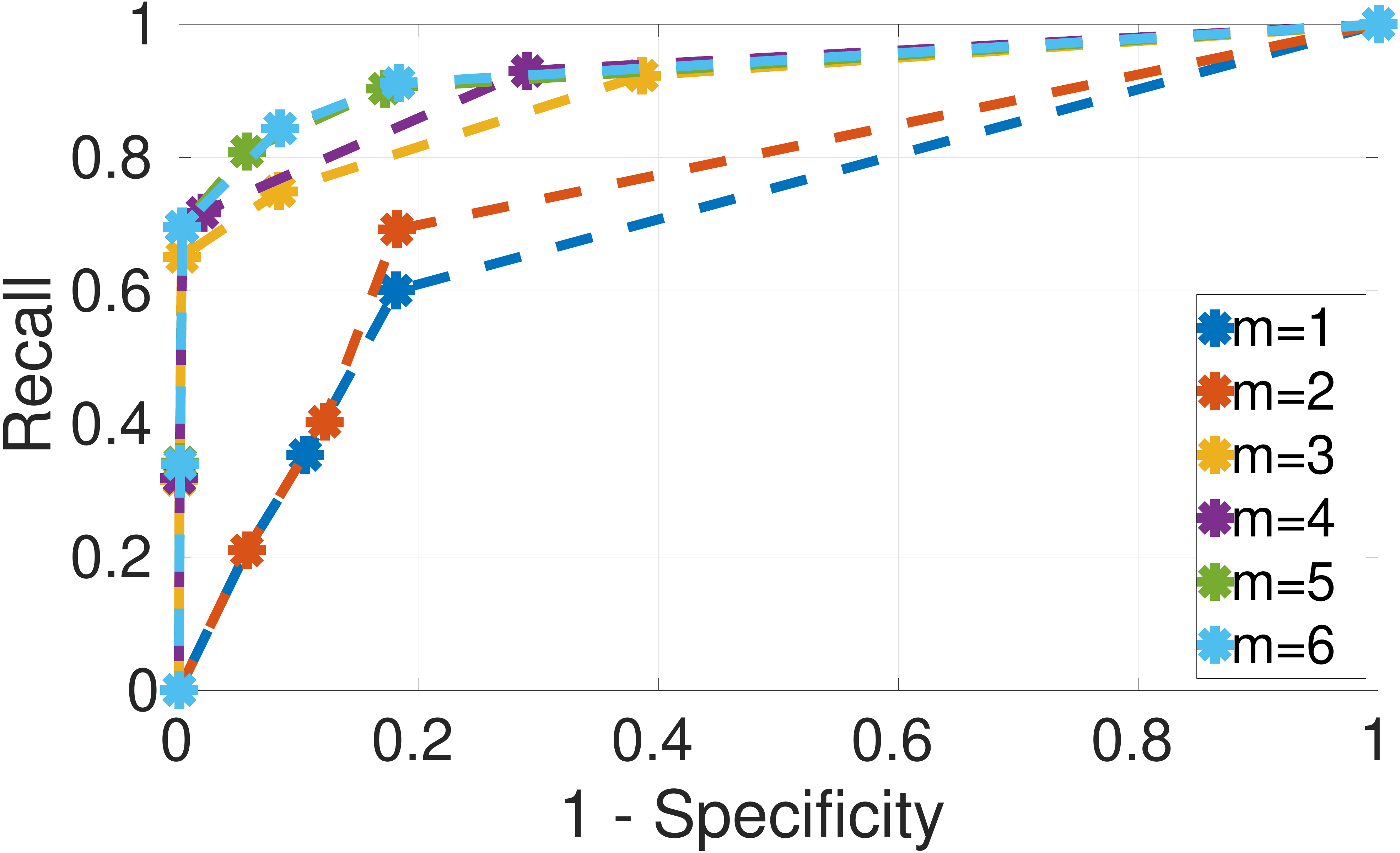}
	\caption{ROC curves for the variable-order model using the \pst\ for various values of the maximum order. $\mathit{distance} \in [0.0,0.5]$.}
	\label{fig:exp-maritime-classification-roc1}
\end{subfigure}
\begin{subfigure}[t]{0.48\textwidth}
	\includegraphics[width=0.99\textwidth]{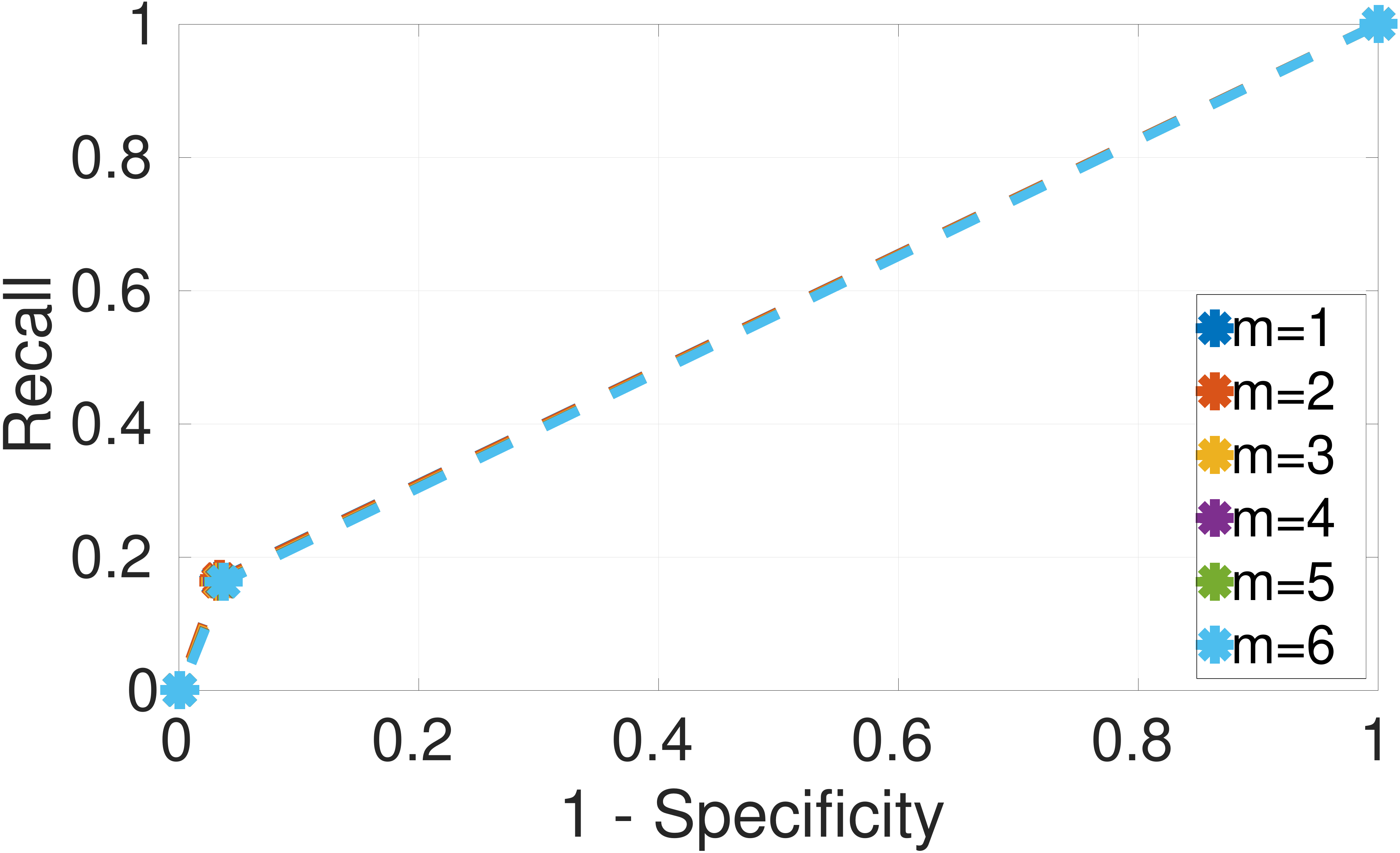}
	\caption{ROC curves for the variable-order model using the \pst\ for various values of the maximum order. $\mathit{distance} \in [0.5,1.0]$.}
	\label{fig:exp-maritime-classification-roc2}
\end{subfigure}\\
\begin{subfigure}[t]{0.65\textwidth}
	\includegraphics[width=0.99\textwidth]{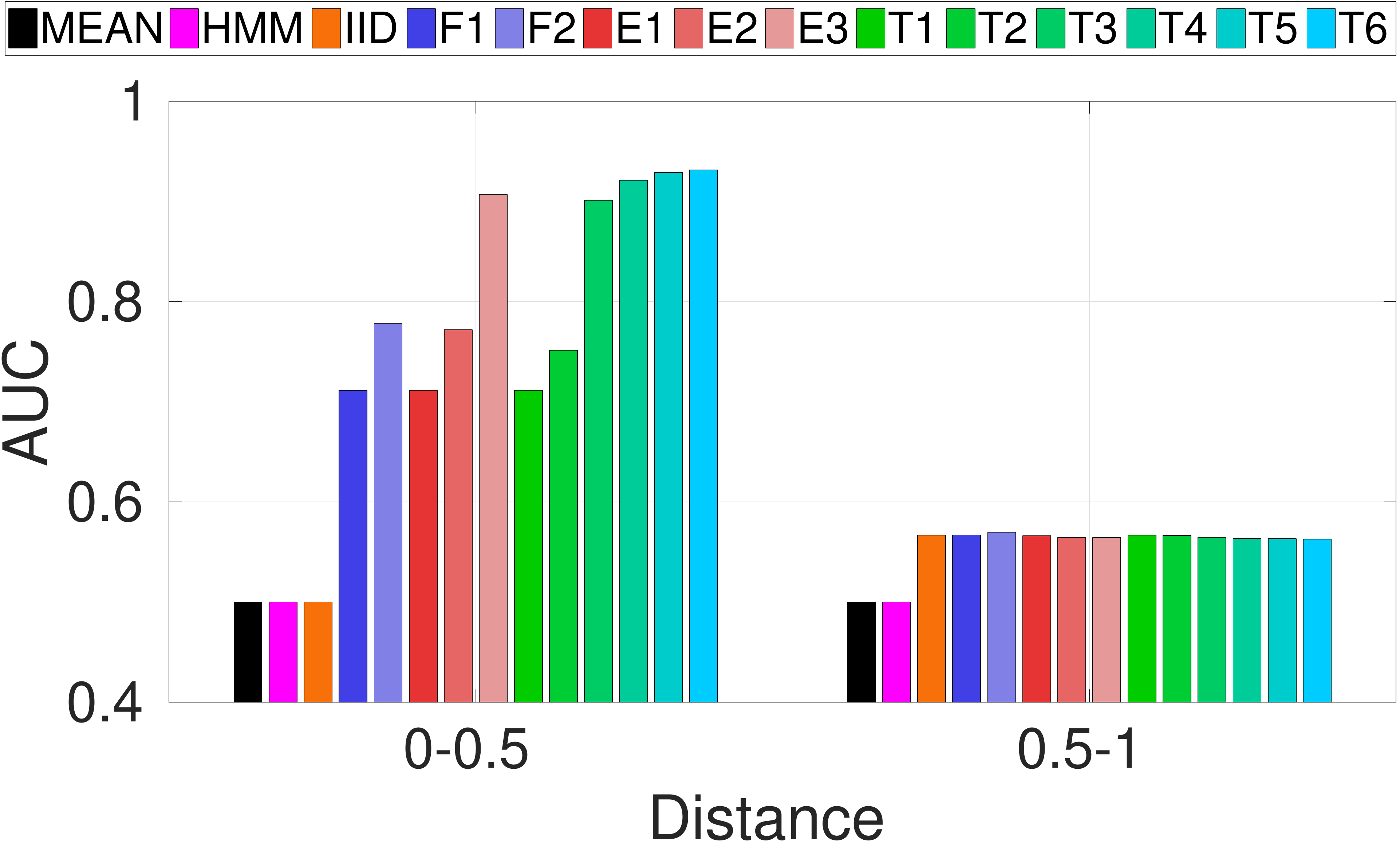}
	\caption{AUC for ROC curves for all models.}
	\label{fig:exp-maritime-classification-aucroc}
\end{subfigure}
\caption{Results for CE forecasting in the domain of maritime situational awareness. Fx stands for a Full-order Markov Model of order $x$. Vx stands for a Variable-order Markov Model of maximum order $x$ using an embedding. Tx stands for a Variable-order Markov Model of maximum order $x$ using a prediction suffix tree. MEAN stands for the method of estimating the mean of ``waiting-times''. HMM stands for Hidden Markov Model. Ex and Tx models are the ones proposed in this paper.}
\label{fig:exp-maritime-classification-roc}
\end{figure}

In addition to the $\mathit{InsidePort(Brest)}$ predicate, 
we included 5 extra ones providing information about the distance of a vessel from a port when it is outside the port.
Each of these predicates evaluates to \true\ when a vessel lies within a specified range of distances from the port.
The first returns \true\ when a vessel has a distance between 5 and 6 km from the port, 
the second when the distance is between 6 and 7 km and the other three extend similarly 1 km until 10 km.
We investigated the sensitivity of our models to the presence of various extra predicates in the recognition pattern. 

For all experimental results that follow,
we always present average values over 4 folds of cross-validation.
We start by analyzing the trajectories of a single vessel and then move to multiple, selected vessels.
There are two issues that we tried to address by separating our experiments into single-vessel and multiple-vessel ones.
First, we wanted to have enough data for training.
For this reason, 
we only retained vessels for which we can detect a significant number of matches for Pattern \eqref{exp:port}.
Second, our system can work in two modes:
a) it can build a separate model for each monitored object and use this collection of models for personalized forecasting;
b) it can build a global model out of all the monitored objects.
We thus wanted to examine whether building a global model from multiple vessels could produce equally good results, 
as these obtained for a single vessel with sufficient training data.

We first used Pattern \eqref{exp:port} to perform recognition on the whole dataset in order to find the number of matches detected for each vessel.
The vessel with the most matches was then isolated and we retained only the events emitted from this vessel.
In total, we detected 368 matches for this vessel and the number of SDEs corresponding to it is $\approx$ 30.000. 
Figure \ref{fig:port-traffic} shows the isolated trajectories for this vessel,
seemingly following a standard route between various ports around the Brest area.

\begin{figure}[t]
\begin{subfigure}[t]{0.45\textwidth}
	\includegraphics[width=0.95\textwidth]{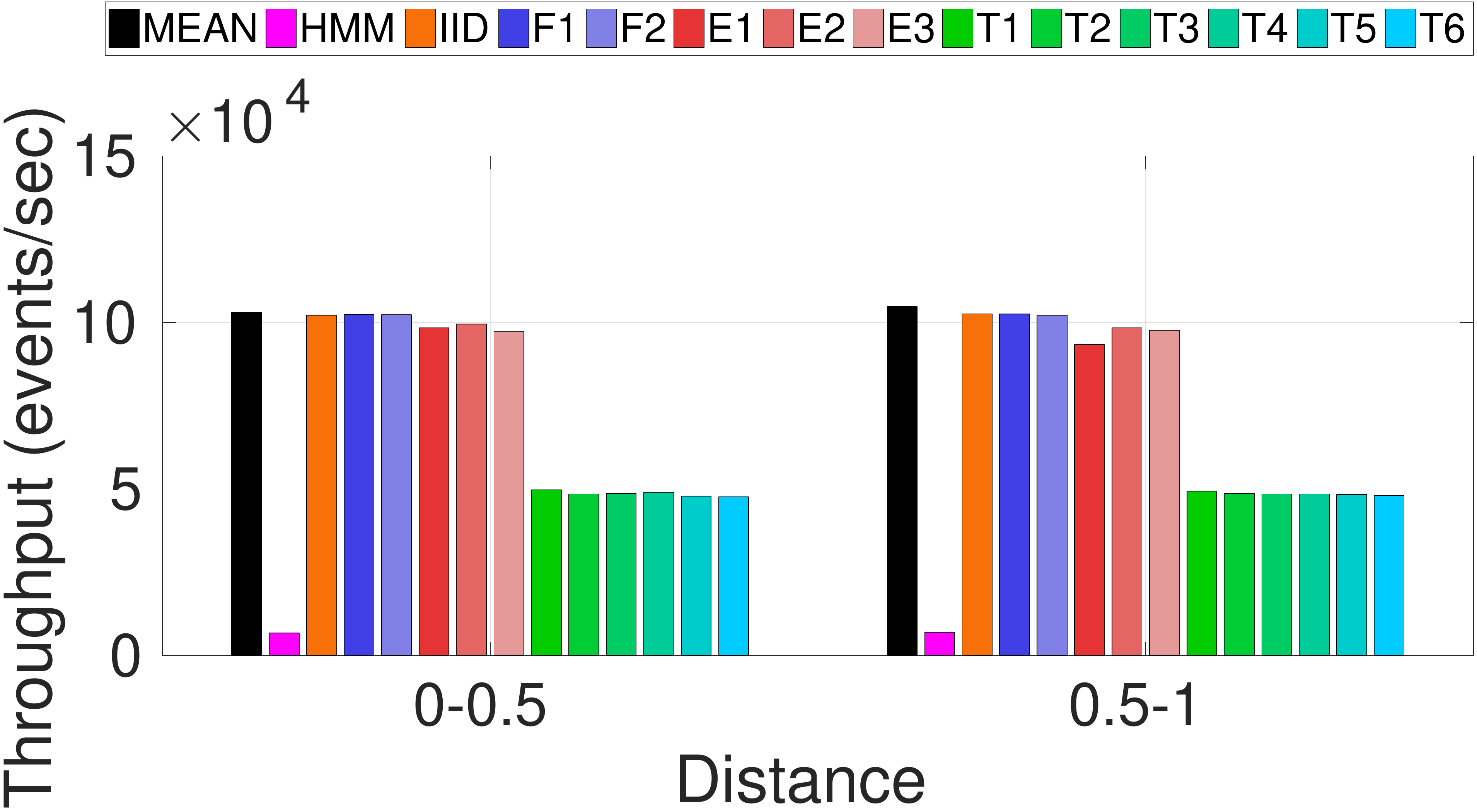}
	\caption{Throughput.}
	\label{fig:exp-maritime-classification-throughput}
\end{subfigure}
\begin{subfigure}[t]{0.45\textwidth}
	\includegraphics[width=0.95\textwidth]{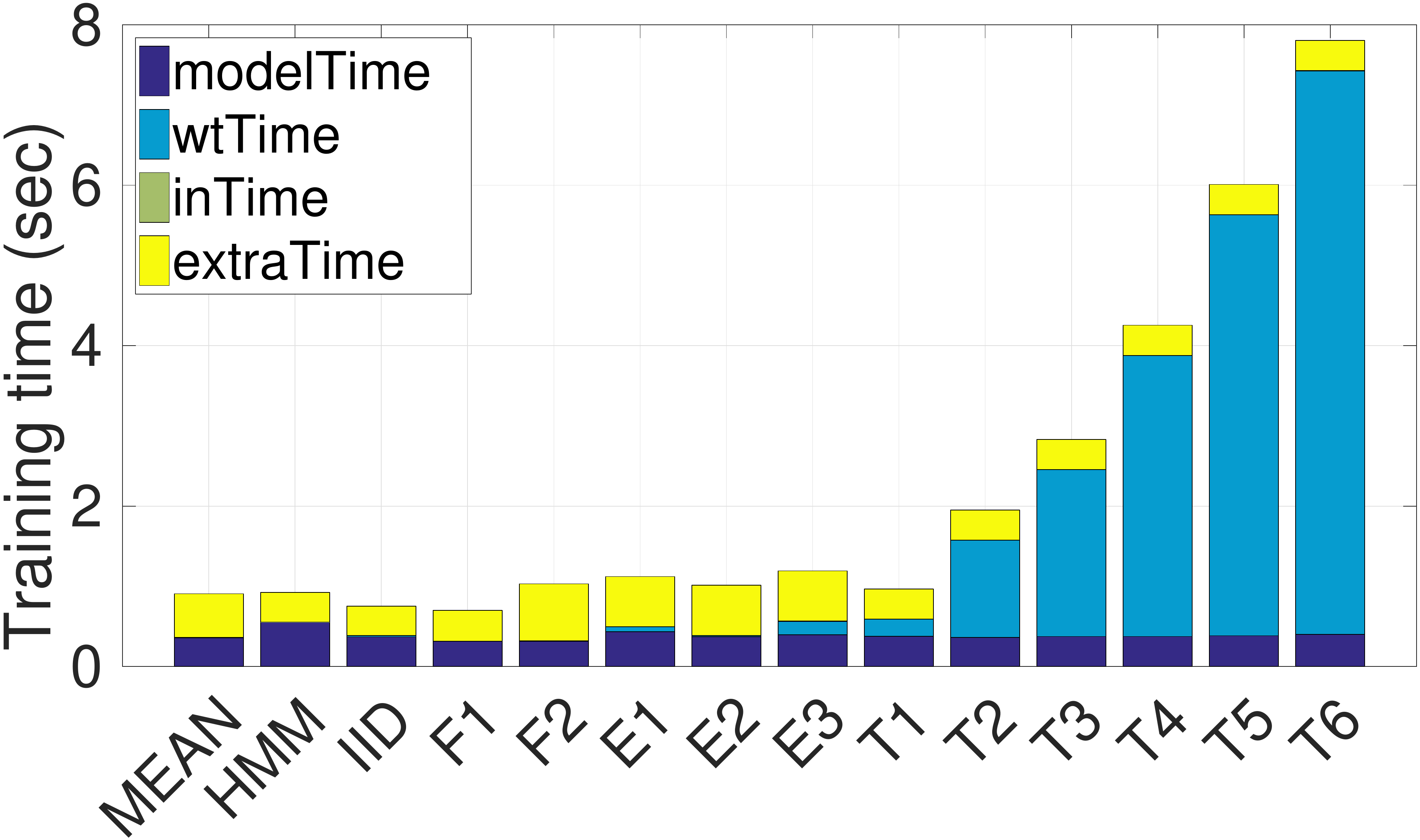}
	\caption{Training time.}
	\label{fig:exp-maritime-classification-training}
\end{subfigure}
\caption{Throughput and training time results for classification CE forecasting for maritime situational awareness. modelTime = time to construct the model. wtTime = time to estimate the waiting-time distributions for all states. inTime = time to estimate the forecast interval of all states from their waiting-time distributions. extraTime = time to determinize an automaton (+ disambiguation time for full-order models).}
\label{fig:exp-maritime-classification-performance}
\end{figure}

Figure \ref{fig:exp-maritime-sde} shows results for SDE forecasting.
The best average log-loss is achieved with a full-order Markov model, with $m=3$, and is $\approx 0.51$.
For the best hyper-parameter values out of those that we tested for the variable-order model,
with $m=10$,
we can achieve an average log-loss that is $\approx 0.57$.
Contrary to the case of credit card data,
increasing the order of the variable-order model does not allow us to achieve a better log-loss score than the best one achieved with a full-oder model.
However, as we will show,
this does not imply that the same is true for CE forecasting.

Using the vessel of Figure \ref{fig:port-traffic}, 
we obtained the results shown in Figures \ref{fig:exp-maritime-classification-roc} and \ref{fig:exp-maritime-classification-performance}.
Since the original \dsfa\ is smaller in this case 
(one start and one final state plus two intermediate states),
we have fewer distance ranges 
(e.g., there no states in the range $[0.4,0.6]$).
Thus, we use only two distance ranges: $[0,0.5]$ and $[0.5,1]$.
We observe the importance of being able to increase the order of our models for distances smaller than $50\%$. 
For distances greater than $50\%$,
the area under curve is $\approx$ 0.5 for all models.
This implies that they cannot effectively differentiate between positives and negatives.
Their forecasts are either all positive,
where we have $\mathit{Recall}=100\%$ and $\mathit{Specificity}=0\%$,
or all negative,
where we have $\mathit{Recall}=0\%$ and $\mathit{Specificity}=100\%$
(see Figure \ref{fig:exp-maritime-classification-roc2}).
Notice that the full-order Markov models can now only go up to $m=2$,
since the existence of multiple extra predicates makes it prohibitive to increase the order any further. 
Achieving higher accuracy with higher-order models comes at a computational cost, 
as shown in Figure \ref{fig:exp-maritime-classification-performance}. 
The results are similar to those in the credit card experiments.
The training time for variable-order models tends to increase as we increase the order,
but is always less than 8 seconds.
The effect on throughput is again significant for the tree-based variable-order models.
Throughput figures are also lower here compared to the credit card fraud experiments,
since the predicates that we need to evaluate for every new input event
(like $\mathit{InsidePort(Brest)}$) involve more complex calculations
(the $\mathit{amountDiff}>0$ predicate is a simple comparison).

As a next step,
we wanted to investigate the effect of the optimization technique mentioned at the end of Section \ref{sec:no-mc} on the accuracy and performance of our system.
The optimization prunes future paths whose probability is below a given cutoff threshold.
We re-run the experiments described above for distances between $0\%$ and $50\%$ for various values of the cutoff threshold,
starting from $0.0001$ up to $0.2$.  
Figure \ref{fig:exp-maritime-classification-cutoff} shows the relevant results.
We observe that the accuracy is affected only for high values of the cutoff threshold, above $0.1$ 
(Figure \ref{fig:exp-maritime-classification-cutoff-auc}).
We can also see that throughput remains essentially unaffected (Figure \ref{fig:exp-maritime-classification-cutoff-throughput}).
This result is expected,
since the cutoff threshold is only used in the estimation of the waiting-time distributions. 
Throughput reflects the online performance of our system,
after the waiting-time distributions have been estimated, 
and is thus not affected by the choice of the cutoff threshold.
However, 
the training time is indeed significantly affected
(Figure \ref{fig:exp-maritime-classification-cutoff-training}).
As expected,
the result of increasing the value of the cutoff threshold is a reduction of the training time,
as fewer paths are retained.
Beyond a certain point though,
further increases of the cutoff threshold affect the accuracy of the system.
Therefore, the cutoff threshold should be below $0.01$ so as not to compromise the accuracy of our forecasts.

\begin{figure}[t]
\centering
\begin{subfigure}[t]{0.48\textwidth}
	\includegraphics[width=0.95\textwidth]{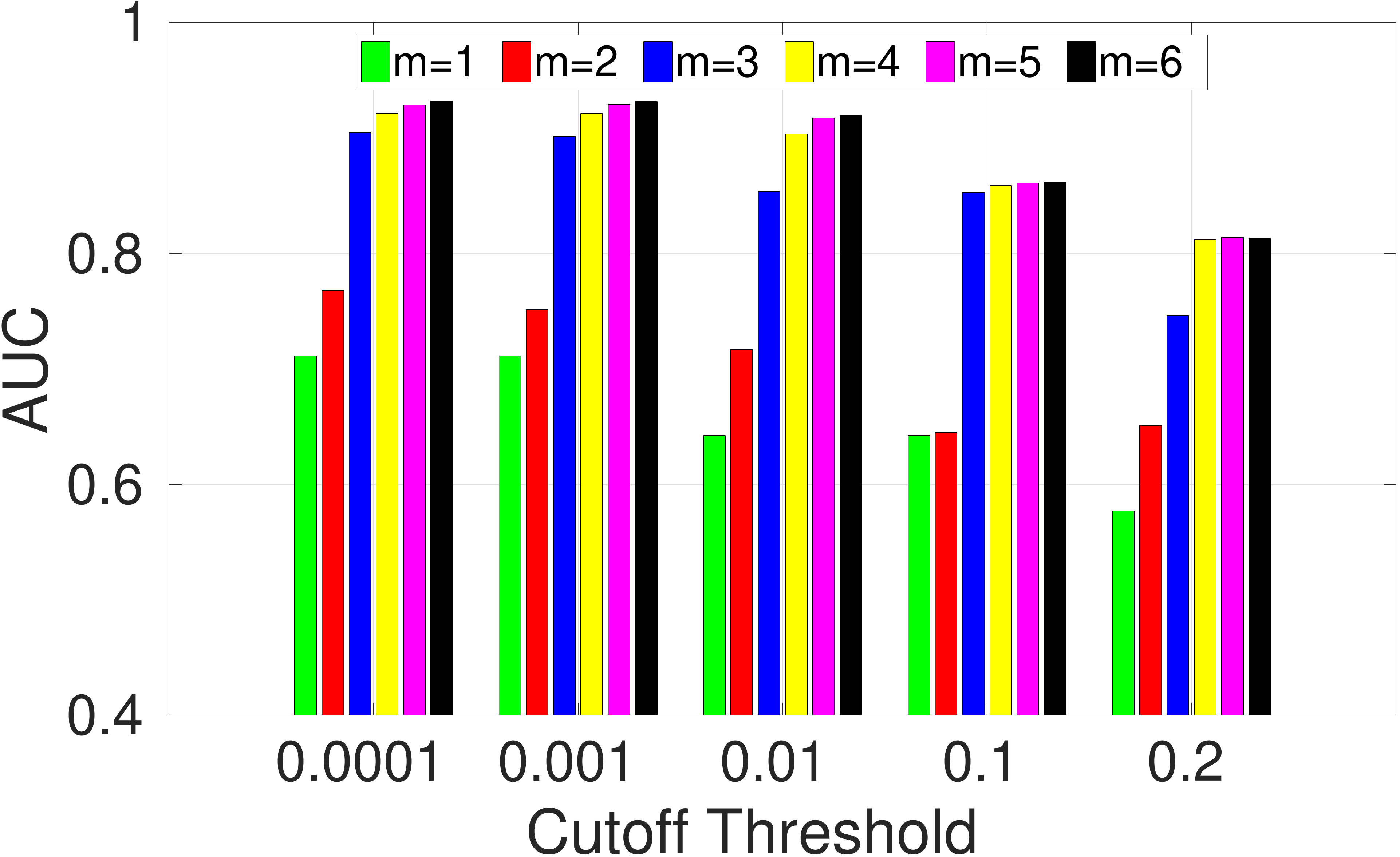}
	\caption{AUC-ROC.}
	\label{fig:exp-maritime-classification-cutoff-auc}
\end{subfigure}\\
\begin{subfigure}[t]{0.48\textwidth}
	\includegraphics[width=0.95\textwidth]{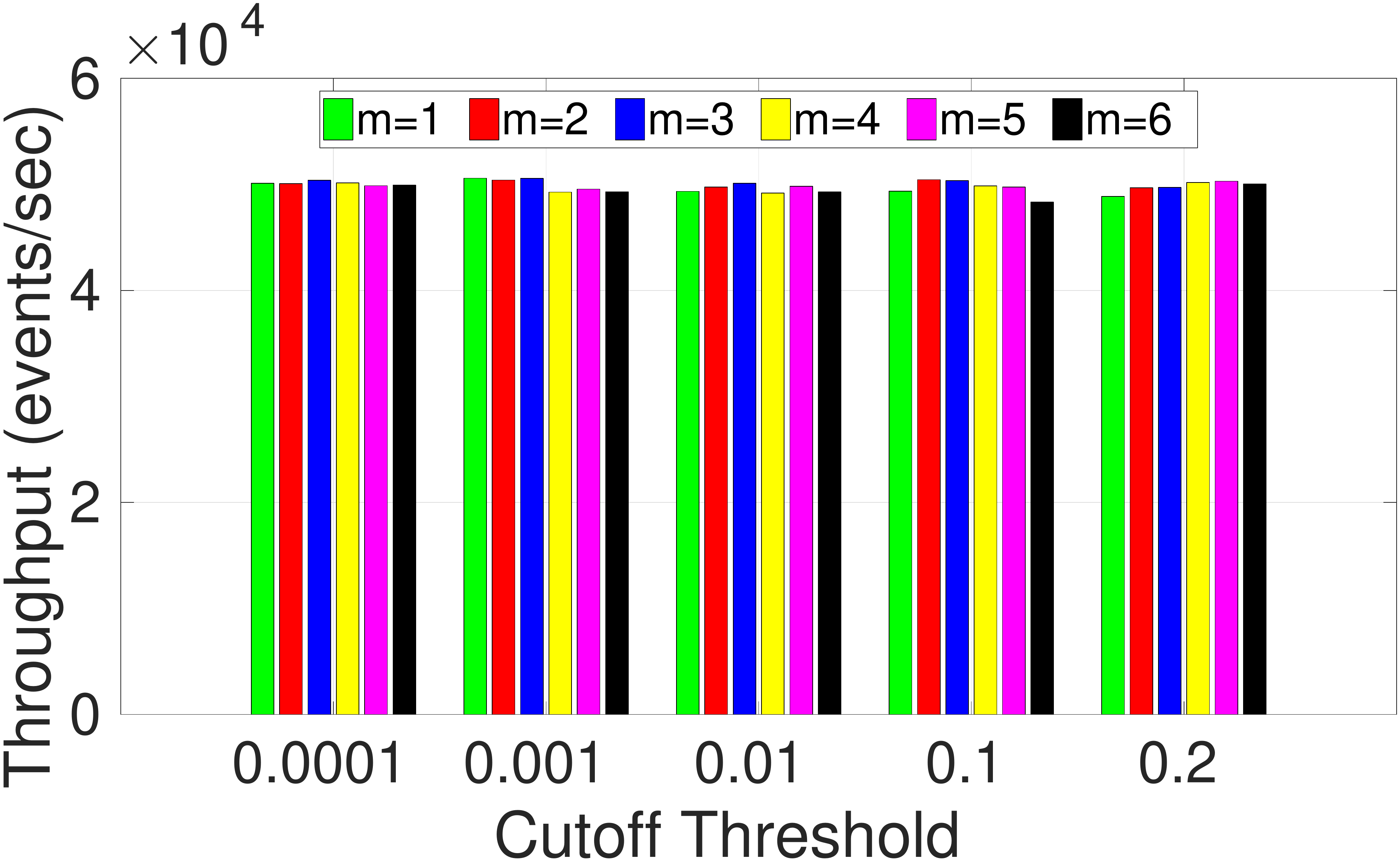}
	\caption{Throughput.}
	\label{fig:exp-maritime-classification-cutoff-throughput}
\end{subfigure}
\begin{subfigure}[t]{0.48\textwidth}
	\includegraphics[width=0.95\textwidth]{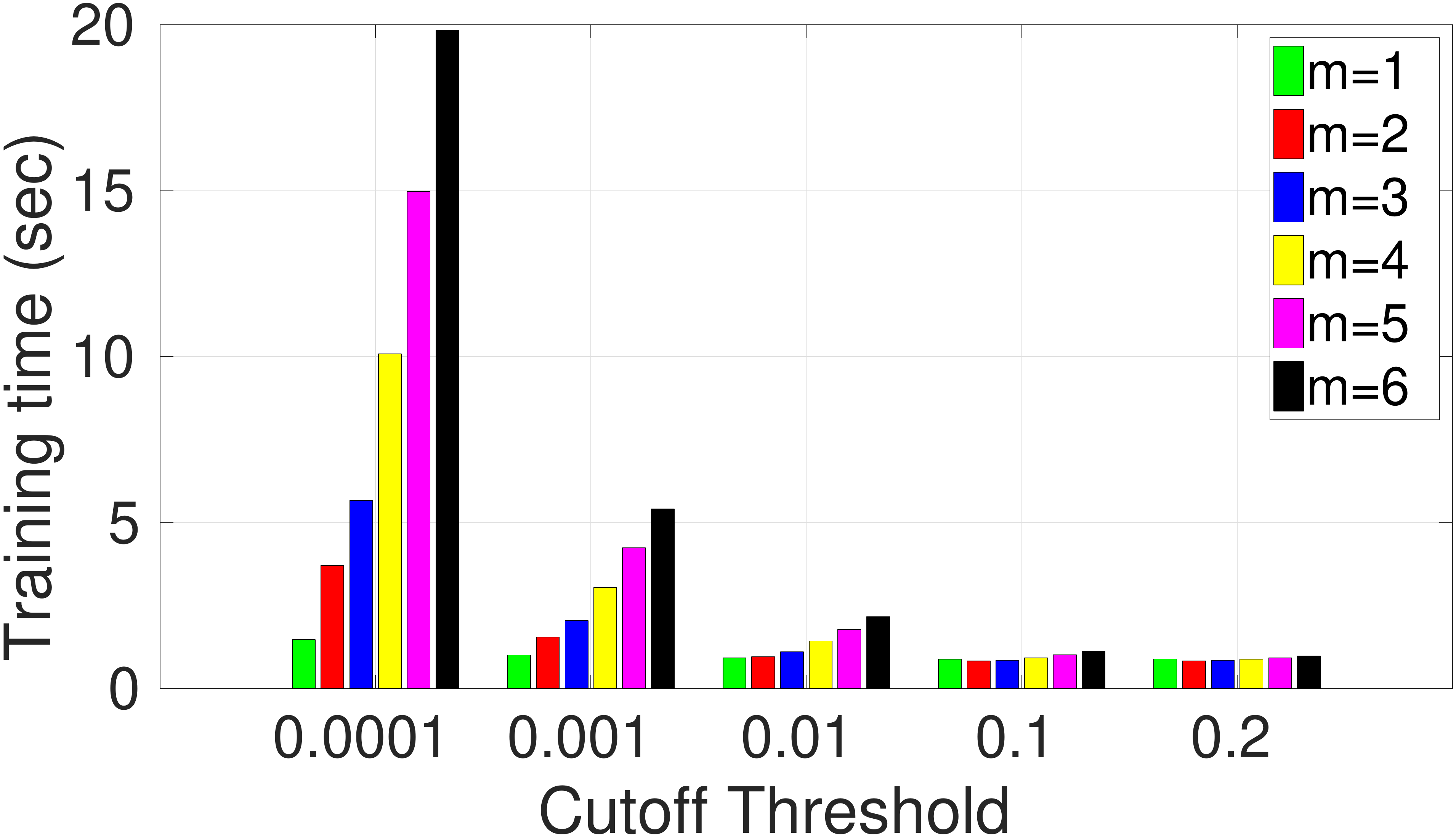}
	\caption{Training time.}
	\label{fig:exp-maritime-classification-cutoff-training}
\end{subfigure}
\caption{Effect of cutoff threshold on accuracy, throughput and training time.}
\label{fig:exp-maritime-classification-cutoff}
\end{figure}

We additionally investigated the sensitivity of our approach to the extra predicates that are used.
Figure \ref{fig:exp-maritime-classification-aucroc-distance3} shows results when the extra 5 predicates referring to the distance of a vessel from the port are modified so that each ``ring'' around the port has a width of 3 km, 
instead of 1 km. 
With these extra features,
increasing the order does indeed make an important difference,
but only when the order becomes high (5 and beyond),
which is possible only by using the tree-based variable-order models.
Moreover, the best score achieved is still lower than the best score achieved with ``rings'' of 1 km (Figure \ref{fig:exp-maritime-classification-aucroc}).
``Rings'' of 1 km are thus more appropriate as predictive features.
We also wanted to investigate whether other information about the vessel's movement could affect forecasting. 
In particular, 
we kept the 1 km ``rings'' and we also added a predicate to check whether the heading of a vessel points towards the port.
More precisely, we used the vessel's speed and heading to project its location 1 hour ahead in the future and then checked whether this projected segment and the circle around the port intersect.
The intuition for adding this feature is that the knowledge of whether a vessel is heading towards the port has predictive value.
As shown in Figure \ref{fig:exp-maritime-classification-aucroc-distanceheading}, 
this additional information led to higher scores even with low full-order orders
(compare to Figure \ref{fig:exp-maritime-classification-aucroc}).
The heading feature is indeed important.
On the other hand, the high-order models that did not use this feature seemed to be able to compensate for the missing information about the vessel's heading by going into higher orders.
A plateau is thus reached, 
which cannot be ``broken'' with the heading information. 
Notice that here we can only go up to $m=1$ for full-order models.
The inclusion of the heading predicate leads to an increase of the number of states beyond $1200$.

\begin{figure}[t]
\centering
\begin{subfigure}[t]{0.46\textwidth}
	\includegraphics[width=0.9\textwidth]{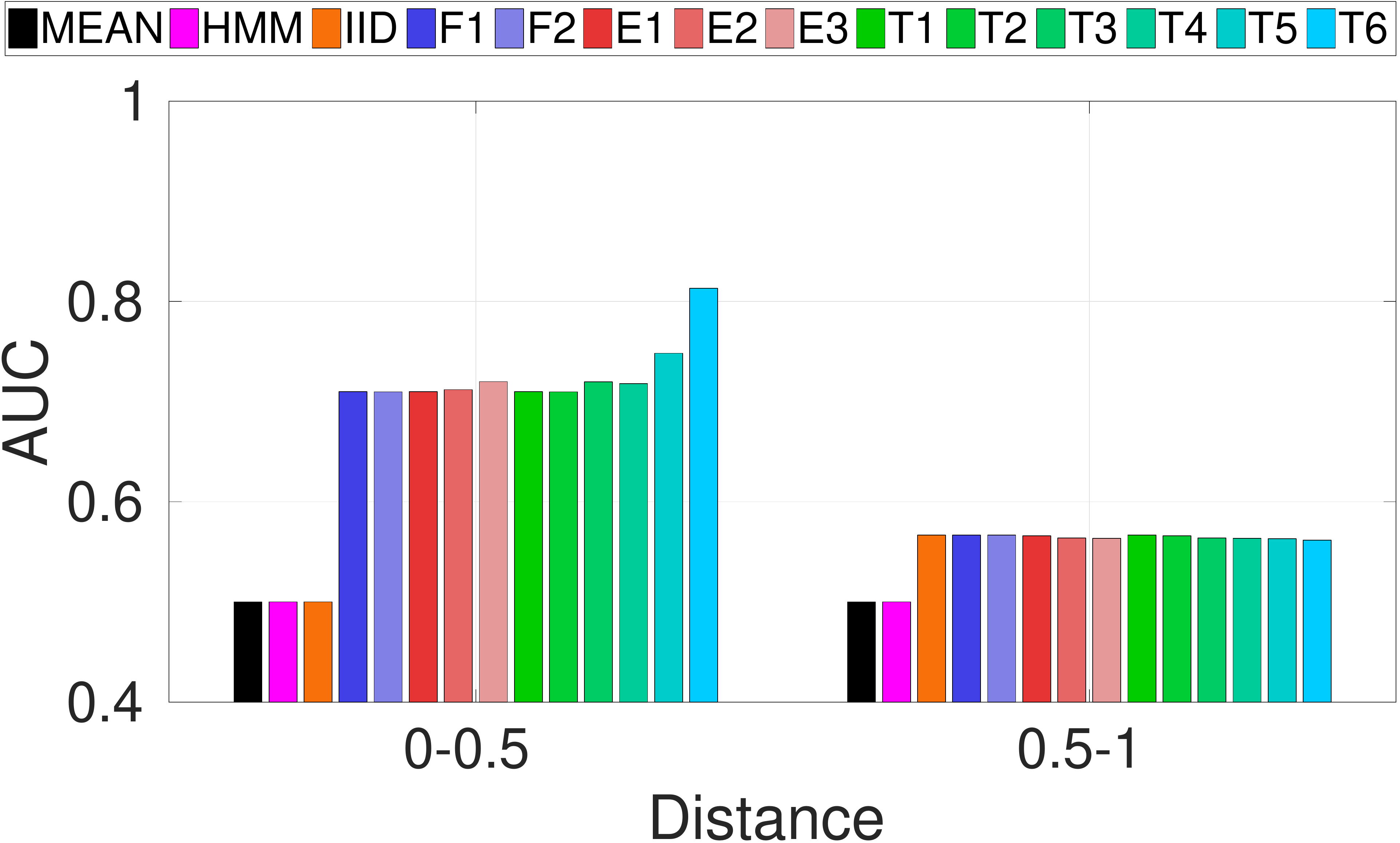}
	\caption{AUC for ROC curves. Extra features included: concentric rings around the port every 3 km. Single vessel.}
	\label{fig:exp-maritime-classification-aucroc-distance3}
\end{subfigure}\\
\begin{subfigure}[t]{0.46\textwidth}
	\includegraphics[width=0.9\textwidth]{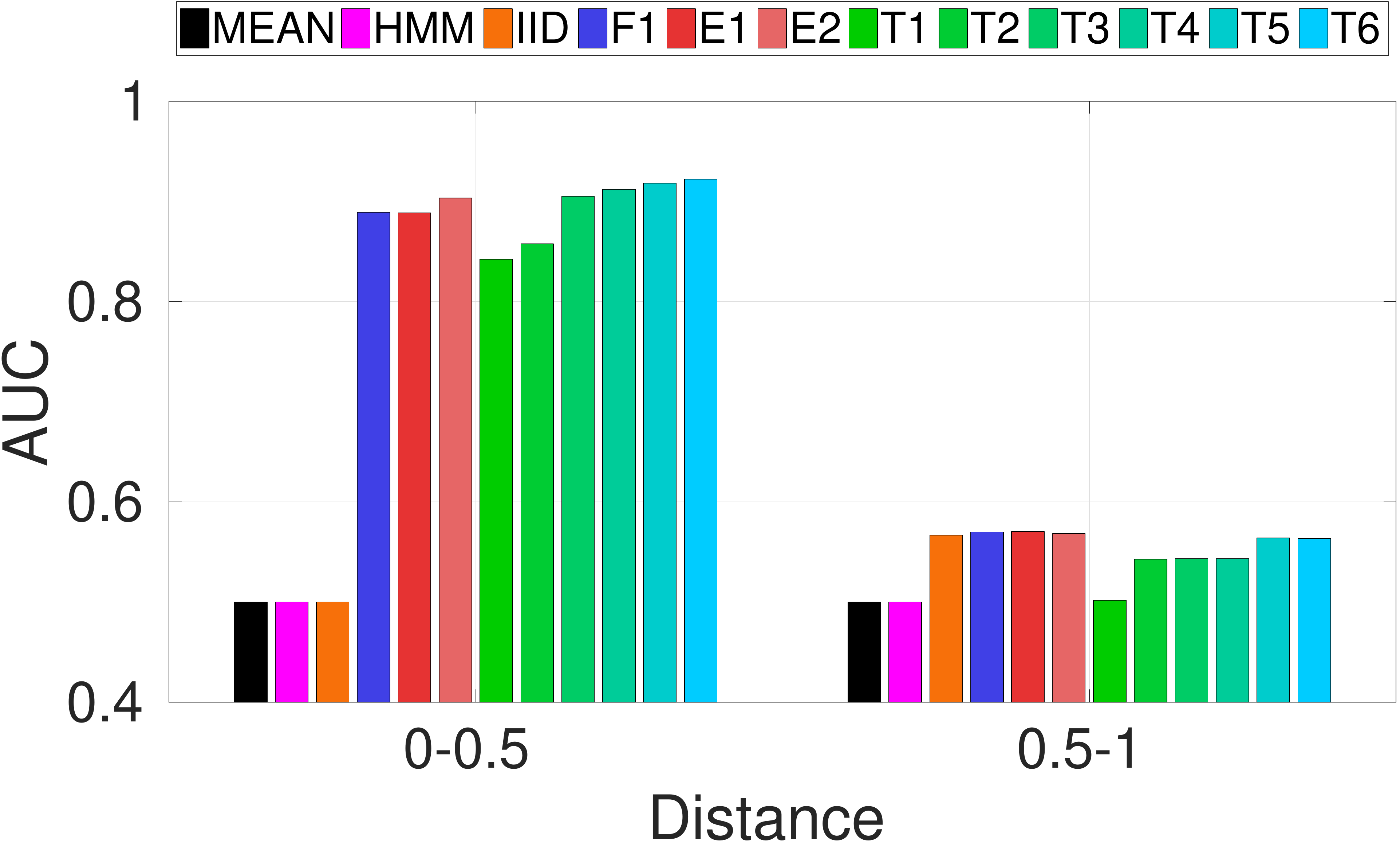}
	\caption{AUC for ROC curves. Extra features included: concentric rings around the port every 1 km and heading. Single vessel.}
	\label{fig:exp-maritime-classification-aucroc-distanceheading}
\end{subfigure}
\begin{subfigure}[t]{0.46\textwidth}
	\includegraphics[width=0.9\textwidth]{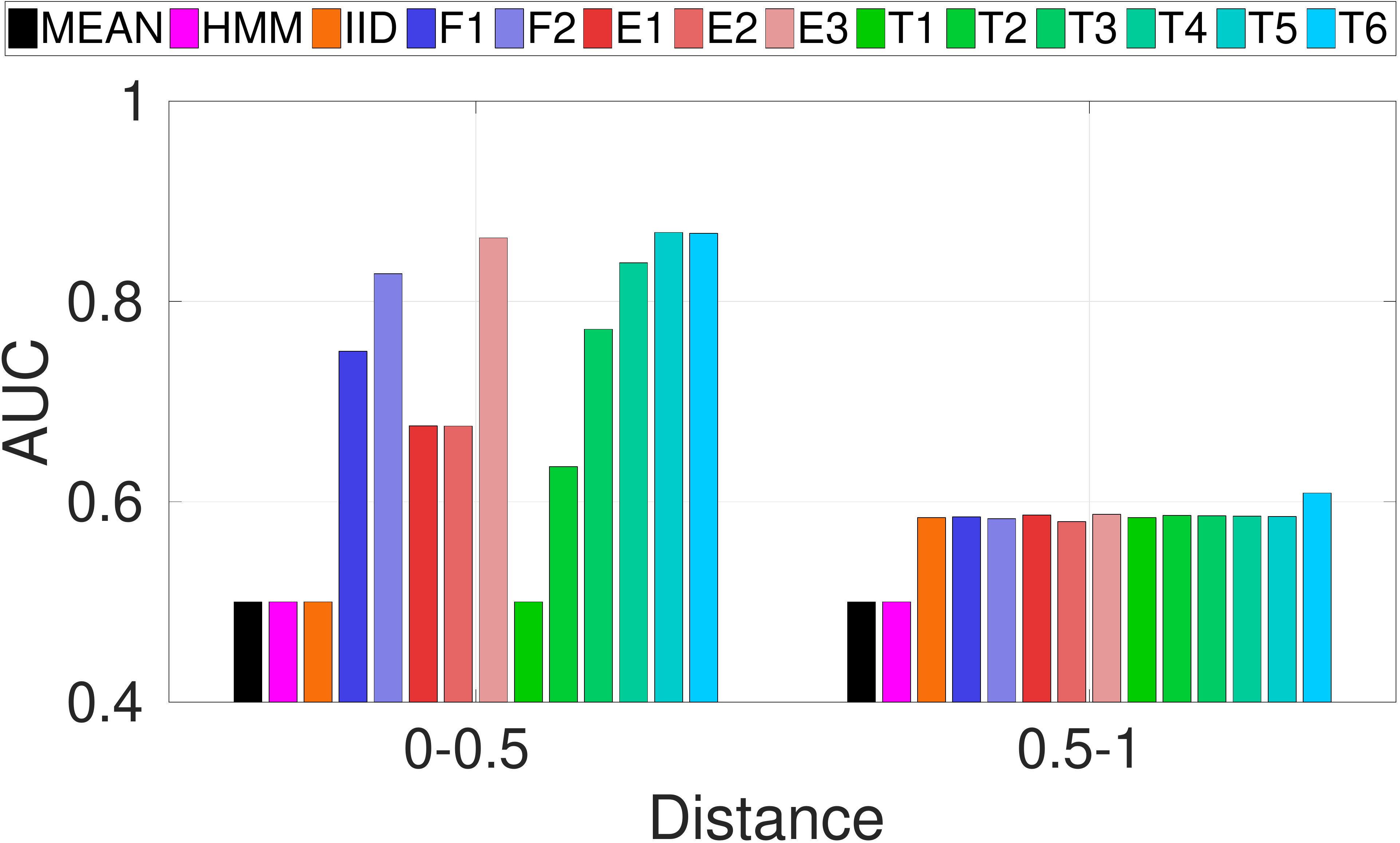}
	\caption{AUC for ROC curves. Extra features included: concentric rings around the port every 1 km. Model constructed for the 9 vessels that have more than 100 matches.}
	\label{fig:exp-maritime-classification-aucroc-multi}
\end{subfigure}
\caption{Results for classification CE forecasting from the domain of maritime monitoring for various sets of extra features and for multiple vessels.}
\label{fig:exp-maritime-classification-auc-features}
\end{figure}

Finally, we also tested our method when more than one vessel need to be monitored.
Instead of isolating the single vessel with the most matches,
we isolated all vessels which had more than 100 matches.
There are in total 9 such vessels in the dataset.
The resulting dataset has $\approx$ 222.000 events.
Out of the 9 retained vessels,
we constructed a global probabilistic model and produced forecasts.
An alternative option would be to build a single model for each vessel,
but in this scenario we wanted to test the robustness of our aprroach when a global model is built from multiple entities.
Figure \ref{fig:exp-maritime-classification-aucroc-multi} presents the corresponding results.
Interestingly, the scores of the global model remain very close to the scores of the experiments for the single vessel with the most matches (Figure \ref{fig:exp-maritime-classification-aucroc}).
This is an indication of the ability of the global model to capture the peculiarities of individual vessels.

\section{Summary \& Future Work}
\label{sec:outro}

We have presented a framework for Complex Event Forecasting (CEF), 
based on a variable-order Markov model.
It allows us to delve deeper into the past and capture long-term dependencies,
not feasible with full-order models.
Our comprehensive evaluation on two application domains has illustrated the advantages of being able to use such high-order models.
Namely, the use of higher-order modeling allows us to achieve higher accuracy than what is possible with full-order models or other state-of-the-art solutions.
We have described two alternative ways in which variable-order models may be used,
depending on the imposed requirements.
One option is to use a highly efficient but less accurate model,
when online performance is a top priority.
We also provide an option that achieves high accuracy scores,
but with a performance cost.
Another important feature of our proposed framework is that it requires minimal intervention by the user.
A given Complex Event pattern is declaratively defined and subsequently automatically translated to an automaton and then to a Markov model,
without requiring domain knowledge that should guide the modeling process.

Still, the user needs to set up the model and there seems to be room for further automation. 
In particular, 
the user needs to set the maximum order allowed by the probabilistic model.
Additionally, we have started investigating ways to handle concept drift by continuously training and updating the probabilistic model of a pattern.

With respect to the expressive power of our framework,
there is one functionality that we do not currently support and whose incorporation would require us to move to a more advanced automaton model. 
This is the functionality of applying $n$-ary (with $n>1$) predicates to two or more sub-expressions of an expression,
instead of only unary predicates,
as is allowed in symbolic automata.
Note that $n$ refers to the number of ``terminal symbols'' / events that a predicate may reference. 
Each transition predicate of a symbolic automaton may refer only to one event,
the last event consumed.
It cannot apply a predicate to the last event and other earlier events.
As an example,
consider the pattern $R := x \cdot y\ \where\ y.\mathit{speed} > x.\mathit{speed}$,
detecting an increase in the speed of a vessel,
where we now need to use the variables $x$ and $y$.
Such patterns cannot be captured with \sfa\,
since they would require a memory structure to store some of the past events of a stream,
as is possible with extended symbolic automata \cite{DBLP:journals/fmsd/DAntoniV15}. 
We intend to present in future work an automaton model which can support patterns with memory,
suitable for CER.
Some results towards this direction may be found in \cite{DBLP:journals/corr/abs-1804-09999}.

Finally, our framework could also be used for a task that is not directly related to Complex Event Forecasting.
Since predictive modeling  and compression are two sides of the same coin,
our framework could be used for pattern-driven lossless stream compression, 
in order to minimize the communication cost,
which is a severe bottleneck for geo-distributed CER \cite{DBLP:journals/vldb/GiatrakosAADG20}.
The probabilistic model that we construct with our approach could be pushed down to the event sources,
such as the vessels in the maritime domain, 
in order to compress each individual stream and then these compressed streams could be transmitted to a centralized CER engine to perform recognition.



\bibliography{refs}

\appendix

\section{Appendix}
\label{sec:appendix}

\subsection{Complete proof of Proposition \ref{proposition:sre2sfa}}
\label{sec:proof:sre2sfa}

\begin{proposition*}
For every symbolic regular expression $R$ there exists a symbolic finite automaton $M$ such that $\mathcal{L}(R) = \mathcal{L}(M)$.
\end{proposition*}
\begin{proof}
Except for the first case, for the other three cases the induction hypothesis is that the theorem holds for the sub-expressions of the initial expression.

\begin{figure}
\begin{subfigure}[t]{0.5\textwidth}
	\includegraphics[width=0.95\textwidth]{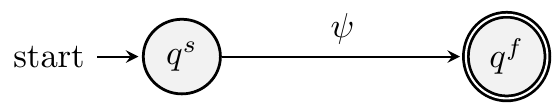}
	\caption{Base case of a single predicate. $R = \psi$.}
	\label{fig:sre2sfa:base}
\end{subfigure}
\begin{subfigure}[t]{0.5\textwidth}
	\includegraphics[width=0.95\textwidth]{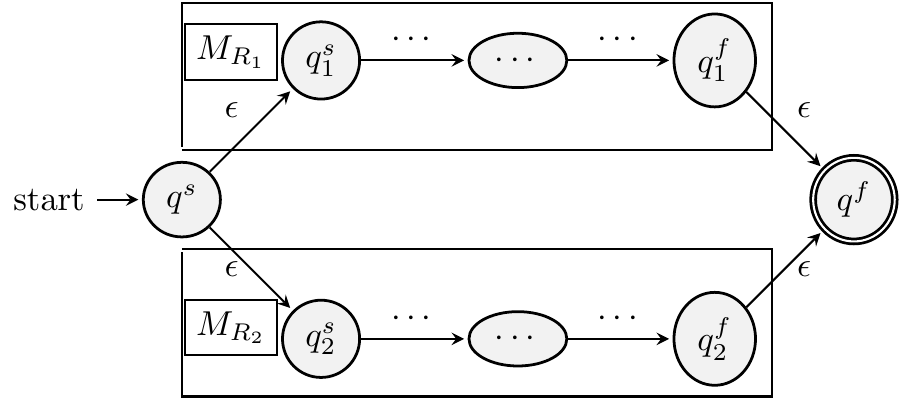}
	\caption{OR. $R = R_{1} + R_{2}$.}
	\label{fig:sre2sfa:or}
\end{subfigure}\\
\begin{subfigure}[t]{0.5\textwidth}
	\includegraphics[width=0.95\textwidth]{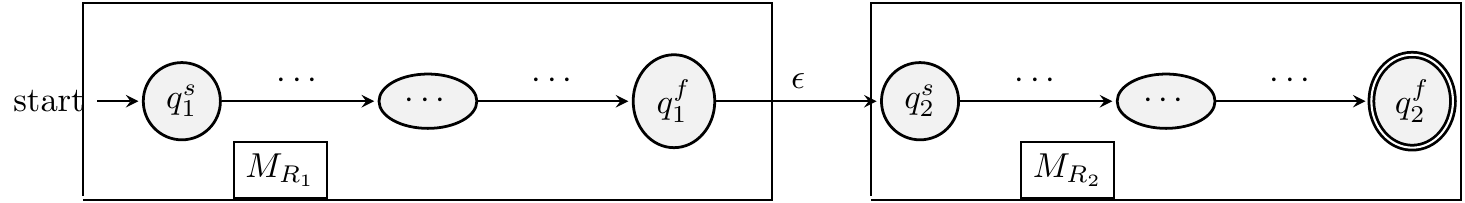}
	\caption{Concatenation. $R = R_{1} \cdot R_{2}$.}
	\label{fig:sre2sfa:seq}
\end{subfigure}
\begin{subfigure}[t]{0.5\textwidth}
	\includegraphics[width=0.95\textwidth]{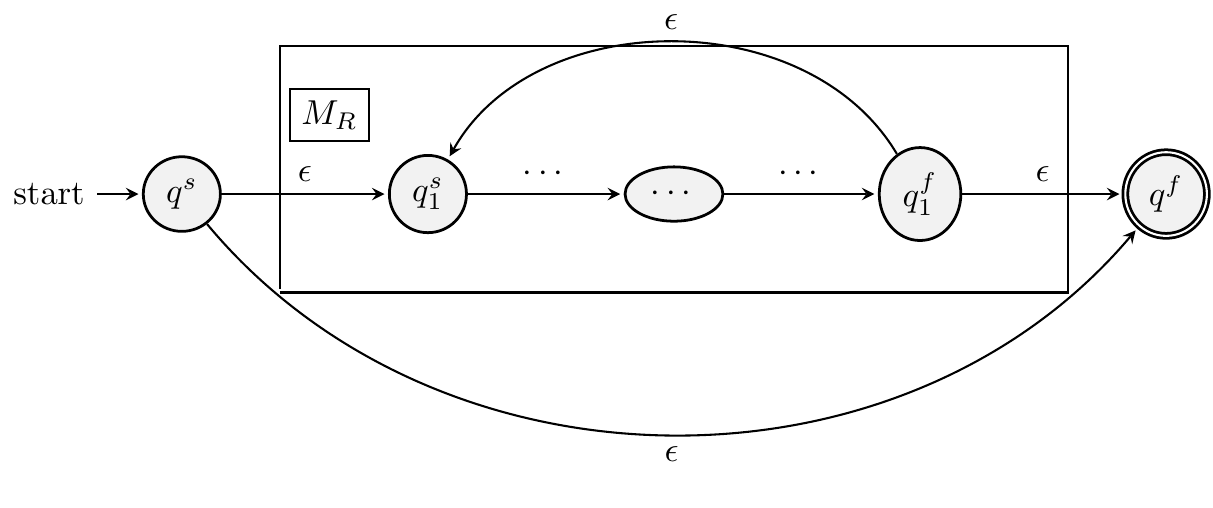}
	\caption{Iteration. $R^{'} = R^{*}$.}
	\label{fig:sre2sfa:iter}
\end{subfigure}
\caption{The four cases for constructing a SFA from a SRE}
\label{fig:sre2sfa}
\end{figure}

\textbf{Case where $R = \psi, \psi \in \Psi$.}\\
We construct a SFA as in Figure \ref{fig:sre2sfa:base}.
If $w \in \mathcal{L}(R)$, then $w$ is a single character and $w \in \llbracket \psi \rrbracket$, i.e., $\psi$ evaluates to \true\ for $w$.
Thus, upon seeing $w$, the SFA of Figure \ref{fig:sre2sfa:base} moves from $q^{s}$ to $q^{f}$ and since $q^{f}$ is a final state, then $w$ is accepted by this SFA. 
Conversely, if a string $w$ is accepted by this SFA then it must again be a single character and $\psi$ must evaluate to \true\ since the SFA moved to its final state through the transition equipped with $\psi$. Thus, $w \in \llbracket \psi \rrbracket$ and $w \in \mathcal{L}(R)$.

\textbf{Case where $R = R_{1} + R_{2}$.}\\
We construct a SFA as in Figure \ref{fig:sre2sfa:or}.
If $w \in \mathcal{L}(R)$, then either $w \in \mathcal{L}(R_{1})$ or $w \in \mathcal{L}(R_{2})$ (or both).
Without loss of generality, assume $w \in \mathcal{L}(R_{1})$. 
From the induction hypothesis, it also holds that $w \in \mathcal{L}(M_{R_{1}})$.
Thus, from Figure \ref{fig:sre2sfa:or}, upon reading $w$, $M_{R_{1}}$ will have reached $q_{1}^{f}$.
Therefore, $M_{R}$ will have reached $q^{f}$ throught the $\epsilon$-transition connecting $q_{1}^{f}$ to $q^{f}$ and thus $w$ is accepted by $M_{R}$.
Conversely, if $w \in \mathcal{L}(M_{R})$, then the SFA $M_{R}$ of Figure \ref{fig:sre2sfa:or} must have reached $q^{f}$ and therefore also $q_{1}^{f}$ or $q_{2}^{f}$ (or both). 
Assume it has reached $q_{1}^{f}$. 
Then $w \in \mathcal{L}(M_{R_{1}})$ and, from the induction hypothesis $w \in \mathcal{L}(R_{1})$.
Similarly, if its has reached $q_{2}^{f}$, then $w \in \mathcal{L}(R_{2})$. 
Therefore, $w \in \mathcal{L}(R_{1}) \cup \mathcal{L}(R_{2}) = \mathcal{L}(R)$.

\textbf{Case where $R = R_{1} \cdot R_{2}$.}\\
We construct a SFA as in Figure \ref{fig:sre2sfa:seq}.
If $w \in \mathcal{L}(R)$, then $w \in \mathcal{L}(R_{1}) \cdot \mathcal{L}(R_{2})$ or $w=w_{1} \cdot w_{2}$ such that $w_{1} \in \mathcal{L}(R_{1})$ and $w_{2} \in \mathcal{L}(R_{2})$.
Therefore, from the induction hypothesis, upon reading $w_{1}$, $M_{R}$ will have reached $q_{1}^{f}$ and $q_{2}^{s}$.
Upon reading the rest of $w$ ($w_{2}$), again from the induction hypothesis, $M_{R}$ will have reached $q_{2}^{f}$.
As a result, $w \in \mathcal{L}(M_{R})$.
Conversely, if $w \in \mathcal{L}(M_{R})$, $M_{R}$ will have reached $q_{2}^{f}$ upon reading $w$ and therefore will have also passed through $q_{1}^{f}$ upon reading a prefix $w_{1}$ of $w$. 
Thus, $w = w_{1} \cdot w_{2}$ with $w_{1} \in \mathcal{L}(M_{R_{1}})$ and $w_{2} \in \mathcal{L}(M_{R_{2}})$.
From the induction hypothesis, it also holds that $w_{1} \in \mathcal{L}(R_{1})$ and $w_{2} \in \mathcal{L}(R_{2})$ and therefore that $w \in \mathcal{L}(R)$.

\textbf{Case where $R^{'} = R^{*}$.}\\
We construct a SFA as in Figure \ref{fig:sre2sfa:iter}.
If $w \in \mathcal{L}(R^{'})$, then $w \in (\mathcal{L}(R))^{*}$ or, equivalently, $w=w_{1} \cdot w_{2} \cdot \cdots \cdot w_{k}$ such that $w_{i} \in \mathcal{L}(R)$ for all $w_{i}$.
From the induction hypothesis and Figure \ref{fig:sre2sfa:iter}, upon reading $w_{1}$, $M_{R^{'}}$ will have reached $q_{1}^{f}$ and $q_{1}^{s}$.
Therefore, the same will be true after reading $w_{2}$ and all other $w_{i}$, including $w_{k}$.
Thus, $w \in \mathcal{L}(M_{R^{'}})$.
Note that if $w=\epsilon$, the $\epsilon$-transition from $q^{s}$ to $q^{f}$ ensures that $w \in \mathcal{L}(M_{R^{'}})$.
Conversely, assume $w \in \mathcal{L}(M_{R^{'}})$.
If $w=\epsilon$, then by the definition of the $^{*}$ operator, $w \in (\mathcal{L(R)})^{*}$.
In every other case, $M_{R^{'}}$ must have reached $q_{1}^{f}$ and must have passed through $q_{1}^{s}$.
Therefore, $w$ may be written as $w=w_{1} \cdot w_{2}$ where $w_{2} \in \mathcal{M_{R}}$
and, for $w_{1}$, upon reading it, $M_{R^{'}}$ must have reached $q_{1}^{s}$. 
There are two cases then: either $w_{1}=\epsilon$ and $q_{1}^{s}$ was reached from $q^{s}$ or $w_{1} \neq \epsilon$ and $q_{1}^{s}$ was reached from $q_{1}^{f}$.
In the former case, $w = \epsilon \cdot w_{2} = w_{2}$ and thus $w \in (\mathcal{L}(R))^{*}$.
In the latter case, we can apply a similar reasoning recursively to $w_{1}$ in order to split it to sub-strings $w_{i}$ such that $w_{i} \in \mathcal{L}(R)$.
Therefore, $w \in (\mathcal{L}(R))^{*}$ and $w \in \mathcal{L}(R^{'})$.  

\end{proof}

\subsection{Proof of Proposition \ref{proposition:streamingsre}}
\label{sec:proof:streamingsre}

\begin{proposition*}
If $S=t_{1},t_{2},\cdots$ is a stream of domain elements from an effective Boolean algebra $\mathcal{A} = (\mathcal{D}$, $\Psi$, $\llbracket \_ \rrbracket$, $\bot$, $\top$, $\vee$, $\wedge$, $\neg$), where $t_{i} \in \mathcal{D}$, and $R$ is a symbolic regular expression over the same algebra,
then, for every $S_{m..k}$, $S_{m..k} \in \mathcal{L}(R)$ iff $S_{1..k} \in \mathcal{L}(R_{s})$ (and $S_{1..k} \in \mathcal{L}(M_{R_{s}})$).
\end{proposition*}
\begin{proof}
First, assume that $S_{m..k} \in \mathcal{L}(R)$ for some $m, 1 \leq m \leq k$
(we set $S_{1..0} = \epsilon$). 
Then, for $S_{1..k} = S_{1..(m-1)} \cdot S_{m..k}$, $S_{1..(m-1)} \in \mathcal{L}(\top^{*})$, 
since $\top^{*}$ accepts every string (sub-stream),
including $\epsilon$. 
We know that $S_{m..k} \in \mathcal{L}(R)$, thus $S_{1..k} \in \mathcal{L}(\top^{*}) \cdot \mathcal{L}(R) = \mathcal{L}(\top^{*} \cdot R) = \mathcal{L}(R_{s})$.
Conversely, assume that $S_{1..k} \in \mathcal{L}(R_{s})$.
Therefore, $S_{1..k} \in \mathcal{L}(\top^{*} \cdot R) = \mathcal{L}(\top^{*}) \cdot \mathcal{L}(R)$.
As a result, $S_{1..k}$ may be split as $S_{1..k} = S_{1..(m-1)} \cdot S_{m..k}$ such that $S_{1..(m-1)} \in \mathcal{L}(\top^{*})$ and $S_{m..k} \in \mathcal{L}(R)$. 
Note that $S_{1..(m-1)} = \epsilon$ is also possible, in which case the result still holds, since $\epsilon \in \mathcal{L}(\top^{*})$.
\end{proof}

\subsection{Proof of Theorem \ref{theorem:finals}}
\label{sec:proof:finals}
\begin{theorem*}
Let $\boldsymbol{\Pi}$ be the transition probability matrix of a homogeneous Markov chain $Y_{t}$ in the form of Equation \eqref{eq:matrix}
and $\boldsymbol{\xi}_{init}$ its initial state distribution.
The probability for the time index $n$ when the system first enters the set of states $F$,
starting from a state in $F$, 
can be obtained from
\begin{equation*}
P(Y_{n} \in F, Y_{n-1} \notin F,\cdots,Y_{1} \in F \mid \boldsymbol{\xi_{init}}) =
  \begin{cases}
    \boldsymbol{\xi_{F}}^{T} \boldsymbol{F}  \boldsymbol{1} & \quad \text{if } n=2   \\
    \boldsymbol{\xi_{F}}^{T}  \boldsymbol{F_{N}} \boldsymbol{N}^{n-2}(\boldsymbol{I}-\boldsymbol{N})\boldsymbol{1} & \quad \text{otherwise} \\
  \end{cases}
\end{equation*}
where $\xi_{F}$ is the vector consisting of the elements of $\xi_{init}$ corresponding to the states of $F$.
\end{theorem*}
\begin{proof}

\textbf{Case where $n=2$.}\\
In this case,
we are in a state $i \in F$ and we take a transition that leads us back to $F$ again.
Therefore, $P(Y_{2} \in F, Y_{1}=i \in F \mid \boldsymbol{\xi_{init}}) = \boldsymbol{\xi}(i) \sum_{j \in F} \pi_{ij}$,
i.e.,
we first take the probability of starting in $i$ and multiply it by the sum of all transitions from $i$ that lead us back to $F$.
This result folds for a certain state $i \in F$.
If we start in any state of $F$, 
$P(Y_{2} \in F, Y_{1} \in F \mid \boldsymbol{\xi_{init}}) = \sum_{i \in F} \boldsymbol{\xi}(i) \sum_{j \in F} \pi_{ij}$.
In matrix notation, this is equivalent to 
$P(Y_{2} \in F, Y_{1} \in F \mid \boldsymbol{\xi_{init}}) = \boldsymbol{\xi_{F}}^{T} \boldsymbol{F}  \boldsymbol{1}$.

\textbf{Case where $n>2$.}\\
In this case,
we must necessarily first take a transition from $i \in F$ to $j \in N$,
then, for multiple transitions we remain in $N$ and we finally take a last transition from $N$ to $F$.
We can write
\begin{equation}
\label{eq:broken}
\begin{aligned}
P(Y_{n} \in F, Y_{n-1} \notin F,...,Y_{1} \in F \mid \boldsymbol{\xi_{init}} ) = &  P(Y_{n} \in F, Y_{n-1} \notin F,...,Y_{2} \notin F \mid \boldsymbol{\xi^{'}_{N}} ) \\
= & P(Y_{n-1} \in F, Y_{n-2} \notin F,...,Y_{1} \notin F \mid \boldsymbol{\xi^{'}_{N}} )
\end{aligned}
\end{equation}
where $\boldsymbol{\xi^{'}_{N}}$ is the state distribution (on states of $N$) after having taken the first transition from $F$ to $N$.
This is given by $\boldsymbol{\xi^{'}_{N}} = \boldsymbol{\xi_{F}}^{T} \boldsymbol{F_{N}}$.
By using this as an initial state distribution in Eq. \ref{eq:wtd:non-finals} and running the index $n$ from $1$ to $n-1$,
as in Eq. \ref{eq:broken},
we get
\begin{equation*}
P(Y_{n} \in F, Y_{n-1} \notin F,...,Y_{1} \in F \mid \boldsymbol{\xi_{init}}) = \boldsymbol{\xi_{F}}^{T}  \boldsymbol{F_{N}} \boldsymbol{N}^{n-2}(\boldsymbol{I}-\boldsymbol{N})\boldsymbol{1}
\end{equation*}
\end{proof}

\subsection{Proof of Correctness for Algorithm \ref{algorithm:interval}}
\label{sec:proof:interval}
\begin{proposition*}
Let $P$ be a waiting-time distribution with horizon $h$ and let $\theta_{fc} < 1.0$ be a confidence threshold. 
Algorithm \ref{algorithm:interval} correctly finds the smallest interval whose probability exceeds $\theta_{fc}$.
\end{proposition*}
\begin{proof}[Proof of correctness for Algorithm \ref{algorithm:interval}]

\begin{figure}[t]
	\centering
	\includegraphics[width=0.65\textwidth]{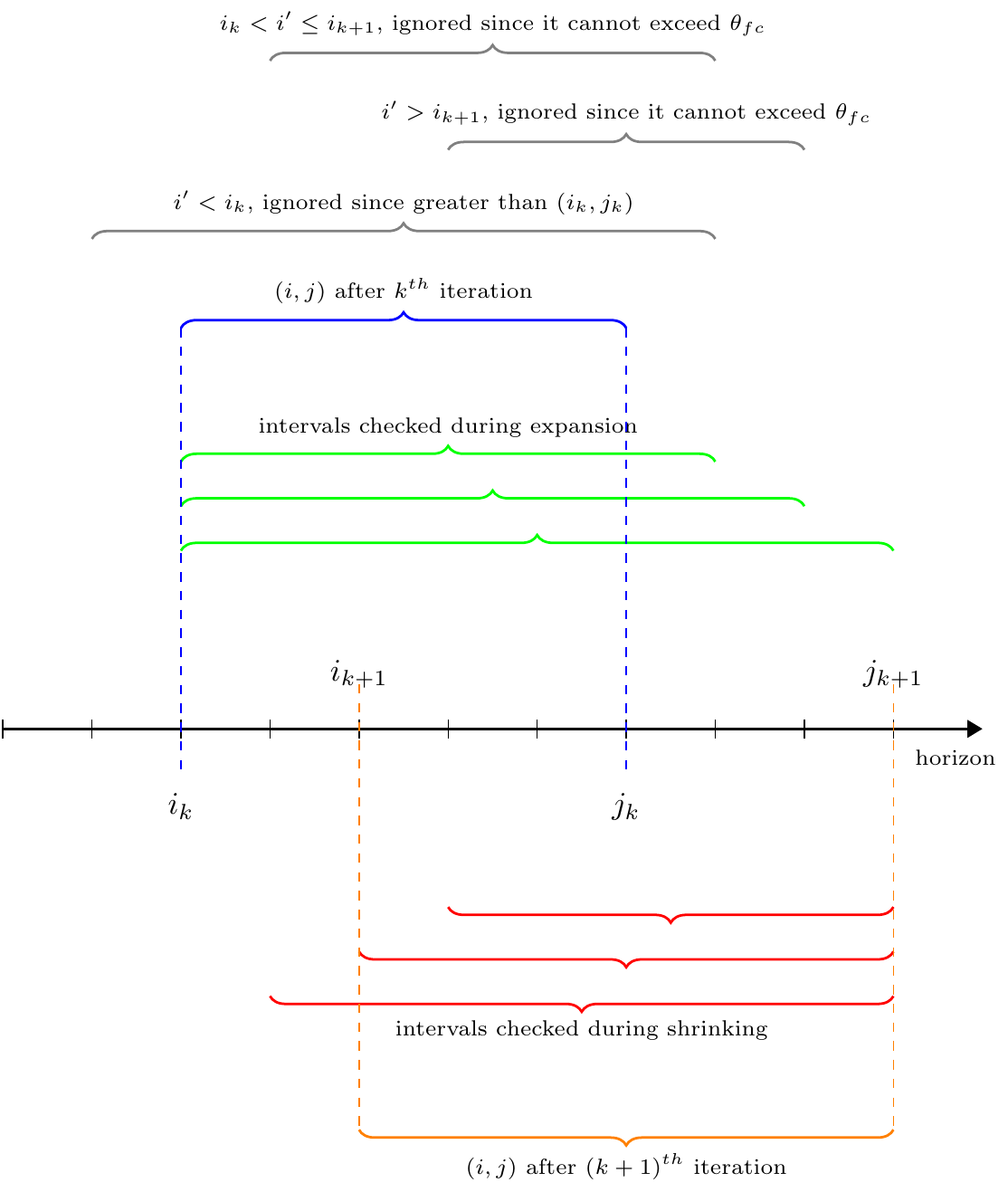}
	\caption{One iteration of Algorithm \ref{algorithm:interval}. After the $k^{th}$ iteration, $(i,j)$ is either the smallest interval or there already exists a smaller one. In the $(k+1)^{th}$ iteration, we need to check intervals for which $j_{k} < j' \leq j_{k+1}$, .i.e., whose right limit is greater than $j_{k}$ and smaller (but including) $j_{k+1}$. Out of those intervals, the ones that start before $i_{k}$ need not be checked because they are greater than $(i_{k},j_{k})$. The ones that start after $i_{k+1}$ are also ignored, since we already know, from the shrinking phase, that their super-interval $(i_{k+1}+1,j_{k+1})$ does not exceed $\theta_{fc}$. Finally, we can also ignore the intervals that start between $i_{k}$ and $i_{k+1}$ (not including $i_{k}$) and end before $j_{k+1}$ (not including $j_{k+1}$), because they cannot exceed $\theta_{fc}$. If such an interval actually existed, then the expansion phase would have stopped before $j_{k+1}$.}
	\label{fig:intervals}
\end{figure}

We only need to prove the loop invariant.
Assume that after the $k^{th}$ iteration of the outer while loop $i=i_{k}$ and $j=j_{k}$ and that after the $(k+1)^{th}$ iteration $i=i_{k+1}$ and $j=j_{k+1}$.
If the invariant holds after the $k_{th}$ iteration,
then all intervals with $e \leq j_{k}$ have been checked and we know that $(s,e)$ is the best interval up to $j_{k}$.
It can be shown that,
during the $(k+1)^{th}$ iteration,
the intervals up to $j_{k+1}$ that are not explicitly checked are intervals which cannot possibly exceed $\theta_{fc}$ or cannot be better than the currently held best interval $(s,e)$.
There are three such sets of unchecked intervals (see also Figure \ref{fig:intervals}):
\begin{itemize}
	\item All intervals $(i',j')$ such that $i' < i_{k}$ and $j_{k} \leq j' \leq j_{k+1}$, 
	i.e., we ignore all intervals that start before $i_{k}$.
	Even if these intervals exceed $\theta_{fc}$, they cannot possibly be smaller than $(s,e)$, since we know that $(s,e)=(i_{k},j_{k})$ or that $(s,e)$ is even smaller than $(i_{k},j_{k})$.
	\item All intervals $(i',j')$ such that $i' > i_{k+1}$ and $j_{k} \leq j' \leq j_{k+1}$,
	i.e., we ignore all intervals that start after $i_{k+1}$.
	These intervals cannot possibly exceed $\theta_{fc}$, since $(i_{k+1}+1,j_{k+1})$ is below $\theta_{fc}$ and all these intervals are sub-intervals of $(i_{k+1}+1,j_{k+1})$.
	\item We are thus left with intervals $(i',j')$ such that $i_{k} \leq i' \leq i_{k+1}$ and $j_{k} \leq j' \leq j_{k+1}$. 
	Out of all the interval that can be constructed from combining these $i'$ and $j'$,
	the algorithm checks the intervals $(i'=i_{k},j')$ and $(i',j'=j_{k+1})$.
	The intervals that are thus left unchecked are the intervals $(i', j')$ such that $i_{k} < i' \leq i_{k+1}$ and $j_{k} \leq j' < j_{k+1}$. The question is: is it possible for such an interval to exceed $\theta_{fc}$. The answer is negative. Assume that there is such an interval $(i',j')$. 
	If this were the case, then the algorithm, during its expansion phase, would have stopped at $j'$,
	because $(i_{k},j')$ would exceed $\theta_{fc}$. Therefore, these intervals cannot exceed $\theta_{fc}$.
\end{itemize}
A similar reasoning allows us to show that the loop invariant holds after the first iteration.
It thus holds after every iteration.
\end{proof}

\subsection{Proofs of Complexity Results}

\subsubsection{Proof of Proposition \ref{proposition:complexity1}}
\label{sec:proof:complexity1}
\begin{proposition*}[Step 1 in Figure \ref{fig:vmmflow}]
Let $S_{1..k}$ be a stream and $m$ the maximum depth of the Counter Suffix Tree $T$ to be constructed from $S_{1..k}$.
The complexity of constructing $T$ is $O(m(k-m))$.
\end{proposition*}
\begin{proof}
There are three operations that affect the cost:
incrementing the counter of a node by $1$, with constant cost $i$;
inserting a new node, with constant cost $n$;
visiting an existing node with constant cost $v$;
We assume that $n > v$.
For every $S_{l-m+1..l}$, $m \leq l \leq k$ of length $m$,
there will be $m$ increment operations and $m$ nodes will be ``touched'',
i.e.,
either visited if already existing or created. 
Therefore, the total number of increment operations is $(k-m+1) m = km - m^{2} +m = m(k-m)+m$.
The same result applies for the number of node ``touches''.
It is always true that $m < k$ and typically $m \ll k$.
Therefore, the cost of increments is $O(m(k-m))$ and the cost of visits/creations is also $O(m(k-m))$.
Thus, the total cost is $O(m(k-m)) + O(m(k-m)) = O(m(k-m))$.
In fact, 
the worst case is when all $S_{l-m+1..l}$ are different and have no common suffixes.
In this case,
there are no visits to existing nodes,
but only insertions, 
which are more expensive than visits.
Their cost would again be $O(nm(k-m))=O(m(k-m))$, ignoring the constant $n$.
\end{proof}

\subsubsection{Proof of Proposition \ref{proposition:complexity3}}
\label{sec:proof:complexity3}
\begin{proposition*}[Step 3a in Figure \ref{fig:vmmflow}]
Let $T$ be a \pst\ of maximum depth $m$, learned with the $t$ minterms of a \dsfa\ $M_{R}$.
The complexity of constructing a \psa\ $M_{S}$ from $T$ is $O(t^{m+1} \cdot m)$.
\end{proposition*}
\begin{proof}
We assume that the cost of creating new states and transitions for $M_{S}$ is constant.
In the worst case,
all possible suffixes of length $m$ have to be added to $T$ as leaves.
$T$ will thus have $t^{m}$ leaves.
The main idea of the algorithm for converting a \pst\ $T$ to a \psa\ $M_{S}$ is to use the leaves
of $T$ as states of $M_{S}$ and for every symbol (minterm) $\sigma$ find the next state/leaf and set the transition probability to be equal to the probability of $\sigma$ from the source leaf. 
If we assume that the cost of accessing a leaf is constant
(e.g., by keeping separate pointers to the leaves),
the cost for constructing $M_{S}$ is dominated by the cost of constructing the $k^{m}$ states of $M_{S}$ and the $t$ transitions from each such state. 
For each transition,
finding the next state requires traversing a path of length $m$ in $T$. 
The total cost is thus $O(t^{m} \cdot t \cdot m)$ = $O(t^{m+1} \cdot m)$.  
\end{proof}

\subsubsection{Proof of Proposition \ref{proposition:complexity4}}
\label{sec:proof:complexity4}
\begin{proposition*}[Step 4 in Figure \ref{fig:vmmflow}]
Let $M_{R}$ be a \dsfa\ with $t$ minterms and $M_{S}$ a \psa\ learned with the minterms of $M_{R}$.
The complexity of constructing an embedding $M$ of $M_{S}$ in $M_{S}$ with Algorithm \ref{algorithm:merging} is $O(t \cdot \lvert M_{R}.Q \times M_{S}.Q\rvert)$.
\end{proposition*}
\begin{proof}
We assume that the cost of constructing new states and transitions for $M$ is constant.
We also assume that the cost of finding a given state in both $M_{R}$ and $M_{S}$ is constant,
e.g., 
by using a linked data structure for representing the automaton with a hash table on its states (or an array),
and the cost of finding the next state from a given state is also constant.
In the worst case,
even with the incremental algorithm \ref{algorithm:merging},
we would need to create the full Cartesian product $M_{R}.Q \times M_{S}.Q$ to get the states of $M$.
For each of these states,
we would need to find the states of $M_{R}$ and $M_{S}$ from which it will be composed and to create $t$ outgoing transitions.
Therefore,
the complexity of creating $M$ would be
$O(t \cdot \lvert M_{R}.Q \times M_{S}.Q\rvert)$.
\end{proof}

\subsubsection{Proof of Proposition \ref{proposition:complexity5}}
\label{sec:proof:complexity5}
\begin{proposition*}[Step 5 in Figure \ref{fig:vmmflow}]
Let $M$ be the embedding of a \psa\ $M_{S}$ in a \dsfa\ $M_{R}$.
The complexity of estimating the waiting-time distribution for a state of $M$ and a horizon of length $h$ using Theorem \ref{theorem:non-finals} is $O((h-1) k^{2.37})$ where $k$ is the dimension of the square matrix $\boldsymbol{N}$.
\end{proposition*}
\begin{proof}
We want to use Equation \ref{eq:wtd:non-finals} to estimate the distribution of a state.
The equation is repeated below:
\begin{equation*}
P(Y_{n} \in F, Y_{n-1} \notin F,...,Y_{1} \notin F \mid \boldsymbol{\xi_{init}}) =
\boldsymbol{\xi_{N}}^{T}\boldsymbol{N}^{n-1}(\boldsymbol{I}-\boldsymbol{N})\boldsymbol{1}
\end{equation*}
We want to estimate the distribution for the $h$ points of the horizon,
i.e.,
for $n=2$, $n=3$ up to $n=h+1$.
For $n=2$,
we have 
\begin{equation*}
P(Y_{2} \in F, Y_{1} \notin F \mid \boldsymbol{\xi_{init}}) =
\boldsymbol{\xi_{N}}^{T}\boldsymbol{N}(\boldsymbol{I}-\boldsymbol{N})\boldsymbol{1}
\end{equation*}
For $n=3$,
we have
\begin{equation*}
P(Y_{3} \in F, Y_{2} \notin F, Y_{1} \notin F \mid \boldsymbol{\xi_{init}}) =
\boldsymbol{\xi_{N}}^{T}\boldsymbol{N}^{2}(\boldsymbol{I}-\boldsymbol{N})\boldsymbol{1}
\end{equation*}
In general,
for $n=i$,
we can use the power of $\boldsymbol{N}$ that we have estimated in the previous step for $n=i-1$, $\boldsymbol{N}^{i-2}$, in order to estimate the next power $\boldsymbol{N}^{i-1}$ via a multiplication by $\boldsymbol{N}$ so as to avoid estimating this power from scratch.
Then $\boldsymbol{N}^{i-1}$ can be multiplied by $(\boldsymbol{I}-\boldsymbol{N})\boldsymbol{1}$,
which remains fixed for all $i$ and can thus be calculated only once.

The cost of estimating $(\boldsymbol{I}-\boldsymbol{N})$ is $k^{2}$ due to the $k^{2}$ subtractions.
Multiplying the matrix $(\boldsymbol{I}-\boldsymbol{N})$ by the vector $\boldsymbol{1}$ results in a new vector with $k$ elements. 
Each of these elements requires $k$ multiplications and $k - 1$ additions or $2k - 1$ operations.
Thus, the estimation of $(\boldsymbol{I}-\boldsymbol{N})\boldsymbol{1}$ has a cost of $k^{2} + k(2k - 1) = 3 k^{2} - k$.

Now, for $n=i$, estimating the power $\boldsymbol{N}^{i-1}$ from $\boldsymbol{N}^{i-2}$ has a cost of $k^{2.37}$ using an efficient multiplication algorithm such as the Coppersmith–Winograd algorithm \cite{DBLP:journals/jsc/CoppersmithW90} or the improvement proposed by Stothers \cite{stothers2010complexity}.

Additionally, $\boldsymbol{N}^{i-1}$ must then be multiplied by the vector $(\boldsymbol{I}-\boldsymbol{N})\boldsymbol{1}$,
with a cost of $2k^{2}-k$,
resulting in a new vector with $k$ elements.

This vector must then be multiplied by $\boldsymbol{\xi_{N}}^{T}$ to produce the final probability value with a cost of $2k-1$ for the $k$ multiplications and the $k-1$ additions.

We thus have a fixed initial cost of $3k^{2} - k$ and then for every iteration $i$ a cost of 
$k^{2.37} I_{\{i>2\}} + 2k^{2}-k + 2k-1 = k^{2.37} I_{\{i>2\}} + 2k^{2} + k - 1$,
where $I_{\{i>2\}}$ is an indicator function ($1$ for $i>2$ and $0$ otherwise).
Note that the cost $k^{2.37}$ is not included for $i=2$ because in this case we do not need to raise $\boldsymbol{N}$ to a power.
The total cost would thus be:
\begin{equation*}
\begin{aligned}
3k^{2} - k  + & k^{2.37} \cdot 0 + 2k^{2} + k - 1 & \text{for } n=2 \\
 + & k^{2.37} \cdot 1 + 2k^{2} + k - 1 & \text{for } n=3 \\
 & \cdots & \\
   + & k^{2.37} \cdot 1 + 2k^{2} + k - 1 & \text{for } n=h+1 \\
   = & (h-1)k^{2.37} + (2h+3)k^{2} + (h-1)k -h & = O((h-1)k^{2.37})
\end{aligned}
\end{equation*}

\end{proof}

\subsubsection{Proof of Proposition \ref{proposition:complexity6}}
\label{sec:proof:complexity6}
\begin{proposition*}[Step 6 in Figure \ref{fig:vmmflow}]
For a waiting-time distribution with a horizon of length $h$,
the complexity of finding the smallest interval that exceeds a confidence threshold $\theta_{fc}$ with Algorithm \ref{algorithm:interval} is $O(h)$. 
\end{proposition*}
\begin{proof}
Indexes $i$ and $j$ of Algorithm \ref{algorithm:interval} scan the distribution only once.
The cost for $j$ is the cost of $h$ points of the distribution that need to be accessed plus 
$h-1$ additions.
Similarly, the cost for $i$ is the cost of (at most) $h$ accessed points plus the cost of (at most) $h-1$ subtractions.
Thus the total cost is $O(h)$.
\end{proof}

\subsubsection{Proof of Proposition \ref{proposition:complexity3prime}}
\label{sec:proof:complexity3prime}
\begin{proposition*}[Step 3b in Figure \ref{fig:vmmflow}]
Let $T$ be a \pst\ of maximum depth $m$, learned with the $t$ minterms of a \dsfa\ $M_{R}$.
The complexity of estimating the waiting-time distribution for a state of $M_{R}$ and a horizon of length $h$ directly from $T$ is $O((m+3) \frac{t - t^{h+1}}{1 - t})$.
\end{proposition*}
\begin{proof}
After every new event arrival,
we first have to construct the tree of future states,
as shown in Figure \ref{fig:future}.
In the worst case,
no paths can be pruned and the tree has to be expanded until level $h$.
The total number of nodes that have to be created is thus a geometric progress:
$t + t^{2} + \cdots + t^{h}=\sum_{i=1}^{h}t^{i}=\frac{t - t^{h+1}}{1-t}$.
Assuming that it takes constant time to create a new node,
this formula gives the cost of creating the nodes of the trees.
Another cost that is involved concerns the time required to find the proper leaf of the \pst\ $T$ before the creation of each new node.
In the worst case,
all leaves will be at level $m$.
The cost of each search will thus be $m$.
The total search cost for all nodes will be
$mt + mt^{2} + \cdots + mt^{h}=\sum_{i=1}^{h}mt^{i}=m\frac{t - t^{h+1}}{1-t}$.
The total cost (node creation and search) for constructing the tree is
$\frac{t - t^{h+1}}{1-t} + m\frac{t - t^{h+1}}{1-t} = (m+1)\frac{t - t^{h+1}}{1-t}$.
With the tree of future states at hand,
we can now estimate the waiting-time distribution.
In the worst case,
the tree will be fully expanded and we will have to access all its paths until level $h$.
We will first have to visit the $t$ nodes of level 1,
then the $t^{2}$ nodes of level 2, etc.
The access cost will thus be  
$t + t^{2} + \cdots + t^{h}=\sum_{i=1}^{h}t^{i}=\frac{t - t^{h+1}}{1-t}$.
We also need to take into account the cost of estimating the probability of each node.
For each node,
one multiplication is required,
assuming that we store partial products and do not have to traverse the whole path to a node to estimate its probability.
As a result,
the number of multiplications will also be $\frac{t - t^{h+1}}{1-t}$.
The total cost (path traversal and multiplications) will thus be $2\frac{t - t^{h+1}}{1-t}$,
where we ignore the cost of summing the probabilities of final states,
assuming it is constant.
By adding the cost of constructing the tree ($(m+1)\frac{t - t^{h+1}}{1-t}$) and the cost of estimating the distribution ($2\frac{t - t^{h+1}}{1-t}$),
we get a complexity of $O((m+3)\frac{t - t^{h+1}}{1-t})$.
\end{proof}

\end{document}